\documentclass[11pt]{report}
\usepackage{upennstyle} 

\usepackage[square,numbers]{natbib}
\title{Testing by Dualization}
\author{Yishuai Li}
\gradgroup{Computer and Information Science} 
\date{2022} 
\supervisor{Benjamin C. Pierce}
\supervisortitle{Professor of Computer and Information Science}
\gradchair{Mayur Naik, Professor of Computer and Information Science}
\committee{Steve Zdancewic, Professor of Computer and Information Science}
\committee{Boon Thau Loo, Professor of Computer and Information Science}
\committee{Rajeev Alur, Professor of Computer and Information Science}
\committee{Leonidas Lampropoulos, Assistant Professor of Computer Science, University of Maryland} 

\authorlegal{Yishuai Li} 
\cclicense{Attribution-Share\-Alike 4.0 International (CC BY-SA 4.0)} 
\cclicenseurl{https://creativecommons.org/licenses/by-sa/4.0/} 

\acknowledgement{I would express my sincere gratitude to my advisors, Benjamin
  C. Pierce and Steve Zdancewic.  Every visit to your office made me a better
  computer scientist.  I hope to become inspiring to others in my future career,
  just like you inspired me throughout this journey.

  I would also thank my proposal committee members, Boon Thau Loo, Mayur Naik,
  and John Hughes, for paving the path toward this dissertation.  Thank you to
  Rajeev Alur and Leonidas Lampropoulos for making my defense possible.  And to
  Stephanie Weirich for hosting the board game nights and making my stay at Penn
  full of joy.

  I was fortunate to be accompanied by my wonderful colleagues at Penn.  Thank
  you to Li-yao Xia for navigating me through the interaction trees and showing
  me how to exploit this language.  Thank you to Nicolas Koh, Yao Li, and
  Hengchu Zhang for collaborating on the DeepSpec project and crafting the web
  server.  To Pritam Choudhury for the discussions on philosophy.  To Caleb
  Stanford and Omar Navarro Leija for organizing the CIS department events.  And
  to the entire PL Club for all the fun we had together.

  My life in Philly was colorful thanks to my friends here: Qizhen, Yi, Haoxian,
  Wenxin, Lei, Yinjun, Chenying, Hangfeng, Danni, Xujie, Ningning, Jiani,
  Ziyang, Hui, Leshang, Wei, Yang, Qiuyue, Tianlin, Yachong, Wanying, and
  Jiehao.  Thank you to Dr. Ketner for leading the Penn Wind Ensemble, where I
  enjoyed making great music.  To my friends from USTC: Wenjun, Wenshuo,
  Zhongshi, Jiahui, Han, Jingyi, Xiufan, Xianfeng, Zhishuai, Ye, Dinglong,
  Jingru, Shang, Huazheng, Yuepeng, Yannan, Renqian, Yushan, and Xiaojing, who
  shared their findings with me through the expedition towards science.  And to
  my friends back home.  Mengting, Changhan, Rui, Ruilin, Xiaoyu, Letian,
  Hongxuan, Zichao, and Di: thank you for being with me throughout these years.

  Last but not least, I want to thank my parents, Cui Haijing and Li Mingzhuo,
  and my grandparents, Li Jingzi and Li Dongxi.  Your love and support shaped me
  so far and beyond.}

\abstract{Software engineering requires rigorous testing to guarantee the
  product's quality.  Semantic testing of functional correctness is challenged
  by nondeterminism in behavior, which makes testers difficult to write
  and reason about.

  This thesis presents a language-based technique for testing interactive
  systems.  I propose a theory for specifying and validating
  nondeterministic behaviors, with guaranteed soundness and correctness.  I then
  apply the theory to testing practices, and show how to derive specifications
  into interactive tester programs.  I also introduce a language design for
  producing test inputs that can effectively detect and reproduce
  invalid behaviors.

  I evaluate the methodology by specifying and testing real-world systems such
  as web servers and file synchronizers, demonstrating the derived testers'
  ability to find disagreements between the specification and the
  implementation.}

\usepackage[dvipsnames,svgnames,x11names]{xcolor}
\usepackage{xspace}
\usepackage{advdate}
\usepackage[iso]{datetime}
\newdateformat{httpdate}{\shortdayofweekname{\THEDAY}{\THEMONTH}{\THEYEAR},\ \THEDAY\ \shortmonthname[\THEMONTH]\ \THEYEAR}

\usepackage{draft} 
\newif\ifdraft\drafttrue
\newnote{lys}{BrickRed}
\newnote{bcp}{Chocolate}
\newnote{sz}{violet}
\newnote{ly}{orange}

\theoremstyle{definition}
\newtheorem{definition}{Definition}
\newtheorem{lemma}{Lemma}[chapter]

\addto\extrasamerican{%

}
\usepackage{textcomp,listings,lstparams,lstcoq,enumitem}
\lstdefinestyle{customcoq}{
  language=Coq,
  literate={<>}{$\ne$}1
  {->}{$\to$}1
  {<-}{$\gets$}1
  {=>}{$\Rightarrow$}1
  {|->}{$\mapsto$}1
  {fun}{$\lambda$}1
  {\\in}{$\in$}1
  {\\\/}{$\vee$}1
  {<>}{$\neq$}1
  {>>}{$\gg$}1
  {>>=}{$\gg\!=$}1
  {forall}{$\forall$}1
  {+'}{$\oplus$}1
  {sigma}{$\sigma$}1
}
\lstdefinestyle{json}{
  language=Coq,
}
\lstdefinestyle{customc}{
  language=C,
  keywordstyle={\bfseries\codestyle\color{Keyword1Color}},
  showstringspaces=false,
  morekeywords={true,false,bool},
  escapechar=\$
}
\newcommand{\http}{HTTP/1.1\xspace}
\newcommand{\inlinec}[1]{\lstinline[style=customc]{#1}}
\newcommand{\ilc}[1]{\lstinline[style=customcoq]{#1}}
\newcommand{\ilj}[1]{\lstinline[style=json]{#1}}
\newcommand{\iletag}[1]{{\color{blue}\ilc{#1}}}
\lstnewenvironment{coq}{\lstset{style=customcoq}\linespread{1}}{}
\lstnewenvironment{cpp}{\lstset{style=customc}\linespread{1}}{}
\lstnewenvironment{json}{\lstset{style=json}\linespread{1}}{}
\usepackage{multicol}
\usepackage{mathtools}
\makeatletter
\g@addto@macro\normalsize{%
  \setlength{\abovedisplayskip}{.5em minus .5em}
  \setlength{\belowdisplayskip}{.5em minus .5em}
  \setlength{\abovedisplayshortskip}{0em minus .5em}
  \setlength{\belowdisplayshortskip}{0em minus .5em}
}
\makeatother

\newcommand{\complies}[2]{\mathsf{comply}_{#1}~{#2}}

\newcommand{\valid}[2]{\mathsf{valid}_{#1}~{#2}}
\newcommand{\invalid}[2]{\neg\valid{#1}{#2}}
\newcommand{\accepts}[2]{\mathsf{accept}_{#1}~{#2}}
\newcommand{\rejects}[2]{\neg\accepts{#1}{#2}}
\newcommand{\rejSound}[2]{#1\;\mathsf{sound}_{#2}^\mathsf{Rej}}
\newcommand{\rejComplete}[2]{#1\;\mathsf{complete}_{#2}^\mathsf{Rej}}

\newcommand{\triggers}[3]{#1\overunderset{#3}{#2}\longrightarrow}
\newcommand{\behaves}[2]{\triggers{#1}{}{#2}}
\newcommand{\llet}{\mathsf{let}\;}
\newcommand{\iin}{\mathsf{in}\;}
\newcommand{\letin}[2]{\llet#1=#2\;\iin}
\newcommand{\option}{\mathsf{option}\;}
\newcommand{\Some}[1]{\mathsf{Some}\;#1}
\newcommand{\None}{\mathsf{None}}
\newcommand{\sigT}[2]{\exists#1,#2}
\newcommand{\existT}[3]{\mathsf{pack}\;#1=#2\;\mathsf{with}\;#3}
\newcommand{\sstep}{\mathsf{sstep}}
\newcommand{\vstep}{\mathsf{vstep}}
\newcommand{\Server}{\mathsf{Server}}
\newcommand{\Validator}{\mathsf{Validator}}
\newcommand{\stepServer}{\mathsf{stepServer}}
\newcommand{\stepValidator}{\mathsf{stepValidator}}
\newcommand{\lam}[2]{\lambda#1\Rightarrow#2}
\newcommand{\Set}{\mathsf{set}\;}
\newcommand{\List}{\mathsf{list}\;}
\newcommand{\nil}{\varepsilon}
\newcommand{\If}{\mathsf{if}\;}
\newcommand{\tthen}{\mathsf{then}\;}
\newcommand{\Then}{\;\tthen}
\newcommand{\eelse}{\mathsf{else}\;}
\newcommand{\Else}{\;\eelse}

\newcommand{\Nat}{\mathbb N}
\newcommand{\Int}{\mathbb Z}
\newcommand{\Bool}{\mathbb B}

\newcommand{\Is}{\;\mathsf{is}\;}
\newcommand{\Prog}{\mathsf{Prog}}

\newcommand{\serverOf}{\mathsf{serverOf}}
\newcommand{\validatorOf}{\mathsf{validatorOf}}
\newcommand{\Return}{\mathsf{return}}
\newcommand{\Sexp}{\mathsf{SExp}}

\newcommand{\Eval}{\mathsf{eval}}
\newcommand{\Exec}{\mathsf{exec}}
\newcommand{\update}[3]{#1[#2\mapsto#3]}
\newcommand{\constraint}{\mathsf{constraint}}
\newcommand{\Fresh}{\mathsf{fresh}\;}
\newcommand{\solvable}{\mathsf{solvable}\;}
\newcommand{\Vexp}{\mathsf{VExp}}
\newcommand{\oop}{\;\mathit{op}\;}
\newcommand{\ccmp}{\;\mathit{cmp}\;}
\newcommand{\vs}{\mathit{vs}}
\newcommand{\cs}{\mathit{cs}}
\newcommand{\src}{\mathit{src}}
\newcommand{\dst}{\mathit{dst}}
\newcommand{\asgn}{\mathit{asgn}}
\newcommand{\Label}{\mathit{label}}
\newcommand{\Function}{\mathit{function}}
\newcommand{\Index}{\mathit{index}}
\newcommand{\Field}{\mathit{field}}
\newcommand{\Var}{\mathsf{var}}
\newcommand{\Reflects}[2]{#1\sim #2}
\newcommand{\satisfy}[2]{#1\;\mathsf{satisfies}\;#2}
\newcommand{\Jref}[3]{#1.#2.#3}
\newcommand{\Jpath}{\mathsf{Jpath}}
\newcommand{\This}{\mathsf{this}}
\newcommand{\Number}{\#}
\newcommand{\At}{@}
\newcommand{\Write}{\mathsf{write}}
\newcommand{\Havoc}{\mathsf{havoc}}
\newcommand{\ie}{\text{i.e.}}
\newcommand{\eg}{\text{e.g.}}

\numberwithin{definition}{chapter}

\makeatletter
\g@addto@macro\@floatboxreset\centering
\makeatother

\begin{document}
\maketitle 
\setcounter{page}{2}

\makecreativecommons 


\makeacknowledgement 

\makeabstract
\tableofcontents


\clearpage \phantomsection \addcontentsline{toc}{chapter}{LIST OF ILLUSTRATIONS} \listoffigures 


\begin{mainf} 

\chpt{Introduction}
\label{chap:introduction}
Software engineering requires rigorous testing of rapidly evolving programs,
which costs manpower comparable to developing the product itself~\cite{vailshery}.  To guarantee
programs' compliance with their specifications, we need testers that can tell
compliant implementations from violating ones.

This thesis studies how to test the semantics of interactive systems: The system
under test (SUT) interacts with the tester by sending and receiving messages,
and the tester determines whether the messages sent by the SUT are valid with
respect to the protocol specification.  This kind of testing is applicable in
many scenarios, including web servers, distributed file systems, etc.

This chapter provides a brief view of interactive testing
(\autoref{sec:interactive-testing}), explains why nondeterminism makes this
problem difficult
(Sections~\ref{sec:internal-external-nondeterminism}--\ref{sec:inter-execution-nondeterminism}),
discusses the field of existing works (\autoref{sec:existing-work}), and
summarizes the contributions of this thesis in addressing the challenges caused
by nondeterminism (\autoref{sec:contribution}).

\paragraph{Convention}
In this thesis, I use standard terminologies and conventions from functional
programming, such as monads and coinduction.  The meta language for data
structures and algorithms is Coq, with syntax simplified in places for
readability.

\section{Interactive Testing}
\label{sec:interactive-testing}
Suppose we want to test a web server that supports GET and PUT methods.  We can
represent the server as a stateful program.
\begin{coq}
  CoFixpoint server (data: key -> value) :=
    request <- recv;;
    match request with
    | GET k   => send (data k);; server  data
    | PUT k v => send  Done   ;; server (data [k |-> v])
    end.
\end{coq}
Here syntax ``\ilc{x <- f;; y}'' encodes a monadic program that binds the result
of computation \ilc f as variable \ilc x in continuation \ilc y.  For example,
to receive a request is to bind the result of \ilc{recv} as variable
\ilc{request} in the remaining program that performs pattern matching on it.
Syntax ``\ilc{data [k|->v]}'' represents a key-value store where \ilc k is
mapped to \ilc v, and all other keys are mapped by \ilc{data}.  That is, for
all \ilc{k'} that are not equal to \ilc k, \ilc{(data [k|->v]) k'} is equal
to \ilc{(data k')}.

The \ilc{server} function iterates with a parameter called \ilc{data}, which is
a key-value store.  In each iteration, the server receives a request and
computes its response.  It then sends back the response and recurses with the
updated data.

We can write a tester client that interacts with the server and determines
whether it behaves correctly:
\begin{coq}
  CoFixpoint tester (data: key -> value) :=
    request <- random;;
    send request;;
    response <- recv;;
    match request with
    | GET k   => if response =? data k
                 then tester data
                 else reject
    | PUT k v => if response =? Done
                 then tester (data [k |-> v])
                 else reject
    end.
\end{coq}
This tester implements a reference server internally that computes the expected
behavior.  The behavior is then compared against that produced by the SUT.  The
tester rejects the SUT upon any difference from the computed expectation.

The above tester can be restructured into two modules: (i) a {\em test harness}
that interacts with the server and produces transactions of sends and receives,
and (ii) a {\em validator} that determines whether the transactions are valid or
not:
\begin{coq}
  (* Compute the expected response and next state of the server. *)
  Definition serverSpec request data :=
    match request with
    | GET k   => (data k, data)
    | PUT k v => (Done  , data [k |-> v])
    end.

  (* Validate the transaction against the stateful specification. *)
  Definition validate spec request response data :=
    let (expect, next) := spec request data in
    if response =? expect then Success next else Failure.

  (* Produce transactions for the validator. *)
  CoFixpoint harness validator state :=
    request <- random;;
    send request;;
    response <- recv;;
    if validator request response state is Success next
    then harness validator next
    else reject.

  Definition tester := harness (validate serverSpec).
\end{coq}
This testing method works for deterministic systems, whose behavior can be
precisely computed from their inputs.  But, many systems are allowed to behave
nondeterministically.  For example, systems may implement various hash
algorithms, or buffer network packets in different ways.  The following sections
discuss nondeterminism by partitioning it in two ways, and explains how they
pose challenges to the validator and the test harness.

\section{Internal and External Nondeterminism}
\label{sec:internal-external-nondeterminism}
When people talk to each other, voice is transmitted over substances like air or
metal.  When testers interact with the SUT, messages are transmitted via the
runtime environment.  The specification might allow SUTs to behave differently
from each other, just like people speaking in different accents; we call it {\em
internal nondeterminism}.  The runtime environment might affect the transmission
of messages, just like solids transmit voice faster than liquids and gases; we
call it {\em external nondeterminism}.

\subsection{Internal nondeterminism}
\label{sec:internal-nondeterminism}
Within the SUT, correct behavior may be \mbox{underspecified}.  Consider web
browsing as an example: If a client has cached a local copy of some resource,
then when fetching updates for the resource, the client can ask the server not
to send the resource's contents if it is the same as the cached copy.  To
achieve this, an HTTP server may generate a short string, called an ``entity
tag'' (ETag)~\cite{rfc7232}, identifying the content of some resource, and send
it to the client:
\begin{multicols}{2}
\begin{cpp}
  /* Client: */
  GET /target HTTP/1.1
\end{cpp}
\columnbreak
\begin{cpp}
  /* Server: */
  HTTP/1.1 200 OK
  ETag: "tag-foo"
  ... content of /target ...
\end{cpp}
\end{multicols}
The next time the client requests the same resource, it can include the ETag in
the GET request, informing the server not to send the content if its ETag still
matches:
\begin{multicols}{2}
\begin{cpp}
  /* Client: */
  GET /target HTTP/1.1
  If-None-Match: "tag-foo"
\end{cpp}
\columnbreak
\begin{cpp}
  /* Server: */
  HTTP/1.1 304 Not Modified
\end{cpp}
\end{multicols}
If the ETag does not match, the server responds with code 200 and the updated
content as usual.

Similarly, if a client wants to modify the server's resource atomically by
compare-and-swap, it can include the ETag in the PUT request as an
\inlinec{If-Match} precondition, which instructs the server to only update the
content if its current ETag matches:
\begin{multicols}{2}
\begin{cpp}
  /* Client: */
  PUT /target HTTP/1.1
  If-Match: "tag-foo"
  ... content (A) ...
\end{cpp}
\columnbreak
\begin{cpp}
  /* Server: */
  HTTP/1.1 204 No Content
\end{cpp}
\end{multicols}
\begin{multicols}{2}
\begin{cpp}
  /* Client: */
  GET /target HTTP/1.1
\end{cpp}
\columnbreak
\begin{cpp}
  /* Server: */
  HTTP/1.1 200 OK
  ETag: "tag-bar"
  ... content (A) ...
\end{cpp}
\end{multicols}
If the ETag does not match, then the server should not perform the requested
operation, and should reject with code 412:
\begin{multicols}{2}
\begin{cpp}
  /* Client: */
  PUT /target HTTP/1.1
  If-Match: "tag-baz"
  ... content (B) ...
\end{cpp}
\columnbreak
\begin{cpp}
  /* Server: */
  HTTP/1.1 412 Precondition Failed
\end{cpp}
\end{multicols}

\begin{multicols}{2}
\begin{cpp}
  /* Client: */
  GET /target HTTP/1.1
\end{cpp}
\columnbreak
\begin{cpp}
  /* Server: */
  HTTP/1.1 200 ok
  ETag: "tag-bar"
  ... content (A) ...
\end{cpp}
\end{multicols}
Whether a server's response should be judged {\em valid} or not depends on the
ETag it generated when creating the resource.  If the tester doesn't know the
server's internal state (\eg, before receiving any 200 response that
includes an ETag), and cannot enumerate all of them (as ETags can be arbitrary
strings), then it needs to maintain a space of all possible values, and narrow
the space upon further interactions with the server.  For example, ``If the
server has revealed some resource's ETag as \inlinec{"tag-foo"}, then it must
not reject requests targeting this resource conditioned over \inlinec{If-Match:
  "tag-foo"}, until the resource has been modified''; and ``Had the server
previously rejected an \inlinec{If-Match} request, it must reject the same
request until its target has been modified.''

\begin{figure}
\begin{coq}
  Definition validate request response
             (data      : key -> value)
             (tag_is    : key -> Maybe etag)
             (tag_is_not: key -> list etag) :=
    match request, response with
    | PUT k t v, NoContent => 
      if t \in tag_is_not k then Failure
      else if (tag_is k =? Unknown) || strong_match (tag_is k) t
      then (* Now the tester knows that the data in [k]
            * is updated to [v], but its new ETag is unknown. *)
        Success (data       [k |-> v],
                 tag_is     [k |-> Unknown],
                 tag_is_not [k |-> [] ])
      else Failure
    | PUT k t v, PreconditionFailed =>
      if strong_match (tag_is k) t then Failure
      else (* Now the tester knows that the ETag of [k]
            * is other than [t]. *)
        Success (data, tag_is, tag_is_not [k |-> t::(tag_is_not k)])
    | GET k t, NotModified =>
      if t \in tag_is_not then Failure
      else if (tag_is k =? Unknown) || weak_match (tag_is k) t
      then (* Now the tester knows that the ETag of [k]
            * is equal to [t]. *)
        Success (data, tag_is [k |-> Known t], tag_is_not)
      else Failure
    | GET k t0, OK t v =>
      if weak_match (tag_is k) t0 then Failure
      else if data k =? v
      then (* Now the tester knows the ETag of [k]. *)
        Success (data, tag_is [k |-> Known t], tag_is_not)
      else Failure
    | _, _ => Failure
    end.    
\end{coq}
\vspace*{1em}
  \caption[Ad hoc tester for \http conditional requests.]{Ad hoc tester for
    \http conditional requests.  \ilc{PUT k t v} represents a PUT request that
    changes \ilc k's value into \ilc v only if its ETag matches \ilc t; \ilc{GET
      k t} is a GET request for \ilc k's value only if its ETag does not match
    \ilc t; \ilc{OK t v} indicates that the request target's value is \ilc v and
    its ETag is \ilc t.}
  \label{fig:etag-tester}
\end{figure}

This idea of remembering matched and mismatched ETags is implemented in
\autoref{fig:etag-tester}.  For each key, the validator maintains three internal
states: (i) The value stored in \ilc{data}, (ii) the corresponding resource's
ETag, if known by the tester, stored in \ilc{tag_is}, and (iii) ETags that are
known to not match the resource's, stored in \ilc{tag_is_not}.  Each pair of
request and response contributes to the validator's knowledge of the target
resource.  The tester rejects the SUT if the observed behavior does not match
the knowledge gained in previous interactions.

Even simple nondeterminism like ETags requires such careful design of the
validator, based on thorough comprehension of the specification.  We need to
construct such validators in a scalable way for more complex protocols.  This is
one challenge posed by internal nondeterminism.

\subsection{External nondeterminism}
\label{sec:intro-external-nondet}
To discuss the nondeterminism caused by the environment, we need to more
precisely define the environment concept as it pertains to this testing
scenario.
\begin{definition}[Environment, input, output, and observations]
\label{def:environment}
  The {\em environment} is the substance that the tester and the SUT interact
  through.  The {\em input} is the subset of the environment that the tester can
  manipulate.  The {\em output} is the subset of the environment that the SUT
  can alter.  The {\em observation} is the tester's view of the inputs and the
  outputs.
\end{definition}
When testing servers, the environment is the network stack between the client
and the server.  The input is the sequence of requests sent by the client, and
the output is the sequence of responses sent by the server.  The responses are
transmitted via the network, until reaching the client side as observations.

\begin{figure}
  \centering
  \begin{minipage}[c]{.3\textwidth}
\begin{coq}
  (* Observation: *)
  1> PUT k "new"
  1< Done
  2> GET k
  2< "new"
\end{coq}
  \end{minipage}\begin{minipage}[c]{.4\textwidth}
  \includegraphics[width=\linewidth]{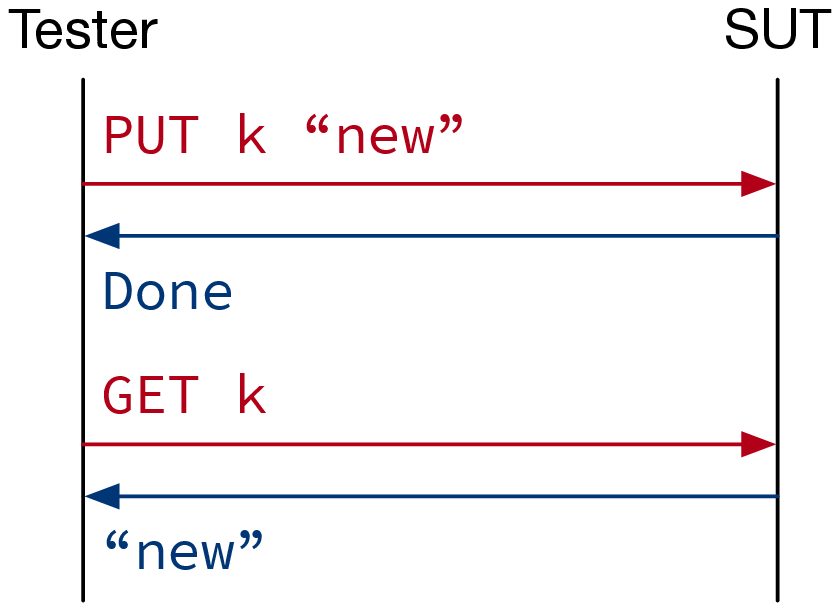}
  \end{minipage}\begin{minipage}[c]{.3\textwidth}
\begin{coq}
  (* Output: *)
  1> PUT k "new"
  1< Done
  2> GET k
  2< "new"
\end{coq}
  \end{minipage}
  \caption[Linear trace with no concurrency.]{With no concurrency, the
    observation is identical to the output.}
  \label{fig:linear-trace}
\end{figure}
\begin{figure}
  \centering
  \begin{minipage}[c]{.3\textwidth}
\begin{coq}
  (* Observation: *)
  1> PUT k "new"
  2> GET k
  1< Done
  2< "old"
\end{coq}
  \end{minipage}\begin{minipage}[c]{.4\textwidth}
    \includegraphics[width=\linewidth]{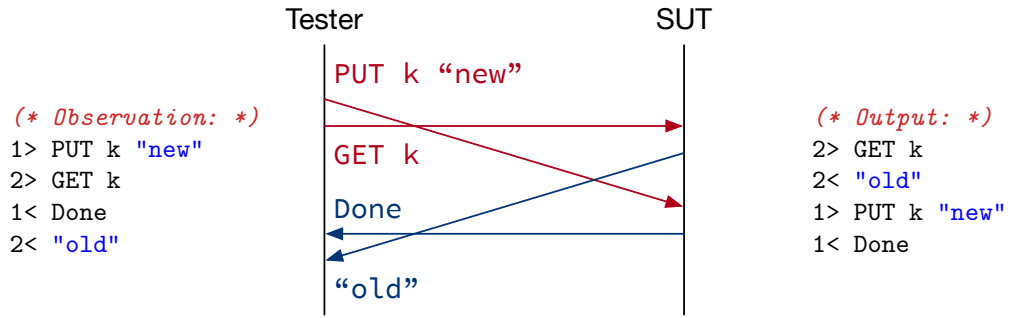}
  \end{minipage}\begin{minipage}[c]{.3\textwidth}
\begin{coq}
  (* Output: *)
  2> GET k
  2< "old"
  1> PUT k "new"
  1< Done
\end{coq}
  \end{minipage}
  \caption[Reordered trace upon network delays.]{Acceptable: The observation can
    be explained by a valid output reordered by the network environment.}
  \label{fig:reordered-trace}
\end{figure}
\begin{figure}
  \centering
  \begin{minipage}[c]{.3\textwidth}
\begin{coq}
  (* Observation: *)
  1> PUT k "new"
  1< Done
  2> GET k
  2< "old"
\end{coq}
  \end{minipage}\begin{minipage}[c]{.35\textwidth}
  \includegraphics[width=\linewidth]{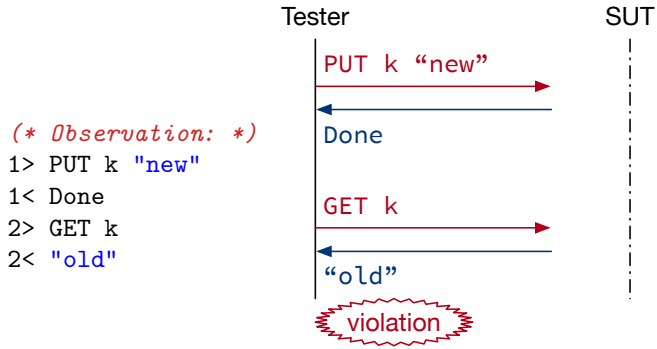}
  \end{minipage}\begin{minipage}[c]{.3\textwidth}\
  \end{minipage}
  \caption[Invalid trace that violates the specification.]{Unacceptable: The
    tester received the \ilc{Done} response before sending the \ilc{GET}
    request, thus the SUT must have processed the \ilc{PUT} request before the
    \ilc{GET} request.  Therefore, the \ilc{"old"} response is invalid.}
  \label{fig:invalid-trace}
\end{figure}
The tester shown in \autoref{sec:interactive-testing} runs one client at a time.
It waits for the response before sending the next request, as shown in
\autoref{fig:linear-trace}.  Such tester's observation is guaranteed identical
to the SUT's output, so it only needs to scan the requests and responses with
one stateful validator.

To observe the server's behavior upon concurrent requests, the tester needs to
simulate multiple clients, sending new requests before receiving previous
responses.  The network delay might cause the server to receive requests in a
different order from that on the tester side.  Conversely, responses sent by the
server might be reordered before arriving at the tester, as shown in
\autoref{fig:reordered-trace}.  Such an observation can be explained by various
outputs on the SUT side.  The validator needs to consider all possible outputs
that can explain such an observation and see if any one of them complies with
the specification.  If no valid output can explain the observation, then the
tester should reject the SUT, as shown in \autoref{fig:invalid-trace}.

We need to construct a tester that can handle external nondeterminism
systematically and provide a generic way for reasoning on the environment.

\section{Test Harness and Inter-execution Nondeterminism}
\label{sec:inter-execution-nondeterminism}
A good tester consists of (i) a validator that accurately determines whether its
observations are valid or not, and (ii) a test harness that can reveal invalid
observations effectively.  \autoref{sec:internal-external-nondeterminism} has
explained the challenges in the validator.  Here we discuss the test harness.

\subsection{Test harness}
Intuitively, a tester generates a test input and launches the test execution.
It then validates the observation and accepts/rejects the SUT, as shown in
\autoref{fig:gen-only}.
\begin{figure}
  \includegraphics[width=.4\linewidth]{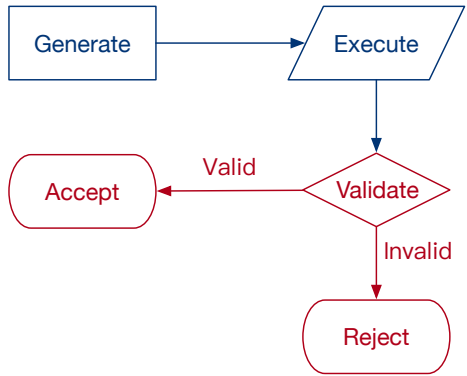}
  \caption{Simple tester architecture without shrinking.}
  \label{fig:gen-only}
\end{figure}

However, to achieve better coverage, a randomized generator might produce huge
test inputs~\cite{Hughes2016}.  Suppose the tester has revealed an invalid
observation after thousands of interactions; such a report provides limited
intuition of where the bug was introduced.  To help developers locate the bug
more effectively, the tester should present a {\em minimal counterexample} that
can reproduce the violation.  This is done by {\em shrinking} the failing input
and rerunning the test with the input's substructures.  As shown
in \autoref{fig:gen-shrink}, if a test input has no substructure that can cause
any failure, then we report it as the minimal counterexample.
\begin{figure}
\vspace*{2em}
  \includegraphics[width=.64\linewidth]{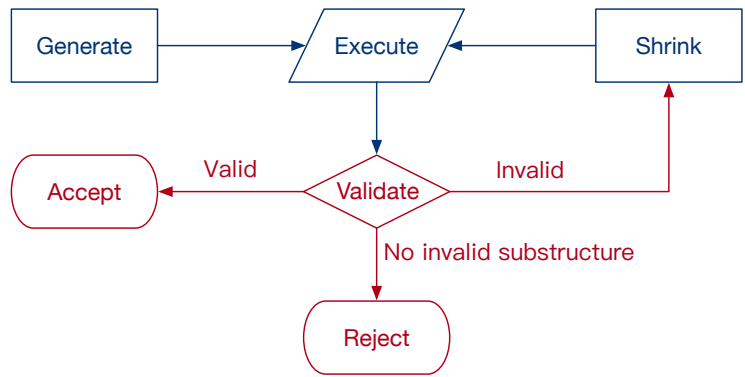}
  \caption{Tester architecture with shrinking mechanism.}
  \label{fig:gen-shrink}
\end{figure}

The test harness consists of generator, shrinker, and executor.  This thesis
studies the generator and the shrinker that produce the test input.

Interesting test inputs are those that are more likely to reveal invalid
observations.  Such a subset is usually sparse and cannot be enumerated within a
reasonable budget, \eg, in \autoref{sec:internal-nondeterminism}, it would be
infeasible to enumerate request ETags to match the target resources'.  The
tester needs to manipulate the inputs' distribution, by implementing {\em
heuristics} that emphasize certain input patterns, which is challenged by
another form of nondeterminism discussed as follows.

\subsection{Inter-execution nondeterminism}
\label{sec:inter-execution}
Consider \http, where requests may be conditioned over timestamps.  If a client
has cached a version with a certain timestamp, then it can send the timestamp as
the \inlinec{If-Modified-Since} precondition.  The server should not transmit
the request target's content if its \inlinec{Last-Modified} timestamp is not
newer than the precondition's:
\begin{lstlisting}[style=customc,escapeinside={(*}{*)}]
  /* Client: */
  GET /index.html HTTP/1.1
  If-Modified-Since: (*\httpdate\DayAfter[-1]~\currenttime *) GMT
                             /* Server: */
                             HTTP/1.1 200 OK
                             Last-Modified: (*\httpdate\today~\currenttime *) GMT
                             ... content of target ...
  /* Client: */
  GET /index.html HTTP/1.1
  If-Modified-Since: (*\httpdate\today~\currenttime *) GMT
                             /* Server: */
                             HTTP/1.1 304 Not Modified
\end{lstlisting}

In this scenario, an interesting candidate for the \inlinec{If-Modified-Since}
precondition in a test case is the \inlinec{Last-Modified} timestamp of a
previous response.  To emphasize this request pattern, the tester needs to
implement heuristics that generate test inputs based on previous observations.

In case the tester has revealed invalid observations from the server, it needs
to rerun the test with shrunk input.  The problem is that the timestamps on the
server might be different from the previous execution, so an interesting
timestamp in a previous run might become meaningless in this run.

The fact that a system may perform differently among executions is called {\em
inter-execution nondeterminism}.  Such nondeterminism poses challenges to the
input minimization process: To preserve the input pattern, the shrunk \http
request should use the timestamps from the new execution.  We need to implement
a generic shrinking mechanism that can reproduce the heuristics in the test
generator's design.

\section{State of the Art}
\label{sec:existing-work}
This section explains the context for this thesis.  I sketch the state of the
art in the relevant parts of the field of testing and describe its limitations
in addressing the challenges posed by nondeterminism.

\subsection{Property-based testing: QuickCheck}
Property-based testing~\cite{pbt} is a testing methodology for validating
semantic properties of programs' behavior.  The properties are specified as
executable boolean predicates over the behavior.  To check whether an SUT
satisfies a specification, the tester generates test input and executes the SUT
with the generated input.  The tester then observes the SUT's behavior and
computes the predicates with the observations.

Practices of property-based testing include QuickCheck for Haskell~\cite{qc} and
its variant QuickChick for Coq~\cite{quickchick}.  These tools can generate
random inputs that satisfy logical conditions~\cite{gengood} and integrate with
fuzz testing~\cite{fuzzchick} and combinatorial testing~\cite{judge-cover}, and
they have tested real-world systems like telecoms software~\cite{Quviq2006} and
Dropbox~\cite{testing-dropbox}.

\subsection{Model-based testing: TorXakis}
Instead of specifying predicates over systems' behavior, model-based
testing~\cite{broy2005model} defines an abstract model that produces valid behavior.  When
a tester observes an SUT's behavior, it checks whether the behavior is
producible by the specification model.

Practical tools for model-based testing include TorXakis~\cite{TorXakis}, whose
modeling language is inspired by Language of Temporal Ordering Specification
(LOTOS)~\cite{lotos}, the ISO standard for specifying distributed systems.  The
tool can compile process algebra specifications into tester programs and can be
used for testing DropBox~\cite{torxakis-dropbox}.

\subsection{Limitations}
In property-based testing, internal and external nondeterminism makes predicates
difficult to write, as discussed
in \autoref{sec:internal-external-nondeterminism}.  TorXakis provides limited
support for internal nondeterminism, allowing the specification to enumerate all
possible values of internal choices.  This works for scenarios where the space
of choices is small, \eg, within a dozen.  When testing ETags
in \autoref{sec:internal-nondeterminism}, it's infeasible to include a list of
all strings in the specification.

To handle inter-execution nondeterminism in QuickCheck, \citet{Hughes2016}
introduced abstract representations for generating and shrinking test inputs
that can adapt to different runtime observations.  His technique works for
synchronous interactions that blocks the tester to wait for observations, and
lacks support for asynchronous testing where the SUT's output may be
indefinitely delayed by the environment.

\section{Contribution}
\label{sec:contribution}
This thesis presents ``testing by dualization'' (TBD), a technique that
addresses the challenges in asynchronous testing caused by various forms of
nondeterminism.  I introduce symbolic languages for specifying the protocol and
representing test input, and I dualize the specification into the tester's (1)
validator, (2) generator, and (3) shrinker:

\begin{enumerate}
\item The specification is written as a reference implementation---a
  nondeterministic program that exhibits all possible behaviors allowed by the
  protocol.  Internal and external nondeterminism are represented by symbolic
  variables, and the space of nondeterministic behavior is defined by all
  possible assignments to the variables.

  For internal nondeterminism, the validator computes the symbolic
  representation of the SUT's output.  The symbolic output expectation is then
  {\em unified} against the tester's observations, reducing the problem of
  validating observations to constraint solving.

  For external nondeterminism, I introduce a model that specifies the
  environment.  The environment model describes the relation between the SUT's
  output and the tester's observations.  By composing the environment model with
  the reference implementation, we get a tester-side specification that defines
  the space of valid observations.
\item Test generation heuristics are defined as computations from observations
  to the next input.  To specify such heuristics in a generic way, I introduce
  intermediate representations for observations and test inputs, which are
  protocol-independent.

  Heuristics in this framework produce symbolic test inputs that are
  parameterized over observations.  During execution, the test harness computes
  the concrete input by {\em instantiating} the symbolic input's arguments with
  runtime observations.
\item The language for test inputs is designed with inter-execution
  nondeterminism in mind.  By instantiating the inputs' symbolic intermediate
  representation with different observations, the test harness gets different
  test inputs but preserves the pattern.

  To minimize counterexamples, the test harness only needs to shrink the inputs'
  symbolic representation.  When rerunning the test, the shrunk input is
  reinstantiated with the new observations, thus reproduces the heuristics by
  the test generator.
\end{enumerate}

\paragraph{Thesis claim}
Testing by dualization can address challenges in testing interactive systems
with uncertain behavior.  Specifying protocols with a symbolic reference
implementation enables validating observations of systems with internal and
external nondeterminism.  Representing test input and observations symbolically
allows generating and shrinking interesting test cases despite inter-execution
nondeterminism.  Combining these methods results in a rigorous tester that can
capture protocol violations effectively.

This claim is supported by the following publications:
\begin{enumerate}
\item {\it From C to Interaction Trees: Specifying, Verifying, and Testing a
  Networked Server}~\citep{cpp19}, with Nicolas Koh, Yao Li, Li-yao Xia, Lennart
  Beringer, Wolf Honor\'e, William Mansky, Benjamin C. Pierce, and Steve
  Zdancewic, where I developed a tester program based on a swap server's
  specification written as ITrees~\citep{itree}, and evaluated the tester's
  effectiveness by mutation testing.
\item {\it Verifying an HTTP Key-Value Server with Interaction Trees and
  VST}~\citep{itp21}, with Hengchu Zhang, Wolf Honor\'e, Nicolas Koh, Yao Li,
  Li-yao Xia, Lennart Beringer, William Mansky, Benjamin C. Pierce, and Steve
  Zdancewic, where I developed the top-level specification for \http, and
  derived a tester client that revealed liveness and interrupt-handling bugs in
  our HTTP server, despite it was formally verified.
\item {\it Model-Based Testing of Networked Applications}~\citep{issta21}, with
  Benjamin C. Pierce and Steve Zdancewic, which describes my technique of
  specifying \http with symbolic reference implementations, and from the
  specification, automatically deriving a tester program that can find bugs in
  Apache and Nginx.
\end{enumerate}

\paragraph{Outline}
This thesis is structured as follows: \autoref{chap:theory} presents a theory
for synchronous testing, introduces a language family for representing
validators, and shows how to reason about their correctness.
\autoref{chap:dualize} applies the validator theory to a computation model
that exhibits internal nondeterminism, showing how to derive validators
automatically from protocol specifications.
\autoref{chap:practices} transitions from synchronous testing to asynchronous
testing, addressing external nondeterminism by specifying the environment.
\autoref{chap:harness} presents a mechanism for generating and shrinking test
inputs that address inter-execution nondeterminism.  To evaluate the techniques
proposed in this thesis, \autoref{chap:eval} applies them to testing web servers
and file synchronizers.  I then compare my technique with related works
in \autoref{chap:related-work} and discuss future directions in
\autoref{chap:discussion}.

\chpt{Validator Theory}
\label{chap:theory}
This chapter provides a theoretical view of synchronous testing that involves
internal nondeterminsm.  \autoref{sec:concepts} defines the basic concepts in
testing.  \autoref{sec:qac} introduces a language family for writing protocol
specifications and validators.  \autoref{sec:correctness} shows how to reason on
the soundness and completeness of validators with respect to the specification.

\section{Concepts}
\label{sec:concepts}
Testers are programs that determine whether implementations are compliant or
not, based on observations.  This section defines basic concepts and notations
in interactive testing.

\begin{definition}[Implementations and Traces]
  {\em Implementations} are programs that can interact with their environment.
  {\em Traces} are the implementations' inputs and outputs during execution.
  ``Implementation $i$ can {\em produce} trace $t$'' is written as ``$\behaves i
  t$''.
\end{definition}

This chapter focuses on synchronous testing and assumes no external
nondeterminism.  Here the tester's observation is identical to the
implementation's output, so the tester-side trace is the same as that on the
implementation side.  Asynchronous testing will be discussed in
\autoref{chap:practices}.

\begin{definition}[Specification, Validity, and Compliance]
  \label{def:compliance}
  A {\em specification} is a description of valid traces.  ``Trace $t$ is {\em
    valid} per specification $s$'' is written as ``$\valid s t$''.

  An implementation $i$ {\em complies} with a specification $s$ (written as
  ``$\complies s i$'') if it only produces traces that are valid per the
  specification:
  \[\complies s i\quad\triangleq\quad\forall t,(\behaves i t)\implies\valid s t\]
\end{definition}

\begin{definition}[Tester components and correctness]
  \label{def:tester}
  A tester consists of (i) a {\em validator} that accepts or rejects
  traces (written as ``$\accepts v t$'' and ``$\rejects v t$''), and
  (ii) a {\em test harness} that triggers different traces with
  various inputs.

  A tester is {\em rejection-sound} if it rejects only non-compliant
  implementations; it is {\em rejection-complete} if it can reject all
  non-compliant implementations, provided sufficient time of execution.  A
  tester is {\em correct} if is rejection-sound and -complete.\footnote{The
    semantics of ``soundness'' and ``completeness'' vary among contexts.  This
    thesis inherits terminologies from existing literature~\cite{Tretmans}, but
    explicitly uses ``rejection-'' prefix for clarity.  ``Rejection soundness''
    is equivalent to ``acceptance completeness'', and vice versa.}
\end{definition}

The tester's correctness is based on its components' properties: A
rejection-sound tester requires its validator to be rejection-sound; A
rejection-complete tester consists of (i) a rejection-complete validator and
(ii) an exhaustive test harness that can eventually trigger invalid traces.  The
validators' soundness and completeness are defined as follows:

\begin{definition}[Correctness of validators]
  A validator $v$ is {\em rejection-sound} with respect to specification $s$
  (written as ``$\rejSound v s$'') if it only rejects traces that are invalid
  per $s$:
  \[\rejSound v s\quad\triangleq\quad\forall t,\rejects v t\implies\invalid s t\]

  A validator $v$ is {\em rejection-complete} with respect to specification $s$
  (written as ``$\rejComplete v s$'') if it rejects all behaviors that are
  invalid per $s$:
  \[\rejComplete v s\quad\triangleq\quad\forall t,\invalid s t\implies\rejects v t\]
\end{definition}

The rest of this chapter shows how to construct specifications and validators,
and how to prove the validators' correctness with respect to the specifications.

\section{QAC Language Family}
\label{sec:qac}
In this section, I introduce the ``query-answer-choice'' (QAC) language family
for writing specifications and validators for network protocols that involve
internal nondeterminism.  Languages in the QAC family can specify protocols of
various message types, server states, and internal choices, by instantiating the
specification language with different type arguments.

\subsection{Specifying protocols with server models}
\label{sec:qac-model}
Network protocols can be specified with ``reference implementations'', \ie,
model programs that exhibit the space of valid behaviors.  For client-server
systems such as WWW, we can specify networked servers as programs that receive
queries and compute the responses.  Here I model the server programs with a data
structure called state monad.

\begin{definition}[State monad]
  Let $S$ be the state type, and $A$ be the result type.  Then type $(S\to
  A\times S)$ represents a computation that, given a pre state, yields a result
  and the post state.  This computation is pronounced a ``state monad with state
  type $S$ and result type $A$''.

  For example, let the state be a key-value mapping $(K\to V)$, then we can
  define \ilc{get} and \ilc{put} computations as follows:
  \begin{align}
    \tag{1}&\mathtt{get}:K\to((K\to V)\to V\times(K\to V))\\
    \tag{2}&\mathtt{get}(k)(f)\triangleq (f(k), f)\\
    \tag{3}&\mathtt{put}:K\times V\to((K\to V)\to ()\times(K\to V))\\
    \tag{4}&\mathtt{put}(k,v)(f)\triangleq((),\update f k v)
  \end{align}

  These function definitions should be read as:
  \begin{enumerate}
    \item The \ilc{get} function takes a key as argument, and constructs a
      state monad with state type $(K\to V)$ and result type $V$.
    \item Given argument $k$ of type $K$, $\mathtt{get}(k)$ takes a mapping $f$
      as pre state and yields the mapped value $f(k)$ as result.  The post state
      is the original mapping $f$ unchanged.
    \item The \ilc{put} function takes a key-value pair as argument, and
      constructs a state monad with state type $(K\to V)$ and result type
      ``$()$'' (unit type, which corresponds to \inlinec{void} return type in
      C/Java functions).
    \item Given argument $(k,v)$ of type $(K\times V)$, $\mathtt{put}(k,v)$
      takes a mapping $f$ as pre state and substitues its value at key $k$ with
      $v$.  The post state is the substituted mapping $\update f k v$.
  \end{enumerate}
\end{definition}

Now we can define the server model in terms of state monad:

\begin{definition}[Deterministic server model]
  \label{def:qaserver}
  A deterministic server is an infinite loop whose loop body takes a query and
  produces a response.  The server definition consists of the loop body and a
  current state:
  \[\mathsf{DeterministicServer}\triangleq\sigT{S}{(Q \to S \to A \times S) \times S}\]
  This type definition is pronounced as: A deterministic server has an initial
  state of some type $S$.  Its loop body takes a request of type $Q$ and
  computes a state monad with state type $S$ and result type $A$, where type $A$
  represents the response.

  Note that the server's state type is existentially quantified~\cite{tapl},
  while its query and response types are not.  This is because a protocol
  specification only defines the space of valid traces, and doesn't require the
  implementation's internal state to be a specific type.

  An instance of server model is written as:
  \[\existT S \sigma (\sstep,state_0)\]
  This expression is pronounced as: The server state is of type $\sigma$.  Its
  loop body is function $\sstep$ (which has type $Q\to\sigma\to A\times\sigma$)
  and its initial state is $state_0$ (which has type $\sigma$).
\end{definition}

For example, consider a compare-and-set (CMP-SET) protocol: The server stores a
number \inlinec n.  If the client sends a request that is smaller than \inlinec
S, then the server responds with \inlinec 0.  Otherwise, the server sets
\inlinec n to the request and responds with \inlinec 1:

\begin{minipage}{\linewidth}
\begin{cpp}
  int n = 0;
  while (true) {
    int request = recv();
    if (request <= n) send(0);
    else { n = request; send(1); }
  }
\end{cpp}
\end{minipage}

Such a server can be modelled as:
\begin{align*}
  \existT{S}{\Int}{(&\lam{(q)(n)}{\begin{cases}
        (0,s)&q\le n\\
        (1,q)&\mathrm{otherwise}
    \end{cases}}\\
    &,0)}
\end{align*}
In general, servers' responses and transitions might depend on choices that are
invisible to the testers, so called internal nondeterminism, as discussed in
\autoref{sec:internal-nondeterminism}.  I represent the space of nondeterminstic
behaviors by parameterizing it over the server's internal choice.

\begin{definition}[Nondeterministic server model]
  \label{def:server}
  A nondeterministic server is an infinite loop whose loop body takes a query
  and an internal choice to produce a response.  The nondeterministic server
  definition extends \autoref{def:qaserver} with a choice argument of type $C$:
  \[\Server\triangleq\sigT S{(Q\times C\to S\to A\times S)\times S}\]
\end{definition}

Consider changing the aforementioned CMP-SET into compare-and-reset (CMP-RST):
When the request is greater than \inlinec S, the server may reset \inlinec S to
any arbitrary number:

\begin{cpp}
  int arbitrary();
  int n = 0;
  while (true) {
    int request = recv();
    if (request <= n) send(0);
    else { n = arbitrary(); send(1); }
  }
\end{cpp}

Its corresponding server model can be written as:
\begin{align*}
  \existT{S}{\Int}{(&\lam{(q,c)(n)}{
      \begin{cases}
        (0,n)& q\le n\\
        (1,c)&\mathrm{otherwise}
      \end{cases}
    }\\
    &,0)}
\end{align*}
This model represents the space of uncertain behavior with the internal choice
parameter of type integer.  For any value $(c:\Int)$, the server is allowed to
reset $S$ to $c$.

\subsection{Valid traces of a server model}
By specifying protocols with server models, we can now instantiate the trace
validity notation ``$\valid s t$'' in \autoref{def:compliance} in terms of
operational semantics.

\begin{definition}[Server transitions]
  \label{def:server-step}
  Upon request $q$ and choice $c$, server model $s$ can step to $s'$ yielding
  response $a$ (written ``$\triggers sc{(q,a)}s'$~'') if and only if the
  response and the post model can be computed by the $\stepServer$ function:
\begin{align*}
  &\triggers sc{(q,a)}s'\quad\triangleq\quad\stepServer(q,c)(s)=(a,s')\\
  &\stepServer:Q\times C \to \Server \to A\times \Server \\
  &\stepServer(q,c)(s)\triangleq\\
  &\qquad\letin{(\existT{S}{\sigma}{(\sstep, state)})}{s}\\
  &\qquad\letin{(a,state')}{\sstep(q,c)(state)} \\
  &\qquad(a, \existT{S}{\sigma}{(\sstep, state')})
\end{align*}
The $\stepServer$ function takes a query and a choice and computes a state monad
with state type $\Server$ and result type $A$, by pattern matching on argument
$(s:\Server)$.  Let $\sigma$ be the server state type of $s$, $\sstep$ the loop
body, and $(state:\sigma)$ the current state of $s$.  Then $\stepServer(q,c)(s)$
computes the result $(a:A)$ and the post state $(state':\sigma)$ using the
$\sstep$ function, and substitutes the server's pre-step $state$ with the
post-step $state'$.
\end{definition}

\begin{definition}[Trace validity in QAC]
  \label{def:trace-validity}
  In the QAC language family, a trace is a sequence of $Q\times A$ pairs.  When
  specifying a protocol with a $\Server$ model, a trace $t$ is valid per
  specification $s$ if and only if it can be {\em produced} by the server model:
  \[\valid s t\quad\triangleq\quad\exists s',\behaves s t s'\]
  Here the producibility relation in \autoref{sec:concepts} is expanded with an
  argument $s'$ representing the post-transition state, pronounced
  ``specification $s$ can produce trace $t$ and step to specification $s'$~'':
  \begin{enumerate}
  \item A server model can produce an empty trace and step to itself:
    \[\behaves s \nil s\]
  \item A server model can produce a non-empty trace if it can produce the head
    of the trace and step to some server model that produces the tail of the
    trace:
    \[\behaves s {t+(q,a)} s_2\quad\triangleq\quad\exists s_1,c,\behaves s t s_1\wedge\triggers {s_1}c{(q,a)}s_2\]
  \end{enumerate}
\end{definition}

\subsection{Validating traces}
The validator takes a trace and determines whether it is valid per the protocol
specification.

\begin{definition}[Validator]
A validator is an infinite loop whose loop body takes a pair of query and
response and determines whether it is valid or not.  The validator iterates over
a state of some type $V$.  Given a $Q\times A$ pair, the loop body may return a
next validator state or return nothing, written as type ``$\option V$'':
\[\begin{array}{lll}
  \Validator&\triangleq&\sigT{V}{(Q\times A\to V\to\option V)\times V}\\
  \option X&\triangleq&\Some(x:X)\mid\None
\end{array}\]
\end{definition}

For example, a validator for the CMP-SET protocol is written as:
\begin{align*}
  \existT{V}{\Int}{(&\lam{(q,a)(v)}{
      \begin{cases}
        \If a\Is 1\Then\Some v\Else\None&q\le v\\
        \If a\Is 1\Then\Some q\Else\None&\mathrm{otherwise}
      \end{cases}
    }\\
    &,0)}
\end{align*}
Here the validator state is the same as the server model's.  The loop body computes
the expected response and compares it with the observed response.  If they are
the same, then the next server state is used as the next validator state.
Otherwise, the function returns $\None$, indicating that the response is
invalid.

Having defined the validator type, we can now instantiate the trace acceptation
notation ``$\accepts v t$'' in \autoref{def:tester} in terms of operational
semantics.

\begin{definition}[Validator transitions]
  \label{def:validator-step}
  Validator $v$ can consume request $q$ and response $a$ and step to $v'$
  (written ``$\behaves v{(q,a)}v'$~'') if and only if the post validator can be
  computed by the $\stepValidator$ function:
\begin{align*}
  &\behaves v{(q,a)}v'\quad\triangleq\quad\stepValidator(q,a)(v)=\Some v'\\
  &\stepValidator:Q\times A\to\Validator\to\option\Validator\\
  &\stepValidator(q,a)(\existT{V}{\beta}{(\vstep,state)})\\
  &\qquad\triangleq\begin{cases}
  \Some{(\existT{V}{\beta}{(\vstep,state')})} & \vstep(q,a,state)=\Some{state'} \\
  \None & \vstep(q,a,state)=\None
  \end{cases}
\end{align*}

The $\stepValidator$ function takes a query and a response, and computes the
validator transition by pattern matching on argument $(v:\Validator)$.  Let
$\beta$ be the validator state type of $v$, $\vstep$ be the loop body and
$(state:\beta)$ the current state of $v$.  Then $\stepValidator(q,a)(v)$ calls
the $\vstep$ function to validate the $Q\times A$ pair.  If the pair is valid,
then $\vstep$ returns a post-validation $state'$, which replaces the validator's
current $state$.  Otherwise, the validator halts with $\None$.
\end{definition}

\begin{definition}[Trace acceptance in QAC]
Validator $v$ accepts trace $t$ if and only if it {\em cosumes} the trace and
steps to some validator $v'$, written as ``$\behaves vtv'~$'':
\[\accepts v t\quad\triangleq\quad\exists v',\behaves v t v'\]
Here the consumability relation is defined as follows:
\begin{enumerate}
\item A validator can consume an empty trace and step to itself:
  \[\behaves v\nil v\]
\item A validator consumes a non-empty trace if it can consume the head of the
  trace, and step to some validator that consumes the tail of the trace:
  \[\behaves v {t+(q,a)} v_2\quad\triangleq\quad\exists v_1,\behaves v t v_1\wedge
  \behaves{v_1}{(q,a)}{v_2}\]
\end{enumerate}
\end{definition}

\section{Soundness and Completeness of Validators}
\label{sec:correctness}
We can now phrase the correctness properties in \autoref{sec:concepts} in terms
of the QAC language family:
\begin{enumerate}
  \item A rejection-sound validator consumes
    all traces that are producible by the protocol specification:
    \[\begin{array}{lrl}
      \rejSound v s&\triangleq&\forall t,\rejects v t\implies\invalid s t\\
      &\triangleq&\forall t,(\exists s',\behaves s t s')\implies\exists v',\behaves v t v'
    \end{array}\]
  \item A rejection-complete validator only
    consumes traces that are producible by the protocol specification:
    \[\begin{array}{lrl}
      \rejComplete v s&\triangleq&\forall t,\invalid s t\implies\rejects v t\\
      &\triangleq&\forall t,(\exists v',\behaves v t v')\implies\exists s',\behaves s t s'
    \end{array}\]
\end{enumerate}

Both the specification and the validator are infinite loops.  To show that the
validator consumes the same space of traces as the specification produces, we
need to show the correspondence between each server and validator step.  This is
done by introducing some loop invariant between the server and validator states
and showing that it is preserved by the server's and the validator's loop body.

This section shows how to prove that validator $\existT{V}{\beta}{(\vstep,v_0)}$
is sound and complete with respect to the server model
$\existT{S}{\sigma}{(\sstep,s_0)}$.  The core of the proof is the loop invariant
defined as a binary relation between the validator state $\beta$ and the server
state $\sigma$.  Notation ``$(\Reflects v s)$'' is pronounced ``validator state
$v$ simulates server state $s$''.

\subsection{Proving rejection soundness}
\label{sec:qac-soundness}
To prove that any trace producible by the server is consumable by the validator,
I perform forward induction on the server's execution path and show that every
step has a corresponding validator step:

\begin{itemize}
\item The initial server state $s_0$ simulates the initial validator state $v_0$:
  \begin{equation}
    \tag{RejSound-Init}
    \label{eq:rs1}
    \Reflects{v_0}{s_0}
  \end{equation}
\item Any server step $\triggers sc{(q,a)}s'$ whose pre-execution state $s$
  reflects some pre-validation state $v$ can be consumed by the validator
  yielding a post-validation state $v'$ that reflects the post-execution state
  $s'$:
  \begin{align*}
    \tag{RejSound-Step}
    \label{eq:rs2}
    &\forall(q:Q)(c:C)(a:A)(s,s':\sigma)(v:\beta),\\
    &\triggers sc{(q,a)}s'\;\wedge\;\Reflects{v}{s}\\
    &\implies\exists(v':\beta),\;\behaves v{(q,a)}v'\;\wedge\;\Reflects{v'}{s'}
  \end{align*}
  \begin{center}
    \includegraphics[width=.5\textwidth]{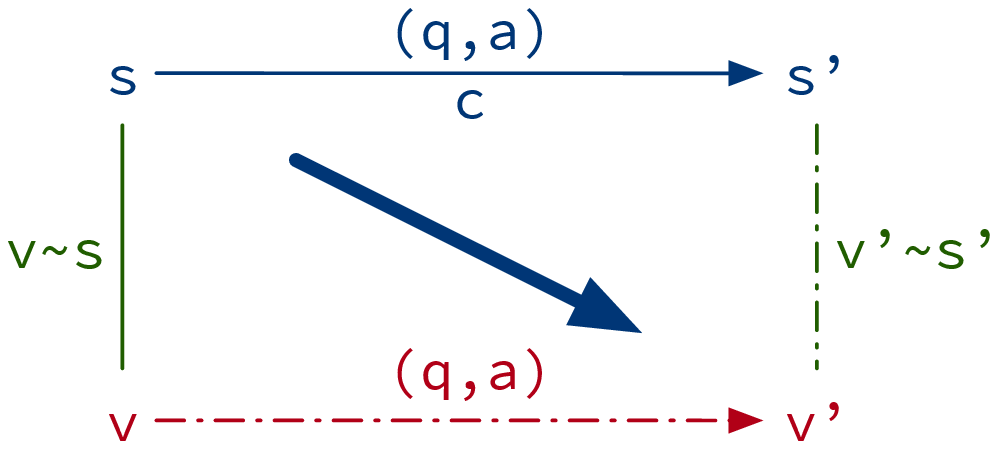}
  \end{center}
\end{itemize}

Here syntax ``$\triggers sc{(q,a)}s'$~'' and ``$\behaves v{(q,a)}v'$~'' are
simplified from \autoref{def:server-step} and \autoref{def:validator-step},
representing the server and validator instances by their states.  This is
because their state types $\sigma,\beta$ and step functions $\sstep,\vstep$
remain constant over the transitions.

\subsection{Proving rejection completeness}
\label{sec:qac-completeness}
Rejection completeness says that any trace consumable by the validator is
producible by the server model.  I construct the server's execution path
$\behaves sts'$ by {\em backward} induction of the validation path $\behaves
vtv'$:
\begin{itemize}
\item Any accepting validator step $\behaves v{(q,a)}v'$ has some server
  state $s'$ that reflects the post-validation state $v'$:
  \begin{align*}
    \tag{RejComplete-End}
    \label{eq:rc1}
    \forall(q:Q)(a:A)(v, v':\beta),\;&\behaves v{(q,a)}v'\\
    &\implies\exists s':\sigma,\Reflects{v'}{s'} 
  \end{align*}
  This gives us a final server state from which we can construct the server's
  execution path inductively.

\item Any accepting validator step $\behaves v{(q,a)}v'$ whose
  post-validation state $v'$ reflects some post-execution server state $s'$
  has a corresponding server step from a pre-execution state $s$
  that reflects the pre-validation state $v$:
  \begin{align*}
    \tag{RejComplete-Step}
    \label{eq:rc2}
    &\forall(q:Q)(a:A)(v,v':\beta)(s':\sigma),\\
    &\vstep(q,a,v)=\Some{v'}\wedge\Reflects{v'}{s'}\\
    &\implies\exists(s:\sigma)(c:C),\sstep(q,c,s)=(a,s')\wedge\Reflects{v}{s}
  \end{align*}
  \begin{center}
    \includegraphics[width=.5\textwidth]{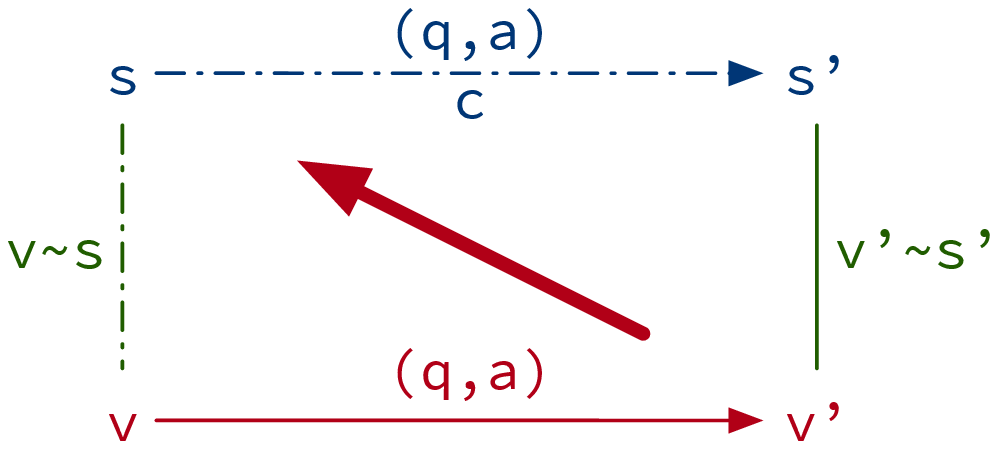}
  \end{center}
  This allows us to inducively construct an execution path starting from some
  server state that reflects the initial validator state $v_0$.  To show the
  existence of an execution path that starts from the initial server state
  $s_0$, we need the following hypothesis:

\item The initial validator state $v_0$ only reflects the initial server state $s_0$:
  \begin{equation}
    \tag{RejComplete-Init}
    \label{eq:rc3}
    \{s\mid\Reflects{v_0}{s}\}=\{s_0\}
  \end{equation}
\end{itemize}

Rejection soundness is proven by forward induction, while rejection completeness
is proven by backward induction.  This is because the choice $C$ is known from
the server step, but unknown from the validator step: Given a validator step, we
cannot predict ``what choices the server will make in the future'', but we can
analyze ``what choices the server might have made in the past''.  This proof
strategy is formalized in the Coq proof assistant and will be demonstrated with
an example in \autoref{sec:proof}.

\chpt{Synchronous Validator by Dualization}
\label{chap:dualize}
As discussed in \autoref{sec:internal-external-nondeterminism}, nondeterminism
makes validators difficult to write.  To address this challenge, I construct
validators {\em automatically} from their specifications.  The key idea is to
encode the specification with a programming language, and {\em dualize} the
specification program to derive a validator.

This chapter demonstrates the dualization technique with a programming language
in the QAC family.  \autoref{sec:encode-spec} introduces the $\Prog$ language
for encoding specifications.  Specifications written in $\Prog$ are dualized
into validators in \autoref{sec:dualize-prog}, with correctness proven in Coq
in \autoref{sec:proof}.

\section{Encoding Specifications}
\label{sec:encode-spec}
Constructing the validator automatically requires analyzing the computations of
the specification program.  The QAC language family in \autoref{sec:qac} only
exposes a state monad interface for server models, which is a function that
cannot be destructed within the meta language to perform program analysis.  This
section introduces a programming language for encoding specifications whose
structures can be analyzed.

For readability, I demonstrate the dualization technique on a subset of QAC
server models called integer machine models, featuring random-access memory (RAM)
and arithmetic operations.  To test real-world systems like web servers, I'll
employ a more complex specification language in \autoref{chap:practices}.

\subsection{Integer machine model}
The server state of an integer machine model is a key-value mapping that
resembles a RAM.  The addresses are natural numbers, and the data are integers.
The initial server state has zero data in all addresses:
\begin{align*}
  &s_0:\Nat\to\Int\\
  &s_0=(\_\mapsto0)\\
  &\ie\text{, }\forall (k:\Nat), s_0!k=0
\end{align*}
Here syntax ``$s!k$'' is pronounced ``data stored in address $!k$ of memory
$s$''.  I use ``$!k$'' to indicate that the natural number $k$ is being thought
of as an address.

The server's queries, responses, and choices ($Q$, $A$, $C$) are of type integer.
At the beginning of each server loop, the internal choice is written to address
$!0$, and the query is written to address $!1$.  The server then performs some
computation $f:(\Nat\to\Int)\to(\Nat\to\Int)$ that manipulates the memory, and
then sends back the value stored in address $!1$ as the response:
\begin{align*}
  \sstep_f(q,c)(s)\quad\triangleq\quad&\letin{s_1}{\update s0c}\\
  &\letin{s_2}{\update{s_1}1q}\\
  &\letin{s_3}{f(s_2)}\\
  &(s_3!1,s_3)
\end{align*}

Each memory-manipulating computation $f$ defines an instance of the integer
machine model:
\[\existT{S}{\Nat\to\Int}{(\sstep_f,s_0)}\]

Dualizing an integer machine model requires structural analysis of its memory
manipulation.  Next, I'll introduce a programming language to encode computations
$(\Nat\to\Int)\to(\Nat\to\Int)$.

\subsection{The $\Prog$ modeling language}
\label{sec:prog-lang}
\paragraph{Syntax}
A program in the $\Prog$ language may read and write at any address of the
memory, perform arithmetic operations, and make conditional branches:
\[\begin{array}{lrll}
\Prog&\triangleq&\Return&\text{end computation}\\
&\mid&!dst\coloneqq\Sexp;\Prog&\text{write to address }dst\in\Nat\\\null
&\mid&\If\Sexp\le\Sexp\Then\Prog\Else\Prog&\text{conditional branch}\\
\Sexp&\triangleq&\Int&\text{constant integer}\\
&\mid&!\src&\text{read from address }\src\in\Nat\\
&\mid&\Sexp\oop\Sexp&op\in\{+,-,\times,\div\}
\end{array}\]

For example, the following program computes the absolute value of data stored in
$!1$ and stores it in address $!1$:
\[\begin{array}{ll}
  \If&!1\le0\\
  \tthen&!1\gets(0-!1);\;\Return\\
  \eelse&\Return
\end{array}\]

\paragraph{Semantics}
Each program $(p:\Prog)$ specifies a computation on the integer machine:
\[\begin{array}{ll}
\multicolumn{2}{l}{\Eval\;:\;\Prog\to(\Nat\to\Int)\to(\Nat\to\Int)}\\
\Eval(p)(s)&\triangleq
\begin{cases}
  s&p\Is\Return\\
  \Eval(p')(\update{s}{\dst}{e^s})&p\Is !\dst\coloneqq e;p'\\
  \Eval(\If {e_1}^s\le{e_2}^s\Then p_1\Else p_2)(s)&p\Is\If e_1\le e_2\Then p_1\Else p_2
\end{cases}\\
e^s&\triangleq
\begin{cases}
  z\hphantom{\Exec(\If {e_1}^s\le{e_2}^s\Then p_1\Else p_2,)}&e\Is z:\Int\\
  s!\src&e\Is !\src\\
  {e_1}^s\oop{e_2}^s&e\Is e_1\oop e_2
\end{cases}
\end{array}\]

Here ``$e^s$'' is pronounced ``evaluating server expression $(e:\Sexp)$ on
memory $(s:\Nat\to\Int)$''.  It substitutes all occurences of ``$!\src$'' with
the data stored in address $!\src$ of memory $s$.

Syntax ``$s[k\mapsto v]$'' is pronounced ``writing value $v$ to address $!k$ of
memory $s$''.  It produces a new memory that maps address $!k$ to $v$, while
other addresses remain unchanged from $s$:
\[s[k\mapsto v]\quad\triangleq\quad k'\mapsto\begin{cases}v&k'=k\\
s!k'&k'\neq k\end{cases}\]

\paragraph{From $\Prog$ to server model}
Every program in the $\Prog$ language corresponds to a server model that
performs the computations it specifies:
\begin{align*}
  &\serverOf:\Prog\to\Server\\
  &\serverOf(p)\;\triangleq\;\existT{S}{\Nat\to\Int}{(\sstep_{\Eval(p)},s_0)}
\end{align*}

For example, the CMP-RST protocol in \autoref{sec:qac-model} can be constructed
by applying $\serverOf$ to the following program:
\[\begin{array}{ll}
\If !1\le!2\Then !1\coloneqq0;\Return&(1)\\
\eelse!1\coloneqq1;!2\coloneqq!0;\Return&(2)
\end{array}\]
The constructed server stores its data $n$ in address $!2$.  When the query
stored in $!1$ is less than or equal to $!2$ (case 1), the server writes $0$ to
address $!1$ as response and leaves the data untouched in address $!2$.  For
queries greater than $!2$ (case 2), the server writes $1$ as response, and
updates the data in $!2$ with the internal choice stored in $!0$.

Based on specifications written in the $\Prog$ language, we can now construct
the validator automatically by dualization.

\section{Dualizing Specification Programs}
\label{sec:dualize-prog}
This section presents an algorithm that constructs a validator from the
specification program:
\[\validatorOf:\Prog\to\Validator\]

For every program $p$, $\validatorOf(p)$ determines whether a trace is
producible by $\serverOf(p)$:
\begin{align*}
  &\forall (p:\Prog)(t:\List(Q\times A),\\
  &(\behaves{\validatorOf(p)}t)\;\iff\;(\behaves{\serverOf(p)}t)
\end{align*}

More specifically, given a trace of $Q\times A$ pairs, the validator determines
whether there exists a sequences of internal choices $C$ that explains how the
server produces the trace in \autoref{def:trace-validity}.

The idea is similar to the \ilc{tester} in \autoref{sec:interactive-testing},
which \ilc{validate}s the trace by executing the \ilc{serverSpec} and comparing
the expected response against the tester's observations.

However, executing a nondeterministic specification does not produce a specific
expectation of response, but a nondeterministic response that depends on the
internal choice.  Therefore, upon observing a response $A$, the validator should
determine whether there is a choice $C$ that leads the specification to produce
this response.

This reduces the trace validation problem to constraint solving.  The validator
maintains a set of constraints that require the responses observed from the
implementation to be explainable by the specification.

More specifically, the validator executes the $\Prog$ model and represents
internal choices with {\em symbolic variables}.  These variables are carried
along the program execution, so the expected responses are computed as {\em
  symbolic expressions} that might depend on those variables.  The validator
then constrains that the symbolic response is equal to the concrete observation.

To achieve this goal, the validator stores a symbolic variable for each address
of the server model.  It also remembers all the constraints added as
observations are made during testing.  This information is called a ``validation
state'':
\[(\Nat\to\Var)\times\Set\constraint\]
Here the $\constraint$s are relations between validator expressions ($\Vexp$s)
that may depend on symbolic $\Var$iables:
\[\begin{array}{lrl@{\qquad}l}
\constraint&\triangleq&\Vexp\ccmp\Vexp&cmp\in\{<,\le,=\}\\
\Vexp&\triangleq&\Int&\text{constant integer}\\
&\mid&\#x&\text{variable }x\in\Var\\
&\mid&\Vexp\oop\Vexp&op\in\{+,-,\times,\div\}\\
\end{array}\]
The $\Vexp$ type replaces $\Sexp$'s address constructor $!\src$ with variable
constructor $\#x$.  This allows the validator to constrain the values of the
same address at different times \eg the internal choice stored in $!0$ in
different iterations.  The validation state maps each address to its current
representing variable $(\Nat\to\Var)$, which updates as the validator
symbolically executes the server program.

Note that the server program in \autoref{sec:prog-lang} has conditional
branches.  When executing the specification program, the branch condition might
depend on the internal choices, which is invisible to the validator.  As a
result, the validator doesn't know the exact branch taken by the specification,
so it maintains multiple validation states, one for each possible execution
path:
\[\Set((\Nat\to\Var)\times\Set\constraint)\]

The initial state of the validator is a single validation state that corresponds
to the specification's initial state:
\[\{(\_\mapsto\#0,\{\#0=0\})\}\]
Here the initial validation state says, ``all addresses are mapped to variable
$\#0$, and the value of variable $\#0$ is constrained to be zero''.  This
reflects the initial server state that maps all addresses to zero value.

The validator's loop body is derived by analyzing the computations of the server
model.

\begin{enumerate}
\item \label{rule:write} When the server performs a write operation
  $!\dst\coloneqq \mathit{exp}$, the validator creates a fresh variable $x$ to
  represent the new value stored in address $!\dst$, and adds a constraint that
  says $x$'s value is equal to that of $\mathit{exp}$.  This rule also applies
  to writing the request to address $!1$ before executing the program.

  For example, if the server performs $!2\gets!0$, then the validator should
  step from $(\vs,\cs)$ to $(\update{\vs}{2}{x},\cs\cup\{\#x=\#(\vs!0)\})$,
  where $\#x$ is a fresh variable.  The new validation state says the value
  stored in address $!2$ is represented by variable $\#x$.  It also constraints
  that $\#x$ should be equal to $\#(\vs!0)$, the variable that represents
  address $!0$ in the pre-validation state $\vs$.
\item \label{rule:branch} When the server makes a nondeterministic branch $\If
  e_1\le e_2\Then p_1\Else p_2$, the validator considers both cases: (a) if
  $p_1$ was taken, then the validator should add a constraint $e_1\le e_2$; or
  (b) if $p_2$ was taken, then the validator should add constraint $e_2<e_1$.
\item \label{rule:choice} Before executing the program, the server writes the internal
    choice $c$ to address $!0$.  Accordingly, the validator creates a fresh
    variable to represent the new value stored in address $!0$, without adding
    any constraint.
\item \label{rule:return} After executing the program, the server sends back the
  value stored in $!1$ as response.  Accordingly, the validator adds a
  constraint that says the variable representing address $!1$ is equal to the
  observed response.
\item \label{rule:unsat} When the constraints of a validation state become
  unsatisfiable, it indicates that the server model cannot explain the
  observation.  This is because either (i) the observation is invalid, \ie,
  not producible by the server model, or (ii) the observation is valid, but was
  produced by a different execution path of the server model.

\item \label{rule:reject} If the set of validation states becomes empty, it
  indicates that the observation cannot be explained by any execution path of
  the specification, so the validator should reject the trace.
\end{enumerate}

\begin{figure}
\[\begin{array}{l@{\;}r@{\;}l}
\vstep_p(q,a)(v)&\triangleq&\letin{v'}{v_0\gets v;\vstep'_p(q,a)(v_0)}\\
&&\If v'\Is\varnothing\Then\None\Else\Some v'\hfill(\ref{rule:reject})\\
\vstep'_p(q,a)(v_0)&\triangleq&\letin{v_1}{\Havoc(0,v_0)}\\
&&\letin{v_2}{\Write(1,q,v_1)}\\
&&(\vs_3,\cs_3)\gets\Exec(p,v_2);\\
&&\letin{\cs_4}{\cs_3\cup\{\#(\vs_3!1)= a\}}\hfill(\ref{rule:return})\\
&&\If\solvable \cs_4\Then\{(\vs_3,\cs_4)\}\Else\varnothing\hfill(\ref{rule:unsat})\\
\Exec(p,(\vs,\cs))&\triangleq&\begin{cases}
  \{(\vs,\cs)\}&\text{if }p\Is\Return\\
  \Exec(p',\Write(d,e,(\vs,\cs)))&\text{if }p\Is(!d\coloneqq e;p')\\
  \left(\begin{array}{@{}l}
    \letin{v_1}{(\vs,\cs\cup\{{e_1}^{\vs}\le{e_2}^{\vs}\})}\\
    \letin{v_2}{(\vs,\cs\cup\{{e_2}^{\vs}<{e_1}^{\vs}\})}\\
    \Exec(p_1,v_1)\cup\Exec(p_2,v_2)\hfill(\ref{rule:branch})
  \end{array}\right)&\begin{array}{@{}l@{}l}\text{if }&p\Is\\
    &(\If e_1\le e_2\\
    &\tthen p_1\Else p_2)\end{array}
\end{cases}\\
\Write(d,e,(\vs,\cs))&\triangleq&\letin{x_e}{\Fresh (\vs,\cs)}\hfill(\ref{rule:write})\\
&&(\update{\vs}{d}{x_e},\cs\cup\{\#x_e=e^{\vs}\})\\
\Havoc(d,(\vs,\cs))&\triangleq&\letin{x_c}{\Fresh (\vs,\cs)}(\update{\vs}{d}{x_c},\cs)\hfill(\ref{rule:choice})
\end{array}\]
\caption[Dualizing server model into validator.]{Dualizing server model into
  validator, with derivation rules annotated.}
\label{fig:dualize}
\end{figure}

This mechanism is formalized in \autoref{fig:dualize}.  Here the notation
``$v_0\gets v;\vstep'_p(q,a)(v_0)$'' is a monadic bind for sets: Let $\vstep'_p$
map each element $v_0$ in $v$ to a set of validation states
$(\vstep'_p(q,a)(v_0):\Set((\Nat\to\Var)\times\Set\constraint))$, and return the
union of all result sets as $v'$.

The validator adds constraints in three circumstances: \autoref{rule:write} says
the write operation updates the destination with the source expression;
\autoref{rule:branch} guards the branch condition to match its corresponding
execution path; \autoref{rule:return} unifies the server's symbolic response
against the concrete response observed from the implementation.

The constraints added in \autoref{rule:write} and \autoref{rule:branch} are
translated from the specification program.  Given a server expression
$(e:\Sexp)$ from the source expression or the branch condition, syntax $e^\vs$
translates it into a validator expression $\Vexp$ using the validation state
$\vs$, by replacing its addresses with the corresponding variables:
\[e^{\vs}\triangleq\begin{cases}
  n&e\Is z:\Int\\
  \#(\vs!\src)&e\Is!\src\\
  {e_1}^{\vs}\oop{e_2}^{\vs}&e\Is e_1\oop e_2
\end{cases}\]

The validator assumes a constraint solver that can determine whether a set of
constraints is satisfiable, \ie, whether there exists an {\em assignment}
of variables $(\Var\to\Int)$ that satisfies all the constraints:
\begin{gather*}
  \forall \cs,\solvable \cs\iff\exists (\asgn:\Var\to\Int),\satisfy{\asgn}\cs\\
  \begin{array}{r@{\;}l}
    \satisfy{\asgn}\cs\triangleq&\forall(e_1\ccmp e_2)\in \cs, {e_1}^{\asgn}\ccmp{e_2}^{\asgn}\\
    e^{\asgn}\triangleq&\begin{cases}
      z&e\Is z:\Int\\
      \asgn!x&e\Is \#x\\
      {e_1}^{\asgn}\oop{e_2}^{\asgn}&e\Is e_1\oop e_2
    \end{cases}
  \end{array}
\end{gather*}
Here ``$e^{\asgn}$'' is pronounced ``evaluating validator expression $(e:\Vexp)$
with assignment $(\asgn:\Var\to\Int)$''.  It substitutes all occurences of
``$\#x$'' with their assigned value $(\asgn!x)$.

Now we have the algorithm that constructs the validator from the specification
program $p$:
\begin{align*}\validatorOf(p)\;\triangleq\quad&\existT{V}{\Set((\Nat\to\Var)\times\Set\constraint)}\\
  &(\vstep_p,\{(\_\mapsto\#0,\{\#0=0\})\})
\end{align*}

\begin{figure}
\begin{align*}
&\existT{V}{\Set((\Nat\to\Var)\times\Set\constraint)}\\
  &\begin{array}{rll}
     (&\lam{(q,a)(v)}{&\llet v'=\begin{array}[t]{@{}l@{}l@{}ll}
       \multicolumn{3}{@{}l}{(vs_0,cs_0)\gets v;}\\
       \letin{vs_1&}{\update{vs_0}{0}{\Fresh (vs_0,cs_0)}&}&\text{(1)}\\
       \letin{x_q&}{\Fresh (vs_1,cs_0)&}\\
       \letin{vs_2&}{\update{vs_1}{1}{x_q}&}\\
       \letin{cs_2&}{cs_0\cup\{\#x_q= q\}&}\\
       \letin{cs_{3a0}&}{cs_2\cup\{\#(vs_2!1)\le\#(vs_2!2))\}&}&\text{(2a)}\\
       \letin{x_{3a1}&}{\Fresh (vs_2,cs_{3a0})&}\\
       \letin{vs_{3a1}&}{\update{vs_2}{1}{x_{3a1}}&}\\
       \letin{cs_{3a1}&}{cs_{3a0}\cup\{\#x_{3a1}=0\}&}\\
       \letin{cs_{3b0}&}{cs_2\cup\{\#(vs_2!2)<\#(vs_2!1)\}&}&\text{(2b)}\\
       \letin{x_{3b1}&}{\Fresh (vs_2,cs_{3b0})&}\\
       \letin{vs_{3b1}&}{\update{vs_2}{1}{x_{3b1}}&}\\
       \letin{cs_{3b1}&}{cs_{3b0}\cup\{\#x_{3b1}=1\}&}\\
       \letin{x_{3b2}&}{\Fresh (vs_{3b1},cs_{3b1})&}\\
       \letin{vs_{3b2}&}{\update{vs_{3b1}}{2}{x_{3b2}}&}\\
       \letin{cs_{3b2}&}{cs_{3b1}\cup\{\#x_{3b2}=\#(vs_{3b2}!1)\}&}\\
       \multicolumn{3}{@{}l}{((vs_4,cs_4)\gets\{(vs_{3a1},cs_{3a1}),(vs_{3b2},cs_{3b2})\};}&\text{(3)}\\
       \letin{cs_5&}{cs_4\cup\{\#(vs_4!1)= a\}&}\\
       \multicolumn{3}{@{}l}{\If\solvable cs_5\Then\{(vs_4,cs_5)\}\Else\varnothing}\\
       \end{array}\\
       &&\iin\\
       &&\If v'\Is\varnothing\Then\None\Else\Some v'
     }\\
     ,&\multicolumn{2}{l}{\{(\_\mapsto\#0,\{\#0=0\})\}})
   \end{array}
\end{align*}
\caption[Validator for protocol CMP-RST.]{Validator for CMP-RST automatically
  derived from its specification in $\Prog$.  This program consists of three
  parts: (1) symbolizing the query and internal choice before executing the
  model, (2) considering both branches in the model program, propagating a
  validation state for each branch, and (3) filtering the validation states by
  constraint satisfiability, removing invalid states.}
\label{fig:derived-validator}
\end{figure}

For example, to construct a validator for the CMP-RST protocol in
\autoref{sec:qac-model}, we first specify it in $\Prog$ as:
\begin{align*}
  &\If!1\le!2\\
  &\tthen!1\gets0;\Return\\
  &\eelse!1\gets1;!2\gets!0;\Return
\end{align*}
This program stores the data $n$ in address $!2$.  If the request is less than
or equal to $n$, then the program writes response $0$ to address $!1$, and
leaves the data unchanged; Otherwise, it writes $1$ as response, and updates
address $!2$ with the internal choice in $!0$.

We then apply function $\validatorOf$ to this $\Prog$-based specification,
resulting in a validator as shown in \autoref{fig:derived-validator}.
Validators constructed in this way are proven correct in the next
section.

\section{Correctness Proof}
\label{sec:proof}
So far, I have introduced the $\Prog$ language for writing specifications and
shown how to construct a validator from it.  This section shows how to prove the
soundness and completeness of all validators dualized from $\Prog$-based
specifications:
\[\begin{array}{r@{\;}l}
\forall p:\Prog,&\letin{s}{\serverOf(p)}\\
&\letin{v}{\validatorOf(p)}\\
&\rejSound v s\wedge\rejComplete v s\\
&\ie\text{, }\forall t:\List(Q\times A),\\
&\qquad\valid s t\iff\accepts v t\\
&\qquad\ie\text{, }\exists s',\behaves s t s'\iff\exists v',\behaves v t v'
\end{array}\]

To apply the proof techniques in \autoref{sec:correctness}, I design the loop
invariant in \autoref{sec:proof-invariant}.  I then use the invariant to prove
the hypotheses for soundness and completeness in \autoref{sec:proof-sound}
and \autoref{sec:proof-complete}, respectively.  The entire proof is formalized
in the Coq proof assistant.

\subsection{Designing the loop invariant}
\label{sec:proof-invariant}
Let $\beta=\Set((\Nat\to\Var)\times\Set\constraint)$ be the validator state type
and $\sigma=\Nat\to\Int$ the server state type.  Then the loop invariant
$\Reflects{(v:\beta)}{(s:\sigma)}$ is a binary relation between the validator
state $v$ and the server state $s$.  To define this relation, I first discuss
the semantics of validator states.

Each validation state in the validator state consists of an address-variable
mapping and a set of constraints over the variables.  The validation state
defines a space of server states.  A validator accepts a trace if it has a
validation state whose constraints are satisfiable, which indicates the
existence of a server that produces the trace.  A corresponding server state can
be constructed from any assignment that satisfies the validation state's
constraints:
\[\vs^{\asgn}\quad\triangleq\quad \mathit{addr}\mapsto \asgn!(\vs!\mathit{addr})\]
Let $(\vs:\Nat\to\Var)$ be a mapping in the validation state,
$(\asgn:\Var\to\Int)$ be an assignment of variables, then
$(\vs^\asgn:\Nat\to\Int)$ is a server state that maps each address
$!\mathit{addr}$ to the value assigned by $\asgn$ to the address' corresponding
variable in $\vs$.

\begin{definition}[Loop invariant for $\Prog$-based validators]
Validator state $v$ reflects server state $s$ (written ``$\Reflects vs$'') if
the server state lies within the space of servers defined by some validation
state $(\vs,cs)\in v$.  That is, there exists an assignment that relates the
server state and the validation state:
\begin{align*}
\Reflects vs\quad\triangleq\quad&\exists(\vs,\cs)\in
  v,\;\exists\asgn,\\
  &\satisfy{\asgn} \cs\;\wedge\;\vs^{\asgn}=s
  \end{align*}
\end{definition}

Having defined the loop invariant, we only need to instantiate the hypotheses
in \autoref{sec:qac-soundness}--\ref{sec:qac-completeness} with $\Prog$-based
definitions.  The rest of this section proves the hypotheses for establishing
validators' soundness and completeness.

\subsection{Proving rejection soundness}
\label{sec:proof-sound}
\begin{lemma}[\ref{eq:rs1}]
\[\begin{array}{ll}
\text{If:}&
\vs=(\_\mapsto\#0)\qquad
\cs=\{\#0=0\}\qquad
s=(\_\mapsto0)\\
\text{Then:}&\Reflects{\{(\vs,\cs)\}}{s}
\end{array}\]
\end{lemma}
\begin{proof}
Since $(\vs,\cs)$ is the only element in the validator state, we only need to show
that:
\[\exists(\asgn:\Var\to\Int),\quad\satisfy{\asgn} \cs\;\wedge\;\vs^{\asgn}=s\]

By constructing the assignment as: \(\asgn=(\_\mapsto0)\)

We have: \(\#0^{\asgn}=0\)

Thus: \(\satisfy{\asgn} \cs\)

We also know that: \[\forall k, \quad\asgn!(\vs!k)=0=(s!k)\]

Therefore: \(\vs^{\asgn}=s\)
\end{proof}

\begin{lemma}[\ref{eq:rs2}]
  \begin{align*}
    &\forall(q,c,a:\Int)(s,s':\sigma)(v:\beta),\\
    &\triggers sc{(q,a)}s'\;\wedge\;\Reflects{v}{s}\\
    &\implies\exists v':\beta,\;\behaves v{(q,a)}v'\;\wedge\;\Reflects{v'}{s'}
  \end{align*}
\begin{proof}
The invariant $\Reflects{v}{s}$ tells us that $v$ contains a pre-validation
state that reflects the pre-step server state $s$:
\[\exists(\vs,\cs)\in v,\;\exists asgn,\quad\satisfy{\asgn} \cs\;\wedge\;{\vs}^{\asgn}=s\]

The $\vstep'_p$ function in \autoref{fig:dualize} leads to a set of
post-validation states.  We need to show that some state in it reflects the
post-step server state $s'$:
\[\exists(\vs',\cs')\in\vstep_p'(q,a)(\vs,\cs),\;\exists\asgn',\quad\satisfy{\asgn'}\cs'\;\wedge\;{\vs'}^{~\asgn'}=s'\]

The server's internal choice $c$ was provided, so we know the server's entire
execution path.  By induction on the $\Prog$ syntax, we can choose the right
post-validation state $(\vs',\cs')$ by computing each branch condition, and
deduce the satisfying assignment $\asgn'$ by computing each write operation's
source value and destination variable.
\end{proof}
\end{lemma}

\subsection{Proving rejection completeness}
\label{sec:proof-complete}

\begin{lemma}[\ref{eq:rc1}]
\begin{align*}
\forall(q,a:\Int)(v,v':\beta),\;&\behaves v{(q,a)}v'\\
&\implies\exists s':\sigma,\;\Reflects{v'}{s'}
\end{align*}
\begin{proof}
Since the $\vstep_p$ function in \autoref{fig:dualize} checks the nonemptiness
of the result, we know that $v'$ must be nonempty.  Consider validation state
$(\vs',\cs')\in v'$.  Since $\vstep'_p$ checks that $(\solvable \cs')$, we know
that:
\[\exists \asgn,\quad\satisfy{\asgn}cs'\]

Let:
\(s'=\vs'^{~\asgn}\)

Then we have:
\begin{align*}
&(\vs',\cs')\in v'\quad\wedge\quad\satisfy{\asgn}{\cs'}\quad\wedge\quad\vs'^{~\asgn}=s'\\
&\ie\text{, }\Reflects{v'}{s'}\tag*{\qedhere}
\end{align*}
\end{proof}
\end{lemma}

\begin{lemma}[\ref{eq:rc2}]
\begin{align*}
&\forall(q,a:\Int)(v,v':\beta)(s':\sigma),\\
&\behaves v{(q,a)}v'\;\wedge\;\Reflects{v'}{s'}\\
&\implies\exists(s:\sigma)(c:\Int),\;\triggers sc{(q,a)}s'\;\wedge\;\Reflects{v}{s}
\end{align*}
\begin{proof}
We first construct the pre-step server state $(s:\sigma\mid\Reflects{v}{s})$.
We then compute the internal choice $c$ and prove the server step $\triggers
sc{(q,a)}s'$.

The definition of $\Reflects{v'}{s'}$ says:
\[\exists (\vs',\cs')\in v',\;\exists \asgn,\quad \satisfy{\asgn}{\cs'}\;\wedge\;\vs'^{~\asgn}=s'\]

From the definition of $\vstep_p$, we know that:
\[\exists (\vs,\cs)\in v,\quad (\vs',\cs')\in\vstep'_p(q,a)(\vs,\cs)\]

Since $\vstep'_p$ monotonically increases the set of constraints, we have
$\cs\subseteq \cs'$.  Therefore: \[\satisfy{\asgn}{\cs}\]

Let: \(s=\vs^{\asgn}\)

Then we have:
\begin{align*}
&(\vs,\cs)\in v\quad\wedge\quad\satisfy{\asgn}{\cs}\quad\wedge\quad \vs^{\asgn}=s\\
&\ie\text{, }\Reflects{v}{s}
\end{align*}

From the definition of $\vstep'_p$, the validator first creates a fresh variable
to represent the server's internal choice, so we can deduce the internal choice
from the assignment:
\[x_c=\Fresh(\vs,\cs)\qquad c=\asgn!x_c\]

Now we need to show that the server's loop body $\sstep_p(q,c)(s)$ results in
response $a$ and post-execution state $s'$.  Since the post-validation state
$v'$ simulates $s'$ and guarantees the response to be $a$, we only need to prove
the post-execution state to be $s'$.

We have known that:
\[s=\vs^\asgn\qquad s'=\vs'^{~\asgn}\qquad\behaves{(\vs,\cs)}{(q,a)}{(\vs',\cs')}\]

Therefore, we can prove $\triggers sc{(q,a)}s'$ by induction on the $\Prog$
syntax, showing that every write or branch operation preserves $\asgn$'s ability
to unify the server state with the validation state.  This leads the
post-execution state to be unifiable against $\vs'$ using $\asgn$, thus to be
$s'$.
\end{proof}

The core of this proof is to find the internal choice $c$ for server step
$\triggers sc{(q,a)}s'$.  The choice was computed with the assignment deduced
from the loop invariant.  The assignment contains the value of all symbolic
variables, which includes all choices made by the server, past and future.  The
loop invariant requires the assignment to explain the past choices but cannot
predict future choices.  Therefore, we can only infer the server step from the
validator step using backward induction.
\end{lemma}

\begin{lemma}[\ref{eq:rc3}]
\[\begin{array}{ll}
\text{If:}&
\vs=(\_\mapsto\#0)\qquad
\cs=\{\#0\equiv0\}\qquad
s_0=(\_\mapsto0)\\
\text{Then:}&\{s\mid\Reflects{\{(\vs,\cs)\}}{s}\}=\{s_0\}
\end{array}\]
\begin{proof}
The requirement for $s$ says:
\[\exists \asgn:\Var\to\Int,\quad\satisfy{\asgn}{\cs}\quad\wedge\quad \vs^{\asgn}=s\]

The constraint satisfaction tells us that:
\(\asgn!0=0\)

We then have:
\[\forall k:\Nat,\quad s!k=\asgn!(vs!k)=\asgn!0=0=s_0!k\]

Therefore, $s_0$ is the only server state that $(\vs,\cs)$ simulates.
\end{proof}
\end{lemma}

Now we have proven that all $\Prog$-based validators satisfy the hypotheses
defined in \autoref{sec:qac-soundness}--\ref{sec:qac-completeness}, and we can
conclude that the validators constructed by dualization are sound and complete.
Next, I'll show how to apply this dualization technique to test real-world
programs.

\chpt{Asynchronous Tester by Dualization}
\label{chap:practices}
So far I've introduced the theory of validating synchronous interactions using
the QAC language family, and shown how to construct validators by dualization
with a simple $\Prog$ language.

However, in real-world testing practices, there are more problems to consider.
For example: How to interact with the SUT via multiple channels?  How to handle
external nondeterminism?

As discussed in \autoref{sec:intro-external-nondet}, a networked server's
response may be delayed by the network environment, and an asynchronous tester
may send other requests rather than waiting for the response.  Therefore, we
cannot view the trace as a sequence of $Q\times A$ pairs like we did
in \autoref{def:trace-validity}, and the state monad in the QAC language family
becomes insufficient for defining the space of asynchronous interactions.

This chapter applies the idea of dualization to testing asynchronous systems.  I
transition from the QAC language family to the ITree specification language, a
data structure for modeling programs' asynchronous interactions in the Coq
proof assistant.  ITree provides more expressiveness than QAC and allows
specifying the external nondeterminism caused by the network environment.  The
ITree-based specifications are derived into tester programs that can interact
with the SUT and reveal potential defects.

\begin{figure}[t]
  \includegraphics[width=\linewidth]{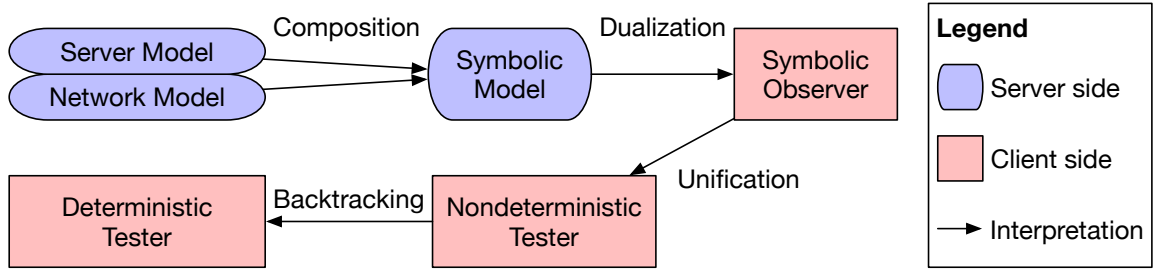}
  \caption{Deriving tester program from specification.}
  \label{fig:framework}
\end{figure}

\autoref{fig:framework} illustrates the derivation framework from ITrees to
testers.  \autoref{sec:itree} introduces the ITree language that encodes each
box in the framework.  \autoref{sec:internal-nondet} and
\autoref{sec:external-nondet} address internal and external nondeterminism in
the ITree context and interprets the ``server model'' into a ``nondeterministic
tester''.  \autoref{sec:backtrack} then explains how to execute the
nondterministic tester model as an interactive tester program that runs on
deterministic machines.

\section{From QAC to Interaction Trees}
\label{sec:itree}
To write specifications for protocols' rich semantics, I employed ``interaction
trees'' (ITrees), a generic data structure for representing interactive programs
in the Coq programming language, introduced by \citet{itree}.  I provide a brief
introduction to the ITree specification language in \autoref{sec:itree-lang}.

ITrees allow specifying protocols as monadic programs that model valid
implementations' possible behavior.  In \autoref{sec:qac-itree}, I show how to
embed specifications in the QAC family in terms of interaction trees.

The derivation from server specifications to interactive testers is by {\em
interpreting} ITree programs, which corresponds to the arrows
in \autoref{fig:framework}.  In \autoref{sec:interp}, I explain the mechanism of
interpretation by deriving testers from deterministic server models.  The rest
of this chapter then shows how to handle nondeterminism when interpreting
ITrees.

\subsection{Language definition}
\label{sec:itree-lang}
Consider an echo program, which keeps reading some data and writing it out
verbatim, until reaching EOF.  We can represent the program in Coq using the
ITrees datatype (left), with a reference in C (right) as:
\begin{multicols}{2}
\begin{coq}
  CoFixpoint echo :=
    c <- getchar;;
    if c is EOF
    then EXIT
    else
      putchar c;;
      echo.
\end{coq}
\columnbreak
\begin{cpp}
  void echo() {
    const char c = getchar();
    if (c == EOF)
      return;
    else {
      putchar(c);
      echo();
    }
\end{cpp}
\end{multicols}
Here the behavior after \ilc{getchar} depends on the value actually read.  The
monadic computation in Coq can be desugared into the following code, which uses
the \ilc{Bind} constructor to represent sequential composition of programs:
\begin{coq}
  CoFixpoint echo :=
    Bind getchar
         (fun c => if c is EOF
                 then EXIT
                 else Bind (putchar c)
                           (fun _ => echo)
         ).
\end{coq}
Here, the first argument of \ilc{Bind} is program that returns some value and
the second argument is a {\em continuation} that represents the subsequent
computation that depends on the value returned by the first argument.  Such
continuation-passing style can be represented as a tree of interactions.  To
help readers better understand the interaction tree language, I first provide a
simplified version of it that better shows its tree structure, and then explain
the actual type definition used in practice.

\begin{figure}
\begin{coq}
  CoInductive itreeM (E: Type -> Type) (R: Type) :=
    Ret     : R   -> itreeM E R
  | Trigger : E R -> itreeM E R
  | Bind    : forall {X : Type}, itreeM E X -> (X -> itreeM E R) -> itreeM E R.
\end{coq}
\vspace*{1em}
\caption{Mock definition of interaction trees.}
\label{fig:mock-itree}
\end{figure}

\begin{figure}
\vspace*{1em}
  \includegraphics[width=.5\linewidth]{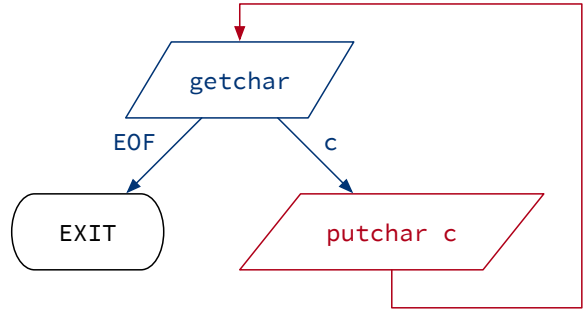}
  \caption{Interaction tree for echo program.}
  \label{fig:echo-itree}
\end{figure}

\paragraph{Mock interaction trees}
As shown in \autoref{fig:mock-itree}, a mock interaction tree (\ilc{itreeM}) has
two kinds of leaves, \ilc{Ret} and \ilc{Trigger}, and has internal nodes
constructed by \ilc{Bind}:
\begin{itemize}
\item \ilc{(Ret r)} represents a pure computation that yields a value \ilc r.
  In the echo example, \ilc{EXIT} halts the program with return value zero:
\begin{coq}
  Definition EXIT := Ret 0.
\end{coq}

\item \ilc{(Trigger e)} performs an impure event \ilc e and returns its result.
  Here \ilc{(e: E R)} is an event whose result is of type \ilc R.  For example,
  \ilc{getchar} has result type \ilc{char}, and \ilc{putchar}'s result type is
  \ilc{unit} (which corresponds to \inlinec{void} in C/C++, or \ilc{()} in
  Haskell).  These effective programs are constructed by triggering standard I/O
  events:
\begin{coq}
  Variant stdioE: Type -> Type := (* event type *)
    GetChar:         stdioE char
  | PutChar: char -> stdioE unit.
  
  Definition getchar           : itreeM stdioE char := Trigger  GetChar.
  Definition putchar (c: char) : itreeM stdioE unit := Trigger (PutChar c).
\end{coq}
\item \ilc{(Bind m k)} binds the return value of \ilc m to the continuation
  function \ilc k.  It first runs program \ilc m until it returns some value of
  type \ilc X.  The return value \ilc{(x: X)} then instantiates \ilc k into the
  following computation \ilc{(k x: itreeM E R)}.  This corresponds to the
  monadic \ilc{(;;)} syntax:
\begin{coq}
  Notation "x <- m1;; m2" := (Bind m1 (fun x => m2)).
  Notation "m1;; m2"      := (Bind m1 (fun _ => m2)).
\end{coq}

As illustrated in \autoref{fig:echo-itree}, each possible return value \ilc x is
an edge that leads to the child it instantiates, \ie, \ilc{(k x)}.  In this
way, the \ilc{Ret} and \ilc{Trigger} nodes are connected into a tree
structure.
\end{itemize}

The \ilc{CoInductive} keyword indicates that type \ilc{itreeM} can represent
infinite data.  For example, the aforementioned \ilc{echo} program (also
illustrated in \autoref{fig:echo-itree}) is an instance of \ilc{(itreeM stdioE
Z)} defined with keyword \ilc{CoFixpoint}, where \ilc Z indicates that the
program's result value (if it ever returns by \ilc{EXIT}) is of type integer.

The mock interaction tree provides an intuitive continuation-passing structure
for representing impure programs.  However, to implement dualization
effectively, we need to revise the language definition of monadic binding.

\paragraph{Practical interaction trees}

\begin{figure}
\begin{coq}
  CoInductive itree (E: Type -> Type) (R: Type) :=
    Pure   : R -> itree E R
  | Impure : forall {X : Type}, E X -> (X -> itree E R) -> itree E R.
\end{coq}
\caption{Formal definition of interaction trees (simplified).}
\label{fig:itrees}
\end{figure}

Instead of binding a program---which might have many events---to a continuation,
ITree binds a single impure event, as shown in \autoref{fig:itrees}.  I
use \ilc{(Impure e k)} to replace \ilc{(Bind (Trigger e) k)} representations
in \ilc{itreeM}.  A \ilc{Pure} computation cannot be bound to a continuation and
must be the leaf of an ITree.\footnote{For readability, the ``practical'' ITree
definition is a simplified version from \citet{itree}.  Here \ilc{Pure}
and \ilc{Impure} correspond to the \ilc{Ret} and \ilc{Vis} constructors.  I
dropped the \ilc{(Tau : itree E R -> itree E R)} constructor, which carries no
semantics and is used for satisfying Coq's guardedness condition.}

The \ilc{Ret}, \ilc{Trigger}, and \ilc{Bind} constructors introduced in
\ilc{itreeM} have equivalent representations in \ilc{itree}, so we can still
write programs in the monadic syntax:
\begin{coq}
  Definition ret {E R} : R -> itree E R := Pure.
  
  Definition trigger {E R} (e: E R) : itree E R := Impure e Pure.

  CoFixpoint bind {X E R} (m: itree E X) (f: X -> itree E R) : itree E R :=
    match m with
    | Pure   x   => f x
    | Impure e k => Impure e (fun r => bind (k r) f)
    end.

  Notation "x <- m1;; m2" := (bind m1 (fun x => m2)).
  Notation "m1;; m2"      := (bind m1 (fun _ => m2)).

  CoFixpoint translateM {E R} (m: itreeM E R) : itree E R :=
    match m with
    | Ret     r => ret r
    | Trigger e => trigger e
    | Bind m1 k => x <- translateM m1;; translateM (k x)
    end.
\end{coq}
\vspace*{1em}

\subsection{QAC in ITrees}
\label{sec:qac-itree}
The ITree specification language is a superset of the QAC language family.  Any
server model in \autoref{sec:qac-model} can be translated into an interaction
tree that receives requests and sends responses.

The deterministic server's event type is defined as follows:
\begin{coq}
  Variant ioE (I O: Type) : Type -> Type := (* I/O event *)
    Recv: qaE I                     (* receive an input  *)
  | Send: O -> qaE unit.            (* send    an output *)

  Definition qaE := ioE Q A.
\end{coq}
Here \ilc{ioE} is a parameterized event that takes types \ilc I and \ilc O as
arguments.  Type \ilc{qaE} is an instance of I/O event that receives requests of
type \ilc Q and sends responses of type \ilc A.

Given a loop body \ilc{sstep} and an initial state \ilc s, we can define the
interaction tree of the server model as a recursive program:
\begin{coq}
  CoFixpoint detServer (sstep: Q -> sigma -> A * sigma) (s: sigma) : itree qaE void :=
    q <- trigger Recv;;
    let (a, s') := sstep q s in
    trigger (Send a);;
    detServer sstep s'.
\end{coq}
The server program iterates over state \ilc{(s: sigma)}.  In each iteration, it
receives a request \ilc{(q: Q)} and computes the response \ilc{(a: A)} using the
state monad function \ilc{sstep}.  It then sends back the response and recurses
with the post state \ilc{(s': sigma)}.

To specify a nondeterministic server with internal choices of type $C$, I
introduce another event to represent its internal choices.
\begin{coq}
  Variant choiceE: Type -> Type :=
    Choice: choiceE C.    (* make an internal choice *)

  Definition qacE: Type -> Type := qaE +' choiceE.
\end{coq}
Here \ilc{qacE} is a sum type of \ilc{qaE} and \ilc{choiceE} events, meaning
that the server's events are either sending/receiving messages or making
internal choices.

At the beginning of each iteration, the nondeterministic server makes an
internal choice \ilc{(c: C)} as the seed of its \ilc{sstep} function:
\begin{coq}
  CoFixpoint server (sstep: Q -> C -> sigma -> A * sigma) (s: sigma) : itree qacE void :=
    c <- trigger Choice;;
    q <- trigger Recv;;
    let (a, s') := sstep q c s in
    trigger (Send a);;
    server sstep s'.
\end{coq}

Such server models can be derived into tester programs by interpretation.

\subsection{Interpreting interaction trees}
\label{sec:interp}
To interpret an ITree is to specify the semantics of its events.
For example, the \ilc{tester} program in \autoref{sec:interactive-testing} can
be constructed by interpreting a deterministic server model in
\autoref{sec:qac-itree}.

\paragraph{Tester event type}
An interactive tester generates requests and sends them to the SUT, receives
responses, and validates its observation against the specification:
\begin{coq}
  Variant genE: Type -> Type :=
    Gen : genE Q.          (* generate a request *)

  Variant exceptE: Type -> Type :=
    Throw: forall X, exceptE X. (* throw an exception *)

  Definition testE := ioE A Q +' genE +' exceptE.
\end{coq}
This type definition is pronounced as: ``The tester receives responses \ilc A as
input, sends requests \ilc Q as output, generates requests to send, or throws an
exception when it detects a violation in its observations.

\paragraph{Dualizing server events}
The tester is constructed by dualizing the server specification's behavior: When
the specification receives a request, the tester generates a request and sends
it to the SUT; When the specification sends a response, the tester expects to
receive the response from the SUT.  We can write this correspondence as a
function from server events to the tester's behaviors:
\begin{coq}
  Definition dualize {R} (e: qaE R) : itree testE R :=
    match e with
    | Recv   => q <- trigger Gen;;
                trigger (Send q);;
                ret q
    | Send a => a' <- trigger Recv;;
                if a' =? a
                then ret tt
                else trigger (Throw ("Expect " ++ a ++ " but observed " ++ a'))
    end.
\end{coq}
The \ilc{dualize} function takes an event and yields an ITree program.  It maps
the server's receive event to the tester's generate and send operations, and
vice versa, maps the server's send event to the tester's receive operation.  If
the tester's received response differs from that specified by the server model,
then the tester throws an exception to indicate a violation it detected.

\paragraph{Applying event handlers}
The \ilc{dualize} function is also called a {\em handler}.  It produces an ITree
that has the same result type as its argument event.  Therefore, we can
construct the tester by substituting the server's events with the operations
defined by the handler function.  This process of substituting events with the
results of handling them is called {\em interpretation}, and it is implemented
as shown below:
\begin{coq}
  CoFixpoint interp {E F T} (handler: forall {R}, E R -> itree F R) (m: itree E T)
             : itree F T :=
    match m with
    | Pure   r   => Pure r
    | Impure e k => x <- f e;;
                    interp handler (k x)
    end.

  Definition tester (sstep: Q -> sigma -> A * sigma) (s: sigma) : itree testE void :=
    interp dualize (detServer sstep s).
\end{coq}
Applying the dualization handler to the server model gives us an ITree program
that performs tester interactions.  The tester keeps generating, sending and
receiving messages until observing an unexpected response when it throws an
exception and rejects the SUT.  It is semantically equivalent to the
\ilc{tester} in \autoref{sec:interactive-testing}.

To derive testers from nondeterministic server models, I design more complex
tester events and handlers, as discussed in the following sections.

\section{Handling Internal Nondeterminism}
\label{sec:internal-nondet}
This section applies the idea of dualization in \autoref{chap:dualize} to the
ITree context, showing how to address internal nondeterminism by symbolic
evaluation based on ITree specifications.  It covers the derivation path from
``symbolic model'' to ``nondeterministic tester'' in \autoref{fig:framework},
using \http entity tags introduced in \autoref{sec:internal-nondeterminism} as
an example.

As discussed in \autoref{sec:encode-spec}, dualization requires refining the
representation of the server's computation, \eg, encoding its branches over
symbolic conditions.  This is done by designing ITrees' event types in
\autoref{sec:symbolic-model} and specifying the server's behavior with a
symbolic model.

The server specification is derived into a tester client by {\em interpreting}
interaction trees.  To interpret is to define semantic rules that transform one
ITree program into another and corresponds to the arrows in
\autoref{fig:framework}.  \autoref{sec:interp} explains the interpretation of
ITrees.

The interpretation from symbolic model to the nondeterministic tester model is
implemented in two phases, illustrated as ``dualization'' and ``unification''
arrows in \autoref{fig:framework}: \autoref{sec:dualize-interaction} dualizes
the server's behavior into the tester client's, resulting in a ``symbolic
observer'' that encodes symbolic evaluation as primitive events.
\autoref{sec:symbolic-eval} then instantiates the primitive events into
pure computations that unify concrete observations against their symbolic
representations.

\subsection{Symbolic server model}
\label{sec:symbolic-model}
The server specification is an ITree program that exhibits all valid behavior of
the protocol.  I combine the $\Prog$ language in \autoref{sec:prog-lang} with
the simple ITree \ilc{ioE} events for sends and receives as we saw in
\autoref{sec:qac-itree}.

\paragraph{Network packet type}
Instead of receiving requests and sending responses, the server receives and
sends {\em packets} that carry routing information.  This will allow us to
specify the server's interaction against concurrent clients.  A packet consists
of headers that indicate its source and destination, and a payload of either a
request or a response:
\begin{coq}
  Notation connection := N. (* N for natural number *)

  Record packet Q A := {
    Source      : endpoint;
    Destination : endpoint;
    Payload     : Q + A
  }.
\end{coq}
This type definition says: the \ilc{packet} type is parameterized over the \ilc
Q and \ilc A types that represent its request and response.  Its \ilc{Source}
and \ilc{Destination} fields each records an \ilc{endpoint} represented as a
natural number.  Its \ilc{Payload} type is the sum of request and response.

Here's an example trace of network packets that represents a simple HTTP
exchange:
\begin{coq}
  Context get: string -> request.
  Context ok : string -> response.

  Definition server_end: endpoint := 0.

  Example trace: list (packet request response) :=
    [ { Source      := 1;
        Destination := server_end;
        Payload     := inl (get "/index.html")
      }
    ; { Source      := server_end;
        Destination := 1;
        Payload     := inr (ok "<p>Hello!</p>")
      }
    ].
\end{coq}

This trace encodes a transaction between client 1 and the server (represented as
endpoint 0).  The client sends a GET request to fetch the resource in path
\ilc{"/index.html"}, and the server responds with 200 OK and content
\ilc{"<p>Hello!</p>"}.  The \ilc{inl} and \ilc{inr} are constructors for sum
types that here are used to distinguish requests from responses:
\begin{coq}
  Context inl: forall {X Y}, X -> X + Y.
  Context inr: forall {X Y}, Y -> X + Y.
\end{coq}
\vspace*{1em}

\paragraph{Symbolic representation}
To specify systems' nondeterministic behavior, the $\Prog$ language in
\autoref{sec:prog-lang} encodes data as symbolic expressions $\Sexp$, so that
the responses and branch conditions may depend on internal choices.  I do the
same for ITree specifications, by symbolizing the choice events and branch
conditions.  Take my HTTP specification~\cite{issta21} as an example.  Its
choice event has symbolic expression as result type:

\begin{coq}
  Variant comparison := Strong | Weak.

  Variant exp: Type -> Type :=
    Const    : string -> exp string
  | Var      : var    -> exp string
  | Compare  : string -> exp string -> comparison -> exp bool.

  Variant choiceE: Type -> Type :=
    Choice: choiceE (exp string).
\end{coq}

Here I instantiate the \ilc{choiceE} in \autoref{sec:itree-lang} with symbolic
return type \ilc{(exp string)}, pronounced ``expression of type string''.  In
this example, I use strings to represent entity tags (ETags) that HTTP servers
may generate, which was discussed in \autoref{sec:internal-nondeterminism}.  The
type interface can be adjusted to other protocols under test.

Symbolic expressions may be constructed as constant values, as variables, or
with operators.  Here are some examples of expressions:
\begin{coq}
  Context x : var.

  Example expression1: exp string := Const "foo".
  Example expression2: exp string := Var x.
  Example expression3: exp bool   := Compare "bar" expression2 Weak.
\end{coq}

The \ilc{Compare} constructor takes an expression of type string and compares it
against a constant string.  \ilc{(Compare t tx cmp)} represents the ETag
comparison between \ilc{t} and \ilc{tx}, using ``strong comparison'' or ``weak
comparison'' mechanism\footnote{\http servers may choose to generate ETags as
``strong validators'' (with uniqueness guarantee) or ``weak validators'' (for
potentially better performance).  Weak validators have prefix \iletag{W/} while
strong validators do not.  When handling compare-and-swap operations such as PUT
requests conditioned over \inlinec{If-Match} in
\autoref{sec:internal-nondeterminism}, the server should evaluate its
precondition with ``strong comparison'' that doesn't allow weak validators,
\eg, \iletag{W/"foo"}) to match any ETag including itself.  For GET requests
conditioned over \inlinec{If-None-Match}, the server may evaluate with ``weak
comparison'' where a weak validator like \iletag{W/"bar"} matches itself and
also matches strong validator \iletag{"bar"}, but doesn't match \iletag{W/"foo"}
or \iletag{"foo"}.\label{foot:etag}} specified by \ilc{cmp}.  The constant ETag
is provided by the request, and the symbolic one comes from the server state.

\begin{figure}
\begin{lstlisting}[numbers=left]
Notation sigma := (path -> resource).

Context OK PreconditionFailed : symbolic_response.
Context process: request -> sigma -> itree smE (symbolic_response * sigma).

CoFixpoint server_http (state: sigma) :=
  pq <- trigger Recv;;
  let respond_with a :=
    trigger (Send { Source      := server_conn
                  ; Destination := pq.(Source)
                  ; Payload     := inr a } ) in
  let q : request    := request_of pq        in
  let v : string     := q.(Content)          in
  let k : path       := q.(TargetPath)       in
  let t : string     := if_match q           in
  let tx: exp string := (state k).(ETag)     in
  IFX (Compare t tx Strong)%\label{line:etag-ifx}%
  THEN
    if q.(Method) is Put%\label{line:etag-pure-if}%
    then
      tx' <- or (trigger Choice)%\label{line:etag-choice}%
                (Pure (Const EmptyString));;
      let state' := state [k |-> {Content := v; ETag := tx'}] in
      respond_with OK;;
      server_http state'
    else                 (* handling other kinds of requests *)
      (a, state') <- process q state;;
      respond_with a;;
      server_http state'
  ELSE
    respond_with PreconditionFailed;;
    server_http s.%\label{line:etag-end}%
\end{lstlisting}
\vspace*{1em}
\caption{Server model for HTTP conditional requests.}
\label{fig:if-match-server}
\end{figure}

\autoref{fig:if-match-server} shows an ITree model for If-Match requests 
(\autoref{sec:internal-nondeterminism}).  The server state type \ilc{sigma} maps
each path to its corresponding ``resource''---file content and metadata like
ETag.  The server first evaluates the request's \inlinec{If-Match} condition by
``strong comparison'' as required by HTTP.  If the request's ETag matches its
target's, then the server updates the target's contents with the request
payload.  The target's new ETag \ilc{tx'} is permitted to be any value, so the
model represents it as \ilc{Choice} event.

Note that the server model exhibits three kinds of branches:
\begin{enumerate}
\item The \ilc{if} branch in \autoref{line:etag-pure-if} is provided by ITree's
embedding language Coq, and takes a boolean value as condition.
\item The \ilc{IFX} branch in \autoref{line:etag-ifx} constructs an ITree that
nondeterministically branches over a condition written as a symbolic expression
of type bool:
\begin{coq}
  Variant branchE: Type -> Type := (* symbolic-conditional branch *)
    Decide: exp bool -> branchE bool.

  Notation "IFX condition THEN x ELSE y" :=
    (b <- trigger (Decide condition);;
     if b then x else y).
\end{coq}
\item The \ilc{or} operator in \autoref{line:etag-choice} takes two ITrees
\ilc x and \ilc y as possible branches and constructs a nondeterministic program
that may behave as either \ilc x or \ilc y:
\begin{coq}
  Variant nondetE: Type -> Type := (* nondeterministic branch *)
    Or: nondetE bool.

  Definition or {E R} `{nondetE -< E} (x y: itree E R) : itree E R :=
    b <- trigger Or;;
    if b then x else y.
\end{coq}
Here \ilc{nondetE} is a special case of \ilc{choiceE} with boolean space of
choices, but they'll be handled differently during test derivation.  The type
signature \ilc{\{E R\} `\{nondetE -< E\}} says the \ilc{(or)} combinator can
apply to ITrees whose event \ilc E is a super-event of \ilc{nondetE}, and with
arbitrary return type \ilc R.  For example, arguments \ilc x and \ilc y can be
of type \ilc{(itree (qaE +' nondetE +' choiceE +' branchE) void)}.
\end{enumerate}

These three kinds of branch conditions play different roles in the
specification, and will be handled differently during testing:
\begin{enumerate}
\item The ``pure'' \ilc{if} condition is used for deterministic branches like
  \ilc{(q.(Method) is Put)} in the example.  Here \ilc q is a ``concrete
  request''---a request that doesn't involve symbolic variables, as opposed to
  ``symbolic'' ones---generated by the client and sent to the server, so its
  method is known by the tester and needn't be symbolically evaluated.
\item The ``symbolic'' \ilc{IFX} condition here plays a similar role as the
  $\mathsf{if}$ branches in the $\Prog$ language: Which branch to take depends
  on the server's internal choices, so the tester needs to consider both cases.
\item The \ilc{or} branch defines multiple control flows the server may take.
  In the HTTP example, the server may generate an ETag for the resource's new
  content but is not obliged to do so.  It may choose to generate no ETags
  instead, using \ilc{(Pure (Const EmptyString))} as an alternative
  to \ilc{(trigger Choice)}.
\end{enumerate}

In addition to \ilc{IFX} branch conditions, the symbolic expressions may also
appear in the server's responses.  For example, after generating an ETag in
\autoref{line:etag-choice} of \autoref{fig:if-match-server}, the server may
receive a GET request and send the ETag to the client:
\begin{coq}
  Example ok_with_etag: symbolic_response :=
    { ResponseLine := { Version := { Major := 1; Minor := 1 }
                      ; Code    := 200
                      ; Phrase  := OK
                      }
    ; ResponseFields :=
      [ { Name := "Content-Length"; Value := Const "13" }
      ; { Name := "ETag"          ; Value := Var    x   }
      ]
    ; ResponseBody := "<p>Hello!</p>"
    }.
\end{coq}
Suppose the server generated the ETag as expression \ilc{(Var x)}, then we can
use the expression to construct the symbolic response in the specification,
rather than determining its concrete value.  The mechanism of producing
expressions for ETags is explained in \autoref{sec:dualize-interaction}.

Now we can define the specification's event type \ilc{smE}.  The symbolic server
model receives concrete requests and sends symbolic responses, so its event is
defined as:
\begin{coq}
  Definition symbolic_packet := packet request symbolic_response.%\label{def:symbolic-packet}%

  Definition qaE := ioE symbolic_packet symbolic_packet.%\label{def:symbolic-qae}%

  Notation smE := (qaE +' nondetE +' choiceE +' branchE).
\end{coq}
The ``Symbolic Model'' in \autoref{fig:framework} is an ITree constructed by
applying the \ilc{server_http} function to an initial state:
\begin{coq}
  Definition sm_http: itree smE void :=
    server_http init_state.
\end{coq}

The rest of this section will explain the interpretations from this symbolic
model.

\subsection{Dualizing symbolic model}
\label{sec:dualize-interaction}
This subsection takes the symbolic model composed in
\autoref{sec:symbolic-model} and dualizes its interactions, which corresponds
to the ``Dualization'' arrow in \autoref{fig:framework}.  It applies the idea of
derivation rules (\ref{rule:write})--(\ref{rule:return}) for $\Prog$
(\autoref{sec:dualize-prog}) to models written as ITrees.

This interpretation phase produces a symbolic observer that models the tester's
observation and validation behavior.  The observer sends a request when the
server wants to receive one and receives a response when the server wants to
send one.  It also creates constraints over the server's internal choices based
on its observations.

\autoref{fig:symbolic-observer} shows the dualization algorithm.  It interprets
the symbolic model's events with the \ilc{observe} handler, whose types are
explained as follows:

\begin{figure}
\begin{lstlisting}[numbers=left]
Notation oE := (observeE +' nondetE +' choiceE +' constraintE).

Definition observe {R} (e: smE R) : itree oE R :=
  match e with
  | Recv      => trigger FromObserver%\label{line:observe-absorb}%
  | Send px   => p <- trigger ToObserver;;%\label{line:observe-emit}%
                 trigger (Guard px p)
  | Decide bx => or (trigger (Unify bx true);;  ret true)%\label{line:observe-branch}%
                    (trigger (Unify bx false);; ret false)
  | Or        => trigger Or%\label{line:observe-or}%
  | Choice    => trigger Choice%\label{line:observe-choice}%
  end.

Definition observer_http: itree oE void :=
  interp observe sm_http.
\end{lstlisting}
\vspace*{1em}
\caption{Dualizing symbolic model into symbolic observer.}
\label{fig:symbolic-observer}
\end{figure}

The tester observes a trace of concrete packets, so observer's interactions
return concrete requests and responses, as opposed to the symbolic model whose
responses are symbolic.
\begin{coq}
  Definition concrete_packet := packet request concrete_response.

  Variant observeE : Type -> Type :=
    FromObserver   : observeE concrete_packet
  | ToObserver     : observeE concrete_packet.
\end{coq}

Note that the observer's send and receive events both return the packet sent or
received, unlike the server model whose \ilc{Send} event takes the sent packet
as argument.  This is because the tester needs to generate the request packet to
send, and the event's result value represents that generated and sent packet.

As discussed in \autoref{sec:dualize-prog}, when the server sends a symbolic
response or branches over a symbolic condition, the tester needs to create
symbolic constraints accordingly.  The observer introduces ``constraint events''
to represent the creation of constraints primitively.
\begin{coq}
  Variant constraintE : Type -> Type :=
    Guard : packet -> concrete_packet -> constraintE unit
  | Unify : exp bool -> bool -> constraintE unit.
\end{coq}

Here \ilc{(Guard px p)} creates a constraint that the symbolic packet \ilc{px}
emitted by the specification matches the concrete packet \ilc p observed during
runtime.  \ilc{(Unify bx b)} creates a constraint that unifies the symbolic
branch condition \ilc{bx} with boolean value \ilc b.  These constraints will be
solved in \autoref{sec:symbolic-eval}.

The dualization algorithm in \autoref{fig:symbolic-observer} does the follows:
\begin{enumerate}
\item When the symbolic model receives a packet, in
\autoref{line:observe-absorb}, the observer generates a request packet.

\item When the symbolic model sends a symbolic packet \ilc{px}, in
\autoref{line:observe-emit}, the observer receives a concrete packet \ilc p, and
adds a \ilc{Guard} constraint that restricts the symbolic and concrete packets
to match each other.

\item When the symbolic model branches on a symbolic condition \ilc{bx},
in \autoref{line:observe-branch}, the tester accepts the observation if it can
be explained by any branch.  This is done by constructing the observer as a
nondeterministic program that has both branches, using the \ilc{or} combinator.
For each branch, the observer adds a \ilc{Unify} constraint that the symbolic
condition matches the chosen branch.

\item Nondeterministic branches in \autoref{line:observe-or} are preserved in
this interpretation phase and will be resolved in \autoref{sec:backtrack}.

\item Internal choices in \autoref{line:observe-choice} are addressed by the
next phase in \autoref{sec:symbolic-eval}, along with the constraints created in
this phase.
\end{enumerate}

The result of dualization is a symbolic observer that models tester behaviors
like sending requests and receiving responses.  The symbolic observer is a
nondeterministic program with primitives events like making choices and adding
constraints over the choices.

For example, dualizing \autoref{line:etag-ifx}--\ref{line:etag-end} in
\autoref{fig:if-match-server} results in an observer program that is equivalent
to \autoref{fig:observer-example}.  Dualization transforms the \ilc{sm_http}
specification into the \ilc{observer_http} program in
\autoref{fig:symbolic-observer} automatically, and includes more details than
the hand-written \ilc{observer_body} example.  The next subsection interprets
the observer's primitive \ilc{Guard} and \ilc{Unify} events into pure
computations of symbolic evaluation.

\begin{figure}
\begin{coq}
  Example observer_body: itree oE void :=
    let guard_response a :=
      p <- trigger ToObserver;;
      trigger (Guard { Source      := server_conn
                     ; Destination := pq.(Source)
                     ; Payload     := inr a } ) in
    let bx: exp bool := Compare t tx Strong in
    or (
        trigger (Unify (Compare bx true));;
        if q.(Method) is Put
        then
          tx' <- or (trigger Choice)
                    (Pure (Const EmptyString));;
          let state' := state [k |-> {Content := v; ETag := tx'}] in
          guard_response OK;;
          interp observe (server_http state')
        else
          (a, state') <- interp observe (process q state);;
          guard_response a;;
          interp observe (server_http state')
       )
       (
        trigger (Unify (Compare bx false));;
        guard_response PreconditionFailed;;
        interp observe (server_http s)
       ).
\end{coq}
\vspace*{1em}
\caption{Symbolic observer example.}
\label{fig:observer-example}
\end{figure}

\subsection{Symbolic evaluation}
\label{sec:symbolic-eval}
This subsection takes the symbolic observer produced in
\autoref{sec:dualize-interaction} and solves the constraints it has created.
The constraints unify symbolic packets and branch conditions against the
concrete observations.  The tester should accept the SUT if the constraints are
satisfiable.

\begin{figure}
\begin{lstlisting}[numbers=left]
Notation ntE := (observeE +' nondetE +' exceptE).

Definition V: Type := list var * list (constraintE unit).
  
Definition unify {R} (e: oE R) (v: V) : itree ntE (V * R) :=
  let (xs, cs) := v in
  match e with
  | Choice => let x: var := fresh v in%\label{line:unify-choice}%
              ret (x::xs, cs, Var x)
  | (constraint: unifyE) => let cs' := constraint::cs in%\label{line:unify-constraint}%
                            if solvable cs'
                            then ret (xs, cs', tt)
                            else Trigger (Throw ("Conflict: " ++ print cs'))
  | Or             => b <- trigger Or;; ret (v, b)%\label{line:unify-or}%
  | (oe: observeE) => r <- trigger oe;; ret (v, r)%\label{line:unify-observe}%
  end.

Definition nondet_tester_http: itree ntE void :=
  (_, vd) <- interp_state unify observer_http initV;;
  match vd in void with end.
\end{lstlisting}
\vspace*{1em}
\caption{Resolving symbolic constraints.}
\label{fig:nondet-tester}
\end{figure}

As shown in \autoref{fig:nondet-tester}, the unification algorithm evaluates the
primitive symbolic events into a stateful checker program, which reflects the
$\Prog$-based validator in \autoref{sec:dualize-prog}.  The interpreter
maintains a validation state \ilc V which stores the symbolic variables and the
constraints over them.  The derivation rules are as follows:
\begin{enumerate}
  \item When the server makes an internal choice in \autoref{line:unify-choice},
    the tester creates a fresh variable and adds it to the validation state.
  \item When the observer creates a constraint in
    \autoref{line:unify-constraint}, the tester adds the constraint to the
    validation state and solves the new set of constraints.  If the constraints
    become unsatisfiable, then the tester \ilc{Throw}s an exception that
    indicates the current execution branch cannot accept the observations:
\begin{coq}
  Variable exceptE: Type -> Type :=
    Throw: forall {X}, string -> exceptE X.
\end{coq}
  \item The observer is a nondeterministic program with multiple execution
    paths, constructed by \ilc{Or} events in \autoref{line:unify-or}.  The
    tester accepts the observation if any of the branches does not throw an
    exception.  These branches will be handled in the next section, along with
    the observer's send/receive interactions in \autoref{line:unify-observe}.
\end{enumerate}

Note that the \ilc{unify} function interprets a symbolic observer's event
\ilc{(oE R)} into a function of type \ilc{(V -> itree tE (V * R))}, a so-called
``state monad transformer''.  It takes a pre-validation state \ilc{(v: V)} and
computes an ITree that yields the observer event's corresponding result
\ilc{(r: R)} along with a post-validation state \ilc{(v': V)}.  This stateful
interpretation process is implemented by a variant of \ilc{interp} called
\ilc{interp_state}:
\begin{coq}
  CoFixpoint interp_state {E F V R}
                          (handler: forall {X}, E X -> V -> itree F (V * X))
                          (m: itree E R) (v: V)
             : itree F (V * R) :=
    match m with
    | Pure   r   => ret (v, r)
    | Impure e k => '(v', r) <- handler e v;;
                    interp_state handler (k r) v'
    end.
\end{coq}

So far I have interpreted the symbolic model into a tester model that observes
incoming and outgoing packets, nondeterministically branches, and in some cases
throws exceptions.  In this section, I used the server model verbatim as the
symbolic model, and the derived tester can handle internal nondeterminism by
symbolic evaluation.

The server model's receive and send events are linear.  It doesn't receive new
requests before responding to its previous request.  As a result, the tester
dualized from the server model doesn't send new requests before it receives the
response to the previous request, so it doesn't reveal the server's interaction
with multiple clients simultaneously.  In the next section, I'll explain how to
create a multi-client tester by extending the symbolic model, addressing
external nondeterminism.

\section{Handling External Nondeterminism}
\label{sec:external-nondet}
As introduced in \autoref{sec:intro-external-nondet}, the environment might
affect the transmission of messages, so called external nondeterminism.  The
tester should take the environment into account when validating its
observations.

This section explains how to address external nondeterminism by specifying the
environment, using the networked server example.  It corresponds to the
``Composition'' arrow in \autoref{fig:framework}.  \autoref{sec:net-tcp} defines
a model for concurrent TCP connections.  \autoref{sec:net-compose} then composes
the network model with the server specification, yielding a tester-side
specification that defines the space of valid observations.

\subsection{Modelling the network}
\label{sec:net-tcp}
When testing servers over the network, request and response packets may be
delayed.  As a result, messages from one end might arrive at the
other end in a different order from that they were sent.

The space of network reorderings can be modelled by a {\em network model}, a
conceptual program for the ``network wire''.  The wire can be viewed as a
buffer, which absorbs packets\footnote{In this section, ``packet'' is a
shorthand for the ``symbolic packet'' defined on
Page~\pageref{def:symbolic-packet}.} and later emits them:
\begin{coq}
  Notation packet := symbolic_packet.

  Definition netE: Type -> Type :=
    ioE packet packet.

  Notation Absorb := Recv.
  Notation Emit   := Send.
\end{coq}

For example, the network model for concurrent TCP connections is defined in
\autoref{fig:tcp-model}.  The model captures TCP's feature of maintaining the
order within each connection, but packets in different connections might be
reordered arbitrarily.  When the wire chooses a packet to send, the candidate
must be the oldest in its connection.

\begin{figure}
\begin{coq}
(* filter the oldest packet in each connection *)
Context oldest_in_each_conn : list packet -> list packet.

Fixpoint pick_one (l: list packet) : itree nondetE (option packet) :=%\label{def:pick-one}%
  if l is p::l'
  then or (Ret (Some p)) (pick_one l')
  else ret None.

CoFixpoint tcp (buffer: list packet) : itree (netE +' nondetE) void :=
  let absorb := pkt <- trigger Absorb;;
                tcp (buffer ++ [pkt])      in
  let emit p := trigger (Emit p);;
                tcp (remove pkt buffer)    in
  let pkts   := oldest_in_each_conn buffer in
  opkt <- pick_one pkts;;
  if opkt is Some pkt
  then emit pkt
  else absorb.
\end{coq}
\vspace*{1em}
\caption[Network model for concurrent TCP connections.]{Network model for
  concurrent TCP connections.  The model is an infinite program iterating over a
  \ilc{buffer} of all packets en route.  In each iteration, the model either
  \ilc{absorb}s or \ilc{emit}s some packet, depending on the current
  \ilc{buffer} state and the choice made in \ilc{pick_one}.  Any absorbed packet
  is appended to the end of buffer.  When emitting a packet, the model may
  choose a connection and send the oldest packet in it.}
\label{fig:tcp-model}
\end{figure}

Note the \ilc{pick_one} function, which might return (i) \ilc{Some p} or (ii)
\ilc{None}.  The network model then (i) emits packet \ilc p or (ii) absorbs a
packet into \ilc{buffer}.

\begin{itemize}
\item When the given list \ilc{pkts} is empty, \ilc{pick_one} always returns
  \ilc{None}, because the wire has no packet in the \ilc{buffer}, and must
  absorb some packet before emitting anything.
\item Given a non-empty linked list \ilc{(p::l')}, with \ilc p as head and
  \ilc{l'} as tail, \ilc{pick_one} might return \ilc{(Some p)} by choosing the
  left branch.  In this case, the wire can emit packet \ilc p.  Or the function
  might choose the right branch and recursing on \ilc{l'}, meaning that the wire
  can emit some packet in \ilc{l'} or absorb some packet into the buffer.
\end{itemize}

This network model reflects the TCP environment, where messages are never lost
but might be indefinitely delayed.  In the next subsection, I'll demonstrate how
to compose the server and network models into a client-side observation model.

\subsection{Network composition}
\label{sec:net-compose}

The network connects the server on one end to the clients on other ends.  When
one end sends some message, the network model absorbs it and later emits it to
the destination.

To {\em compose} a server model with a network model is to pair the server's
\ilc{Send} and \ilc{Recv} events with the network's \ilc{Absorb} and \ilc{Emit}
events.  Since the network model is nondeterministic, it might not be ready at
some given moment to absorb packets sent by the server.  The network might also
emit a packet before the server is ready to receive it.

\begin{figure}
  \includegraphics[width=\textwidth]{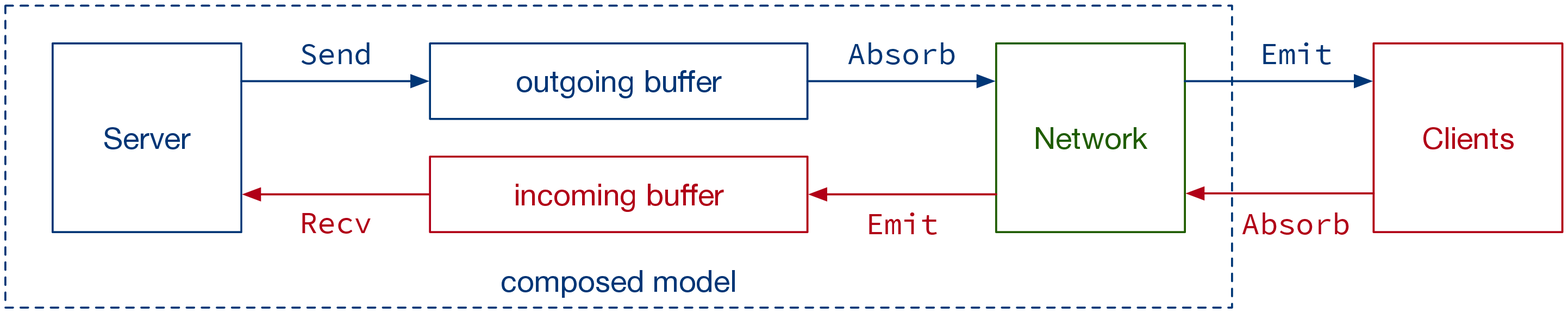}
  \caption{Network composition architecture.}
  \label{fig:net-compose}
\end{figure}

To handle the asynchronicity among the server and network events, I insert
message buffers between them.  As shown in \autoref{fig:net-compose}, the {\em
  incoming buffer} stores the packets emitted by the network but not yet
consumed by the server's \ilc{Recv} events, and the {\em outgoing buffer} stores
the packets sent by the server but not yet absorbed by the network.

\begin{figure}
\begin{lstlisting}[numbers=left]
CoFixpoint compose {E} (srv: itree smE void)          (* server  model *)
           (net  : itree (netE +' nondetE) void)      (* network model *)
           (bi bo: list packet)       (* incoming and outgoing buffers *)
           : itree (netE +' nondetE +' E) void :=
  let step_net :=%\label{line:step-net-def}%
    match net with
    | Impure Absorb knet =>
      match bo with
      | pkt::bo' => compose srv (knet pkt) bi bo'%\label{line:net-absorb}%
      | []       => pkt <- trigger Absorb;;%\label{line:client-send}%
                    compose srv (knet pkt) bi bo
      end
    | Impure (Emit pkt) knet =>%\label{line:net-emit}%
      if toServer pkt
      then compose srv (knet tt) (bi++[pkt]) bo%\label{line:srv-incoming}%
      else trigger (Emit pkt);;%\label{line:net-send}%
           compose srv (knet tt) bi bo
    | Impure Or knet => b <- trigger Or;;
                      compose srv (knet b) bi bo
    | Pure vd => match vd in void with end
    end
  in
  match srv with
  | Impure Recv ksrv =>%\label{line:srv-recv}%
    match bi with
    | pkt::bi' => compose (ksrv pkt) net bi' bo%\label{line:srv-consume}%
    | [] => step_net%\label{line:step-net}%
    end
  | Impure (Send pkt) ksrv =>%\label{line:srv-send}%
    compose (ksrv tt) net bi (bo++[pkt])
  | Impure e ksrv =>        (* other events performed by the server *)
    r <- trigger e;; compose (ksrv r) net bi bo
  | Pure vd => match vd in void with end
  end.
\end{lstlisting}
\vspace*{1em}
\caption[Network composition algorithm.]{Network composition algorithm.  When the
  server wants to send a packet in \autoref{line:srv-send}, the packet is
  appended to the outgoing buffer until absorbed by the network
  in \autoref{line:net-absorb}, and eventually emitted to the client
  in \autoref{line:net-send}.  Conversely, a packet sent by the client is
  absorbed by the network in \autoref{line:client-send}, emitted to the server's
  incoming buffer in \autoref{line:srv-incoming}, until the server consumes it
  in \autoref{line:srv-consume}.}
\label{fig:net-compose-code}
\end{figure}

The server and the clients are the opposite ends of the network.  Each packet
has routing fields that indicate its source and destination.  When the network
emits a packet, we need to determine whether the packet is emitted to the
server's incoming buffer or to the clients, by inspecting its destination:
\begin{coq}
  Definition toServer (p: packet) : bool :=
    if p.(Destination) is server_conn then true else false.
\end{coq}

Now we can define the composition algorithm formally, as shown in
\autoref{fig:net-compose-code}.  The function takes the symbolic server model
in \autoref{sec:symbolic-model} and the network model in \autoref{sec:net-tcp},
and yields a symbolic model of the server's behavior observable from across the
network.

The composition function analyzes the server and the network's behavior:
\begin{enumerate}
\item When the server wants to send a packet in \autoref{line:srv-send}, the
packet is appended to the outgoing buffer.
\item When the network wants to absorb a packet in \autoref{line:net-absorb},
it first checks whether the server has sent some packet to its outgoing buffer.
If yes, then the network absorbs the oldest packet in the buffer.  Otherwise, it
absorbs from the clients.
\item After absorbing some packets, when the network wants to emit a packet in
\autoref{line:net-emit}, the packet is either emitted to the client or appended
to the server's incoming buffer, based on its destination.
\item When the server wants to receive a packet in \autoref{line:srv-recv},
it first checks whether the network has emitted some packet to the incoming
buffer.  If yes, then the server takes the oldest packet in the buffer.
Otherwise, it waits for the network model to emit one.
\end{enumerate}

Note that this algorithm schedules the server at a higher priority than the
network model.  The composed model only steps into the network model when the
server is starved in \autoref{line:step-net}, by calling the \ilc{step_net}
process defined in \autoref{line:step-net-def}.  This design is to avoid
divergence of the derived tester program, which I'll further explain in
\autoref{sec:backtrack}.

By composing the server and network models, we get a symbolic model that may
send and receive packets asynchronously, as opposed to the server model that
processes one request at a time.  Dualizing this asynchronous symbolic model
results in an asynchronous tester model that may send multiple requests
simultaneously rather than waiting for previous responses.  Next I'll show how
to execute the tester's ITree model against the SUT.

\section{Executing the Tester Model}
\label{sec:backtrack}
This section takes the nondeterministic tester model derived in
\autoref{sec:symbolic-eval} and transforms it into an interactive program.
\autoref{sec:backtracking} handles the nondeterministic branches via
backtracking and produces a deterministic tester model.  \autoref{sec:itree-io}
then interprets the deterministic tester into an IO program that interacts with
the SUT.

\subsection{Backtrack execution}
\label{sec:backtracking}
This subsection explains how to run the nondeterministic tester on a
deterministic machine.  It reflects the derivation rules (\ref{rule:unsat}) and
(\ref{rule:reject}) for $\Prog$ in \autoref{sec:dualize-prog}, and constructs
the ``Backtracking'' arrow in \autoref{fig:framework}.

The deterministic tester implements a client that sends and receives concrete
packets:
\begin{coq}
  Variant clientE: Type -> Type :=
    ClientSend: concrete_packet -> clientE unit
  | ClientRecv: clientE (option concrete_packet).
\end{coq}

Note that the \ilc{ClientRecv} event might return \ilc{(Some pkt)}, indicating
that the SUT has sent a packet \ilc{pkt} to the tester; or it might
return \ilc{None}, when the SUT is silent, or its sent packet hasn't arrived at
the tester side.  This allows the tester to perform non-blocking interactions,
instead of waiting for the SUT, which might cause starvation.

\begin{figure}
\begin{lstlisting}[numbers=left]
Notation tE := (clientE +' genE +' exceptE).

CoFixpoint backtrack (current:      itree ntE void)
                     (others: list (itree ntE void))
           : itree tE void :=
  match current with
  | Impure e k =>
    match e with
    | Or           => b <- trigger GenBool;;%\label{line:backtrack-or}%
                      backtrack (k b) (k (negb b)::others)
    | Throw msg    => match others with%\label{line:backtrack-throw}%
                      | other::ot' => backtrack other ot'
                      | []         => trigger (Throw msg)
                      end
    | FromObserver => q <- trigger GenPacket;;%\label{line:backtrack-send}%
                      trigger (ClientSend q);;
                      let others' := expect FromObserver q others in
                      backtrack (k q) others'
    | ToObserver   =>
      oa <- trigger ClientRecv;;%\label{line:backtrack-recv}%
      match oa with
      | Some oa => let others' := expect ToObserver a others in
                   backtrack (k a) others'
      | None    =>%\label{line:backtrack-silent}%
        match others with
        | other::ot' => backtrack other (ot'++[current]) (* postpone *)%\label{line:backtrack-postpone}%
        | []         => backtrack m     []               (* retry    *)
        end
      end
    end
  | Pure vd => match vd in void with end
  end.

Definition tester_http: itree tE void :=
  backtrack nondet_tester_http [].
\end{lstlisting}
\vspace*{1em}
\caption{Backtrack execution of nondeterministic tester.}
\label{fig:backtrack}
\end{figure}

\autoref{fig:backtrack} shows the backtracking algorithm.  It interacts with the
SUT and checks whether the observations can be explained by the nondeterministic
tester model.  That is, checking whether the tester has an execution path that
matches its interactions.  This is done by maintaining a list of all possible
branches in the tester, and checking if any of them accepts the observation.

The tester exhibits two kinds of randomness: (1) When sending a request packet
to the SUT, it generates the packet randomly with \ilc{GenPacket}; (2) When the
nondeterministic tester model branches, the deterministic tester randomly picks
one branch to evaluate, using \ilc{GenBool}:
\begin{coq}
  Variant genE: Type -> Type :=
    GenPacket : genE concrete_packet
  | GenBool   : genE bool.
\end{coq}

\begin{figure}
\begin{lstlisting}[numbers=left]
CoFixpoint match_observe {R} (e: observeE R) (r: R)
                             (m: itree ntE (V * void))
           : itree ntE (V * void) :=
  match m with
  | Impure (oe: observeE concrete_packet) k =>
    match oe, e with
    | FromObserver, FromObserver%\label{line:match-from}%
    | ToObserver  , ToObserver => k r%\label{line:match-to}%
    | _, _ => trigger (Throw ("Expect " ++ print oe%\label{line:match-throw}%
                           ++ " but observed " ++ print e))
    end
  | Impure e0 k =>
    r0 <- trigger e0;;
    match_observe e r (k r0)
  | Pure (_, vd) => match vd in void with end
  end.

Definition expect {R} (e: observeE R) (r: R)
  : list (itree ntE (V * void)) -> list (itree ntE (V * void))
  := map (match_observe e r).
\end{lstlisting}
\vspace*{1em}
\caption{Matching tester model against existing observation.}
\label{fig:match-observe}
\end{figure}

The execution rule is defined as follows:
\begin{enumerate}
\item When the tester nondeterministically branches, in
\autoref{line:backtrack-or}, randomly pick a branch \ilc{(k b)} to evaluate and
push the other branch \ilc{(k (negb b))} to the list of other possible cases.

\item When the \ilc{current} tester throws an exception, in
\autoref{line:backtrack-throw}, it indicates that the current execution path
rejects the observations.  The tester should try to explain its observations
with other branches of the tester model.  If the \ilc{others} list is empty, it
indicates that the observation is beyond the specification's producible
behavior, so the tester should reject the SUT.

\item When the tester wants to observe a packet {\em from} itself, it generates
a packet and sends it to the SUT in \autoref{line:backtrack-send}.

Note that if the current branch is rejected and the tester backtracks to other
branches, the sent packet cannot be recalled from the environment.  Therefore,
all other branches should recognize this sent packet and check whether future
interactions are valid follow-ups of it.  This is done by matching the branches
against the send event, using the \ilc{expect} function.

  As shown in \autoref{fig:match-observe}, \ilc{(expect e r l)} matches every
  tester in list \ilc l against the observation \ilc e that has return
  value \ilc r.  For each element $\texttt m\in\texttt l$, if \ilc m's first
  observer event \ilc{oe} matches the observation \ilc e
  (\autoref{line:match-from} and \autoref{line:match-to}),
  then \ilc{match_observe} instantiates the tester's continuation function \ilc
  k with the observed result \ilc r.  Otherwise, the tester throws an exception
  in \autoref{line:match-throw}, indicating that this branch cannot explain the
  observation because they performed different events.
  \label{rule:backtrack-send}

\item When the current tester wants to observe a packet {\em to} itself, it
  triggers the \ilc{ClientRecv} event in \autoref{line:backtrack-recv}.  If a
  packet has indeed arrived, then it instantiates the current branch as well as
  other possible branches, as discussed in Rule~(\ref{rule:backtrack-send}).

  If the tester hasn't received a packet from the SUT
  (\autoref{line:backtrack-silent}), it doesn't reject the SUT, because the
  expected packet might be delayed in the environment.  If there are \ilc{other}
  branches to evaluate (\autoref{line:backtrack-postpone}), then the tester
  postpones the \ilc{current} branch by appending it to the back of the queue.
  Otherwise, if the current branch is the only one that hasn't rejected, then
  the tester retries the receive interaction.

  Note that if the SUT keeps silent, then the tester will starve but won't
  reject, because (i) such silence is indistinguishable from the SUT sending a
  packet that is delayed by the environment, and (ii) the SUT hasn't {\em
  exhibited} any violations against the specification.  The starvation issue is
  addressed in \autoref{sec:itree-io}.
\end{enumerate}

The backtracking algorithm also explains the network composition design
in \autoref{fig:net-compose-code}, where the server model is scheduled at a
higher priority than the network model.  Suppose the SUT has produced some
invalid output.  Then the tester should reject its observation by throwing an
exception in every branch.  However, the network model is always ready to absorb
a packet, because the \ilc{pick_one} function on Page~\pageref{def:pick-one}
always includes a branch that returns \ilc{None}.  If the composition algorithm
prioritizes the network model over the server model, then when one branch
rejects, the derived tester backtracks to other branches, which includes
generating and sending another packet to the SUT (dualized from the network
model's absorption event).  Evaluating the network model lazily prevents the
composed symbolic model from having infinitely many absorbing branches.  This
allows the derived tester to converge to rejection upon violation, instead of
continuously sending request packets.

Now we have derived the specification into a deterministic tester model in
ITree.  The tester's events reflect actual computations of a client program.  In
the next subsection, I'll translate the ITree model into a binary executable
that runs on silicon and metal.

\subsection{From ITree model to IO program}
\label{sec:itree-io}
The deterministic tester model derived in \autoref{fig:backtrack} is an ITree
program that never returns (its result type \ilc{void} has no elements).  It
represents a client program that keeps interacting with the SUT until it reveals
a violation and throws an exception.

In practice, if the tester hasn't found any violation after performing a certain
number of interactions, then it accepts the SUT.  This is done by executing the
ITree until reaching a certain depth.

\begin{figure}
\begin{lstlisting}[numbers=left]
Fixpoint execute (fuel: nat) (m: itree tE void) : IO bool :=
  match fuel with
  | O       => ret true            (* accept if out of fuel *)%\label{line:execute-accept}%
  | S fuel' =>
    match m with
    | Impure e k =>
      match e with
      | Throw _      => ret false  (* reject upon exception *)%\label{line:execute-reject}%
      | ClientSend q => client_send q;;
                        execute fuel' (k tt)
      | ClientRecv   => oa <- client_recv;;
                        execute fuel' (k oa)
      | GenPacket    => pkt <- gen_packet;;%\label{line:execute-gen}%
                        execute fuel' (k pkt)
      | GenBool      => b <- ORandom.bool;;
                        execute fuel' (k b)
      end
    | Pure vd => match vd in void with end
    end
  end.

Definition test_http: IO bool :=
  execute bigNumber tester_http.
\end{lstlisting}
\vspace*{1em}
\caption{Interpreting ITree tester to IO monad.}
\label{fig:execute}
\end{figure}

As shown in \autoref{fig:execute}, the \ilc{execute} function takes an
argument \ilc{fuel} that indicates the remaining depth to explore in the ITree.
If the execution ran out of fuel (\autoref{line:execute-accept}), then the test
accepts; If the tester model throws an exception
(\autoref{line:execute-reject}), then the test rejects.  Otherwise, it
translates the ITree's primitive events into IO computations in
Coq~\cite{SimpleIO}, which are eventually extracted into OCaml programs that can
be compiled into executables that can communicate with the SUT over the
operating system's network stack.

This concludes my validation methodology.  In this chapter, I have shown how to
test real-world systems that exhibit internal and external nondeterminism.  I
applied the dualization theory in \autoref{chap:theory} to address internal
nondeterminism and handled external nondeterminism by specifying the
environment's space of uncertainty.  The specification is derived into an
executable tester program by multiple phases of interpretations.  The
derivation framework is built on the ITree specification language, but the
method is applicable to other languages that allow destructing and analyzing the
model programs.

So far, I have answered ``how to tell compliant implementations from violating
ones''.  The next chapter will answer ``how to generate and shrink test inputs
that reveal violations effectively'' and unveil the techniques behind
\ilc{gen_packet} in \autoref{line:execute-gen} of \autoref{fig:execute}.

\chpt{Test Harness Design}
\label{chap:harness}
A tester consists of a validator and a test harness.  Chapters~\ref{chap:theory}
and \ref{chap:practices} have explained the validator's theory and practices.
In this chapter, I present a language-based design for the test harnesses,
showing how to generate and shrink test inputs effectively and address
inter-execution nondeterminism.

\autoref{sec:harness-overview} provides a brief overview of how test harnesses
work.  \autoref{sec:heuristics} explains how to write heuristics to generate
interesting test inputs.  \autoref{sec:shrinking} then shows a shrinking
mechanism that keeps the test inputs interesting among different executions.

\section{Overview}
\label{sec:harness-overview}
\begin{figure}[t]
  \centering
  \includegraphics[width=.9\textwidth]{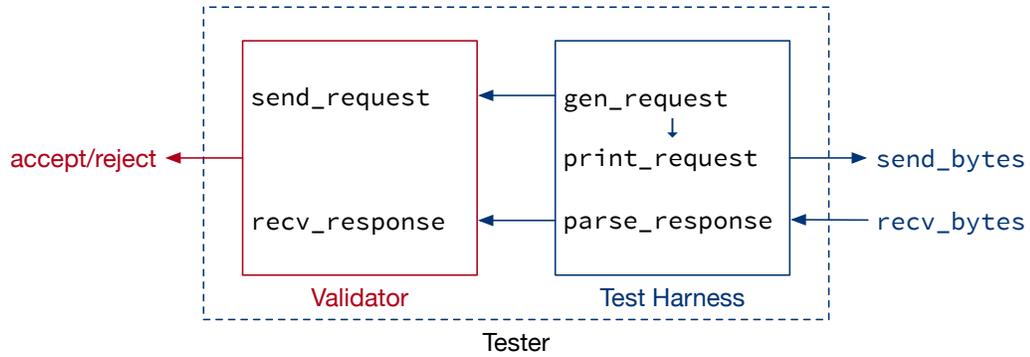}
  \caption{Tester architecture outline.}
  \label{fig:overview}
\end{figure}

This section introduces the abstract architecture of an interactive tester,
using the networked server as an example.  I'll present a na\"ive implementation
of the test harness, which will be improved in the following sections.

The test harness interacts with the environment and provides the observations
for the validator.  The validator may represent requests and responses as
abstract datatypes for the convenience of specification.  The test harness
translates these abstract representations into bytes transmitted on the
underlying channel.

As shown in \autoref{fig:overview}, when the validator wants to observe a sent
request, the harness generates the request and encodes it into bytes to send.
Conversely, when the validator wants to observe a received response, the harness
receives bytes from the environment and decodes them into abstract messages.

\begin{figure}
\begin{lstlisting}[numbers=left]
Definition gen_packet: IO concrete_packet :=
  src          <- random_conn;;
  method       <- oneof [Get; Put];;
  target       <- random_path;;%\label{line:random-path}%
  precondition <- oneof [IfMatch, IfNoneMatch];;
  etag         <- random_etag;;
  payload      <- random_string;;
  ret { Source      := src;
        Destination := server_conn;
        Data        := inr { Method     := method;
                             TargetPath := target;
                             Headers    := [(precondition, etag)];
                             Payload    := payload
                           }
      }.
\end{lstlisting}
\vspace*{1em}
\caption{Na\"ive generator for HTTP requests.}
\label{fig:naive-generator}
\end{figure}

A generator is a randomized program that produces test inputs.  One example is
the \ilc{gen_packet} function in \autoref{fig:execute}.  The HTTP packet
generator can be na\"ively implementation as shown in
\autoref{fig:naive-generator}.  This version fills in the request's fields with
arbitrary values, and has limited coverage of the SUT's behavior.  This is
because the request target and ETags are both generated randomly, but a request
is interesting only if its ETag matches its target's corresponding resource
stored on the server.  A randomly generated request would result in 404 Not
Found and 412 Precondition Failed in almost all cases.

To reveal more interesting behavior from the SUT, we should tune the generator's
distribution to emphasize certain patterns of the test input.  For example, if
the tester knows the set of paths where the server has stored resources, then it
can generate more paths within the set to hit the existing resources; if the
tester has observed some ETags generated by the server, then it can include
these ETags in future requests.  In the next section, I'll explain how to
implement such heuristics in ITree-based testers.

\section{Heuristics for Test Generation}
\label{sec:heuristics}
This section implements heuristics for generating test inputs.  I'll use the
HTTP tester as an example to show how to make requests more interesting by
parameterizing them over the model state (\autoref{sec:heuristic-state}) and the
trace (\autoref{sec:heuristic-trace}).

\subsection{State-based heuristics}
\label{sec:heuristic-state}
\paragraph{Motivation}
The model state may instruct the test generator to produce more interesting test
inputs.  For example, consider the \ilc{random_path} generator in
\autoref{line:random-path} of \autoref{fig:naive-generator}.  One way to improve
it is to generate more paths that have corresponding resources on the server:
\begin{coq}
  Definition gen_path (state: list (path * resource)) : IO path :=
    let paths: list path := map fst state in
    freq [(90, oneof paths);
          (10, random_path)].
\end{coq}

Here I modify the server model's state type \ilc{sigma} from \ilc{(path ->
  resource)} in \autoref{fig:if-match-server} into \ilc{(list (path *
  resource))}, which allows the generator to access the list of all \ilc{paths}
  in the server state.  The generator chooses from these existing paths in 90\%
  of the cases, as assigned by the \ilc{freq} combinator.  The remaining 10\%
  are still generated randomly, to discover how the SUT handles nonexistent
  paths.

For the \ilc{gen_packet} generator in \autoref{fig:naive-generator}, replacing
its \ilc{random_path} with the improved \ilc{gen_path} would generate more
interesting request targets.  This requires the \ilc{gen_packet} function to
carry the server state to instantiate \ilc{gen_path}.

As shown in \autoref{fig:backtrack}, the \ilc{GenPacket} generator is triggered
when the tester wants to observe a packet from itself to the SUT.
\autoref{fig:symbolic-observer} then shows that such \ilc{FromObserver} expectation
happens when the symbolic model \ilc{Send}s a packet.  Such a \ilc{Send} event
only happens when the server wants to receive a packet in
\autoref{fig:net-compose}.  The \ilc{Recv} events are triggered by the server
model in \autoref{fig:if-match-server}, which iterates over the server state
\ilc{sigma}.

\paragraph{Implementation}
To expose the server state to the request generator, I extend the symbolic
server model's \ilc{Recv} event type on Page~\pageref{def:symbolic-qae} to
include the server state:

\begin{minipage}{\linewidth}
\begin{coq}
  Variant qaE: Type -> Type :=
    Recv : sigma      -> qaE packet
  | Send : packet -> qaE unit.
\end{coq}
\end{minipage}

\vspace*{1em}
Now when the server wants to receive a request, it triggers \ilc{(Recv state)},
where \ilc{(state: sigma)} contains the server's paths and resources at that
point.  The \ilc{state} argument is then carried to the generator, by adding
parameters to the event types along the interpretation:
\begin{coq}
  Variant netE: Type -> Type :=
    Emit  : packet -> netE unit
  | Absorb: sigma      -> netE packet.

  Variant observeE : Type -> Type :=
    FromObserver   : sigma -> observeE concrete_packet
  | ToObserver     : observeE concrete_packet.

  Variant genE: Type -> Type :=
    GenPacket : sigma -> genE concrete_packet
  | GenBool   : genE bool.

  Definition gen_packet: sigma -> IO concrete_packet.
\end{coq}

\begin{figure}
\includegraphics[width=.6\linewidth]{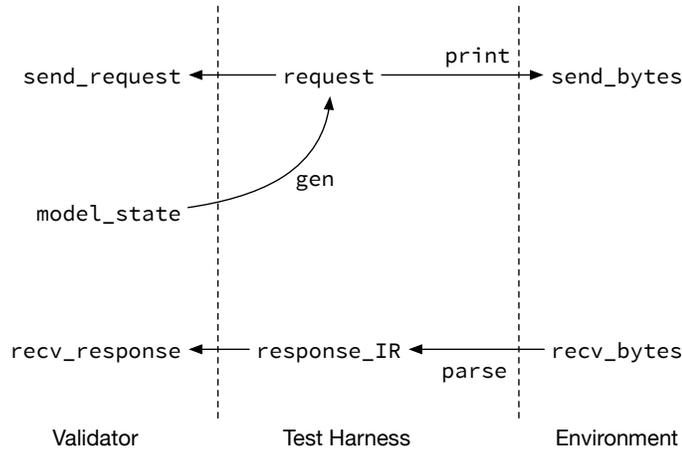}
\caption{State-based heuristics.}
\label{fig:stategen}
\end{figure}

As a result, when instantiating the \ilc{(GenPacket state)} event in
\autoref{fig:execute}, we can feed the \ilc{gen_packet} function with argument
\ilc{state}, so that \ilc{gen_path} can generate interesting paths based on the
server state.  \autoref{fig:stategen} illustrates this state-based heuristics.
It refines the test harness box in \autoref{fig:overview}, and will be extended
with trace-based heuristics in the next subsection.

\subsection{Trace-based heuristics}
\label{sec:heuristic-trace}

When the SUT makes internal choices, \eg, generating ETags, the
specification represents them as symbolic variables.  These variables' concrete
values are not stored in the specification state, but may be observed during
execution.  For example, when an HTTP server responds to a GET request, it might
include the resource's ETag as shown in \autoref{sec:internal-nondeterminism}.

To improve the generator in \autoref{fig:naive-generator}, we can generate
interesting ETags based on the trace produced during execution.  The trace is a
list of packets sent and received by the tester, and the packets' payloads may
include responses that have an ETag field.  The \ilc{gen_etag} function
emphasizes ETags that were observed in the trace, which are more likely to match
those generated by the SUT:
\begin{coq}
  Definition gen_etag (trace: list concrete_packet) : IO string :=
    let etags: list string := tags_of trace in
    freq [(90, oneof etags);
          (10, random_etag)].
\end{coq}

To utilize this improved generator for ETags, the tester needs to record the
trace of packets sent and received.  This is done by modifying the \ilc{execute}
function in \autoref{fig:execute}, adding an accumulator called ``\ilc{trace}''
as the recursion parameter:
\begin{coq}
  Fixpoint execute (fuel: nat) (trace: list concrete_packet)
                   (m: itree tE void) : IO bool :=
    match fuel with
    | S fuel' =>
      match m with
      | Impure e k =>
        match e with
        | (ClientSend q|) => client_send q;;
                             execute fuel' (trace ++ [q]) (k tt)
        | (ClientRecv|)   => oa <- client_recv;;                             
                             let trace' := match oa with
                                           | Some a => trace ++ [a]
                                           | None   => trace
                                           end in
                             execute fuel' trace' (k oa)
        | (|GenPacket state|) => pkt <- gen_packet state trace;;
                                 execute fuel' trace (k pkt)
        ... (* similar to %\autoref{fig:execute}% *)
\end{coq}

When the tester sends or receives a packet, the packet is appended to the
runtime \ilc{trace}.  Then the \ilc{gen_packet} generator can take the trace
accumulated so far and feed it to the ETag generator:
\begin{coq}
  Definition gen_packet (state: sigma) (trace: list concrete_packet) :=
    target <- gen_path state;;
    etag   <- gen_etag trace;;
    ... (* same as %\autoref{fig:naive-generator}% *)
\end{coq}

\begin{figure}
\includegraphics[width=.7\linewidth]{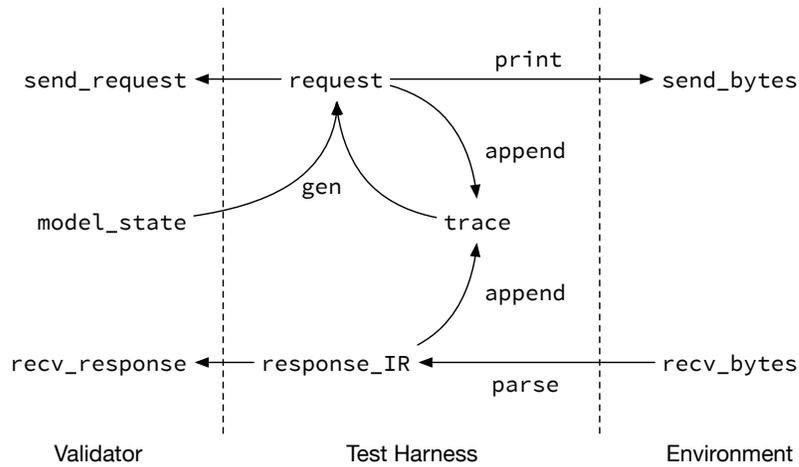}
\caption{Combining state-based and trace-based heuristics.}
\label{fig:gen}
\end{figure}

Now we can generate interesting test inputs by combining state-based and
trace-based heuristics, as shown in \autoref{fig:gen}.  In the next section,
I'll further extend this framework and shrink the test inputs while keeping them
interesting, addressing inter-execution nondeterminism.

\section{Shrinking Interactive Tests}
\label{sec:shrinking}
Suppose we have generated a test input that has caused invalid observations of
the SUT.  The generated counterexample consists of (1) {\em signal} that is
essential to triggering violations, and (2) {\em noise} that does not contribute
to revealing such violations.  We need to shrink the counterexample by removing
the noise and keeping the signal.

For interactive testing, the test input is a sequence of request messages.  An
intuitive way of shrinking is to remove some requests from the original sequence
and rerun the test.  However, rerunning an interesting request might produce
trivial results, due to inter-execution nondeterminism discussed in
\autoref{sec:inter-execution}.

To prevent turning signal into noise when rerunning the test, I shrink the
heuristics instead of shrinking the generated test input.
\autoref{sec:shrink-architecture} introduces the architecture for interactive
shrinking, then \autoref{sec:shrink-ir} explains the language design beneath
that addresses inter-execution nondeterminism.

\subsection{Architecture}
\label{sec:shrink-architecture}

I propose a generic framework for generating and shrinking interactive tests.
The key idea is to introduce an abstract representation for test inputs that
embeds trace-based heuristics.  When shrinking the counterexample, the test
harness picks a substructure of the abstract representation and computes the
corresponding test input using the new runtime trace.

For example, when generating a timestamp, instead of producing the concrete
value, \eg, ``\httpdate\today~\currenttime~GMT'', the generator returns an
abstract representation that says, ``use the timestamp observed in the last
response''.  When rerunning the test, the timestamp is computed from the new
trace, \eg, ``\httpdate\DayAfter~\currenttime~GMT''.

\begin{figure}
  \includegraphics[width=.85\textwidth]{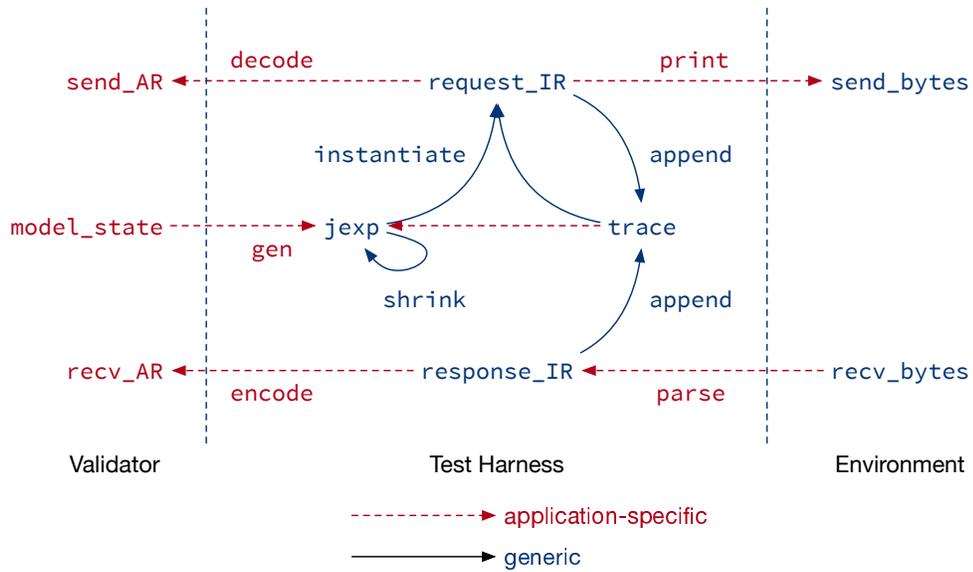}
  \caption{Complete architecture of the test harness.}
  \label{fig:shrink}
\end{figure}

The test generation and shrinking framework is shown in \autoref{fig:shrink}.
It involves four languages, from right to left:
\begin{enumerate}
  \item Byte representation, in which the tester interacts with the environment.
    This can be network packets, file contents, or other serialized data
    produced and observed by the tester.

  \item Intermediate representation (IR), a generic language that abstracts the
    byte representation as structured data.  The test harness {\em parses} byte
    observations and records its trace in terms of the IR, which allows
    representing trace-based heuristics as a generic language, \ie,
    J-expressions.

  \item J-expression (Jexp), a symbolic abstraction of the IR.  The IR
    corresponds to concrete inputs and outputs, whereas Jexp defines a
    computation from trace to IR.  The generator provides test inputs in terms
    of Jexps; The test harness {\em instantiates} the generated Jexps into
    request IR and {\em prints} them into byte representation.

    When shrinking test inputs, the test harness shrinks the sequence of Jexps.
    The shrunk Jexps are then instantiated by the new trace during runtime.

    The intermediate representation and J-expression will be further explained
    in \autoref{sec:shrink-ir}.

  \item Application representation (AR), including the request (\ilc Q),
    response (\ilc A), and state (\ilc S) types used for specifying the
    protocol.  Specification writers can choose the type interface at their
    convenience, provided the request and response types are embeddable into the
    IR.
\end{enumerate}

The testing framework implements protocol-independent mechanisms like recording
the trace and shrinking counterexamples, which correspond to the blue solid
arrows in \autoref{fig:shrink}.  It can be used for testing various protocols,
provided application-specific translations from IR to AR and between IR and
bytes, illustrated as red dotted arrows.  The test developer needs to tune the
generator that produces Jexps, encoding their domain knowledge as state-based
and trace-based heuristics.

\subsection{Abstract representation languages}
\label{sec:shrink-ir}
I choose JSON as the IR in this framework, which allows syntax trees to be
arbitrarily wide and deep and provides sufficient expressiveness for encoding
message data types in general.

\begin{figure}
\[\begin{array}{r@{\;}l}
\mathsf{JSON^T}\triangleq&\mathsf T\mid\{\mathsf{object^T}\}\mid[\mathsf{array^T}]\mid\mathsf{string}\mid\Int\mid\Bool\mid\mathsf{null}\\
\mathsf{object^T}\triangleq&\nil\mid\mathsf{``string": JSON^T,object^T}\\
\mathsf{array^T}\triangleq&\nil\mid\mathsf{JSON^T,array^T}\\
\mathsf{IR}\triangleq&\mathsf{JSON^{IR}}\\
\mathsf{Jexp}\triangleq&\mathsf{JSON^{\Jref{\Label}{\Jpath}{\Function}}}\\
&\text{where }\Label\in\Nat,\Function\in\mathsf{IR}\to\mathsf{IR}\\
\Jpath\triangleq&\This\mid\Jpath\Number\Index\mid\Jpath\At\Field\\
&\text{where }\Index\in\Nat,\Field\in\mathsf{string}
\end{array}\]
\caption{Intermediate representation and J-expression.}
\label{fig:ir-jexp}
\end{figure}

\begin{figure}
\begin{minipage}{.6\textwidth}
\begin{coq}
Notation   labelT := nat.
Definition traceT := list (labelT * IR).

Context q1 q2 a1 a2 : IR.
Example labelled_trace: traceT :=
  [(1, q1); (3, q2); (4, a2); (2, a1)].
\end{coq}
\end{minipage}\begin{minipage}{.3\textwidth}
  \includegraphics[width=\linewidth]{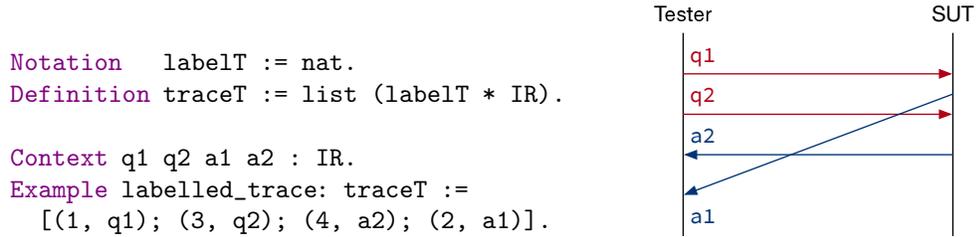}
\end{minipage}
\caption{Labelled trace example.}
\label{fig:ir-trace}
\end{figure}

\begin{figure}
  \begin{minipage}[t]{.4\textwidth}
\begin{json}
  (* a2 = *)
  {
    "files": [
      {
        "name": "foo",
        "mode": 755
      },
      {
        "name": "bar",
        "mode": 500
      }
    ],
    "exitCode": 0
  }
\end{json}
  \end{minipage}\begin{minipage}[t]{.5\textwidth}
\begin{coq}
(* Jpath syntax defined in %\autoref{fig:ir-jexp}%. *)
Example second_file_mode: jpath :=
  this @ "files" # 2 @ "mode".

Example mode_add_write (j: IR) : IR :=
  match j with
  | JSON_Number n =>
    JSON_Number (mode_bits_or 200 n)
  | _ => j
  end.

Example id (j: IR) : IR := j.
\end{coq}
  \end{minipage}
\vspace*{1em}
  \caption{IR, Jpath, and heuristics function example.}
  \label{fig:ir-jpath}
\end{figure}

\paragraph{Syntax}
The J-expression is an extension of JSON that can encode trace-based heuristics.
As shown in \autoref{fig:ir-jexp}, a Jexp may include syntax
$(\Label.\Jpath.\Function)$ that represents trace-based heuristics, specified as:
\begin{enumerate}
\item The test harness records the trace as a list of labelled messages, where
  the requests are labelled odd, and their responses are labelled as the next
  even number.  The $\Label$ in a Jexp locates the IR in the trace with which
  the heuristics computes the input.  Labelling messages allows the reproducing
  trace-based heuristics despite shrinking and inter-execution nondeterminism.

  For example, consider the trace in \autoref{fig:ir-trace}: If a trace-based
  heuristics is interested in \ilc{q2}'s response \ilc{a2}, then it can be
  encoded as ``compute the test input based on message labelled 4'':
\begin{coq}
  Context get_label: labelT -> traceT -> IR.

  Compute get_label 4 labelled_trace.
  (* = a2 : IR *)
\end{coq}
  
  Suppose the test input is shrunk by removing \ilc{q1}, the label for \ilc{q2}
  remains unchanged as 3, so label 4 corresponds to the new response to
  \ilc{q2}:
\begin{coq}
  Example new_trace: traceT :=
    [(3, q2); (4, a2')].

  Compute get_label 4 new_trace.
  (* = a2' : IR *)
\end{coq}

As a result, the trace-based heuristics are preserved and adapted to new
executions during the shrinking process.
\item The $\Jpath$ is a path in the IR's syntax tree and refers to a
  substructure of the IR that the heuristics uses.

  For example, suppose request \ilc{q2} lists files in a directory using the
  POSIX \inlinec{ls} command, and its response \ilc{a2} is encoded as the IR
  shown in \autoref{fig:ir-jpath}.  The response IR is a JSON object whose
  \inlinec{"files"} field is an array of objects, each has a \inlinec{"name"}
  and a \inlinec{"mode"} field.  A heuristic can refer to the second file's
  mode bits by Jpath \ilj{(this@"files"#2@"mode")}, which will guide the test
  harness to locate its corresponding value:
\begin{coq}
  Context get_jpath: jpath -> IR -> IR.

  Compute get_jpath second_file_mode a2.
  (* = JSON_Number 500 : IR *)
\end{coq}
\item The $\Function$ has type $(\mathsf{IR}\to\mathsf{IR})$, and defines the
  computation based on the sub-IR located by the Jpath.

  Consider the mode bits located in the previous example: If the heuristic
  wants to add write permission to the mode bits, it can do so with the
  \ilc{mode_add_write} function in \autoref{fig:ir-jpath}, which produces mode
  700.  Some heuristics might use the sub-IR 500 as-is, using the identity
  function \ilc{id}.
\end{enumerate}

\paragraph{Semantics}
J-expression provides a generic interface for test developers to implement
trace-based heuristics.  For the aforementioned file system example, the tester
can generate a request that changes the mode bits of an observed file, with the
following Jexp:
\begin{json}
  (* e5 = *)
  {
    "command": "chmod",
    "args":
      [ 4.(this@"files"#2@"mode").mode_add_write
      , 4.(this@"files"#2@"name").id ]
  }
\end{json}

To instantiate Jexps into request IR, the test harness substitutes all
occurences of $(\Jref{l}{p}{f})$ in the Jexp with its corresponding IR computed
from the runtime trace:
\begin{coq}
  Definition eval (l: labelT) (p: jpath) (f: IR -> IR) (t: traceT) : IR :=
    let a: IR := get_label l t in
    let j: IR := get_jpath p a in
    f j.
\end{coq}

For example, given the runtime trace in \autoref{fig:ir-trace}, with \ilc{a2} is
defined in \autoref{fig:ir-jexp}, the the above Jexp is instantiated into the
following request:

\begin{json}
  (* instantiate e5 labelled_trace = *)
  { 
    "command": "chmod",
    "args": [ 700, "bar" ]
  }
\end{json}

However, when rerunning the test, the \ilc{new_trace} has a different response
associated with label 4.  The new response \ilc{a2'} might have fewer than 2
files in its payload.  Moreover, the response \ilc{a2'} might have not appeared
in the trace, due to delays in the environment.

To instantiate the original Jexp in such situations, I loosen the
\ilc{get_jpath} and \ilc{get_label} functions when evaluating the heuristics:
\begin{enumerate}
\item When evaluating a Jpath starting with \ilj{p#n}, if \ilj p corresponds to
  an array with fewer than \ilj n elements, or the array's \ilj n-th element
  cannot properly evaluate the remaining path, then try continuing the
  evaluation with any other element in the array.

  For example, consider evaluating \ilj{(this@3#"bar")} on the following IR's:
  \begin{multicols}{2}
\begin{json}
  (* j2 = *)
  [
    { "foo": 21 },
    { "bar": 22 }
  ]
\end{json}
\columnbreak
\begin{json}
  (* j3 = *)
  [
    { "bar": 31 },
    { "baz": 32 },
    { "foo": 33 }
  ]
\end{json}
  \end{multicols}

  Here \ilj{j2} doesn't have a third element, and \ilj{j3}'s third element
  doesn't have field \ilj{"bar"}.  In these cases, \ilj{get_jpath} chooses other
  elements in the two arrays, resulting in value \ilj{22} for \ilj{j2}, and
  \ilj{31} for \ilj{j3}.
  
\item When evaluating label \ilj l and Jpath \ilj p on a trace, if the message
  labelled \ilj l does not exist in the trace, or cannot evaluate Jpath \ilj p
  properly, then try continuing the evaluation with any other IR in the trace.

  For example, consider evaluating J-expression \ilj{6.(this#2@"foo").id} on the
  following traces:
\begin{coq}
  Definition t1: traceT :=
    [(1,q1); (2,j2); (5,q2)].

  Definition t2: traceT :=
    [(3,q1); (4,j3); (5,q2); (6,a2)].
\end{coq}

Here \ilj{t1} doesn't have a message labelled 6, probably caused by environment
delays; \ilj{t2} has label 6 but its corresponding message is an object rather
than an array expected by the Jexp.  In these cases, \ilc{eval} chooses other
messages in the trace to evaluate, resulting in value \ilc{21} for \ilc{t1}, and
\ilc{33} for \ilc{t2}.
\end{enumerate}

By introducing loose evaluation of J-expressions, my test harness allows partial
instantiation of heuristics when the runtime trace is less than satisfying.

So far I have shown how to generate and shrink interactive test inputs and
address inter-execution nondeterminism.  In the next chapter, I'll combine this
test harness design with the validator practice in \autoref{chap:practices}, and
evaluate these techniques by testing real-world systems like HTTP servers and
file synchronizers.

\chpt{Evaluation}
\label{chap:eval}
This chapter evaluates the testing methodology presented in this thesis, by
deriving testers from specifications and running them against systems under
test.  

I conduct the experiments on two kinds of systems, \http server
(\autoref{sec:http}) and file synchronizer (\autoref{sec:sync}).  The research
question is whether the tester can reveal the SUT's violation against the
specification.

\section{Testing Web Servers}
\label{sec:http}

This thesis is motivated by the Deep Specification project~\cite{deepspec},
whose goal is to build systems with rigorous guarantees of functional
correctness, studying HTTP as an example.  I formalized a subset of \http
specification, featuring WebDAV requests GET, PUT, and POST~\cite{rfc4918}, ETag
preconditions~\cite{rfc7232}, and forward proxying~\cite{rfc7231}.

From the protocol specification written as ITrees, I derive a tester client
that sends and receives network packets.  \autoref{sec:http-sut} explains the
system under test and the experiment setup.  \autoref{sec:http-qual} and
\autoref{sec:http-quant} then describe the evaluation results qualitatively and
quantitatively.

\subsection{Systems Under Test}
\label{sec:http-sut}
I ran the derived tester against three server implementations:

\begin{itemize}
\item Apache HTTP Server~\cite{Apache}, one of the most popular servers
  on the World Wide Web~\cite{http-netcraft,http-stats}.  I used the latest
  release 2.4.46, and edited its configuration file to enable WebDAV and proxy
  modules.
\item Nginx~\cite{nginx}, the other most popular server.  The experiment was
  conducted on the latest release 1.19.10, with only WebDAV module enabled,
  because Nginx doesn't fully support forward proxying like Apache does.
\item DeepWeb server developed in collaboration with \citet{itp21}, supporting
  GET and POST requests.  The server's functional correctness was formally
  verified in Coq.
\end{itemize}

The tests were performed on a laptop computer (with Intel Core i7 CPU at 3.1
GHz, 16GB LPDDR3 memory at 2133MHz, and macOS 10.15.7).  The Apache and Nginx
servers were deployed as Docker instances, using the same host machine as the
tester runs on.  Our simple DeepWeb server was compiled into an executable
binary, and also ran on localhost.

The tester communicated with the server via POSIX system calls, in the same way
as over the Internet except using address \inlinec{127.0.0.1}.  The round-trip
time (RTT) of local loopback was $0.08\pm0.04$ microsecond (at 90\% confidence).

\subsection{Qualitative Results}
\label{sec:http-qual}
\paragraph{Apache}
My tester rejected the Apache HTTP Server, which uses strong comparison
(Page \pageref{foot:etag}) for PUT requests conditioned
over \inlinec{If-None-Match}, while RFC~7232 specified that
\inlinec{If-None-Match} preconditions must be evaluated with weak comparison.  I
reported this bug to the developers and figured out that Apache was conforming
with an obsoleted \http standard~\cite{rfc2616}.  The latest standard has
changed the semantics of \inlinec{If-None-Match} preconditions, but Apache
didn't update the logic correspondingly.

To further evaluate the tester's performance in finding other violations, I
fixed the precondition bug by deleting 13 lines of source code and recompiling
the container.

The tester accepted the fixed implementation, which can be explained in two
ways: (1) The server now complies with the specification, or (2) The server has
a bug that the tester cannot detect.  To provide more confidence that (1) is the
case, I ran the tester against servers that are known as buggy and expected the
tester to detect all the intentional bugs.

The buggy implementations were created by mutating the Apache source code
manually\footnote{I didn't use automatic mutant generators because (i) Existing
tools could not mutate all modules I'm interested in; and (ii) The automatically
generated mutants could not cause semantic violations against my protocol
specification.} and recompiling the server program.  I created 20 mutants whose
bugs were located in various modules of the Apache server: \inlinec{core},
\inlinec{http}, \inlinec{dav}, and \inlinec{proxy}.  Some bugs appeared in the
control flow, \eg, removing a return statement in the request handler (as shown
in \autoref{fig:mutant-example}) or skipping the check of preconditions.  Others
appeared in the data values, \eg, calling functions with wrong parameters,
flipping bits in computations, accessing buffer off by one byte, etc.
\begin{figure}
\begin{cpp}
static int default_handler(request_rec *r) {
    ...
    if (r->finfo.filetype == APR_NOFILE) {
        ap_log_rerror(APLOG_MARK, APLOG_INFO, 0, r, APLOGNO(00128)
                      "File does not exist: %s",
                      apr_pstrcat(r->pool, r->filename, r->path_info, NULL));
        // return HTTP_NOT_FOUND;
    }
    ...
\end{cpp}
\vspace*{1em}
\caption{Server mutant example, created by commenting out a line of code.}
\label{fig:mutant-example}
\end{figure}

The tester rejected all of the 20 mutants.  Some mutants took the tester a
longer time to reveal than others, which will be discussed in
\autoref{sec:http-quant}.

\paragraph{Nginx}
When testing Nginx, I found that its WebDAV module did not check the
\inlinec{If-Match} and \inlinec{If-None-Match} preconditions of PUT requests.  I
then browsed the Nginx bug tracker and found a ticket opened by
\citet{nginx242}, reporting the same issue with \inlinec{If-Unmodified-Since}
preconditions.

This issue has been recognized by the developers in 2016 but never resolved.  To
fix this bug, Nginx needs to redesign its core logic and evaluate the request's
precondition {\em before}---instead of {\em after}---handling the request
itself.  As a result, I tested mutants only for Apache, whose original violation
was fixed by simply removing a few lines of bad code.

\paragraph{DeepWeb}
My test derivation framework was developed in parallel with the DeepWeb server.
After my collaborators finished the formal proof of the server's functional
correctness, I tested the server with my derived tester.  The tester
revealed a liveness issue---when a client pipelines more than one requests in a
single connection, the server may hang without processing the latter requests.

This liveness bug was out of scope for the server's functional correctness,
which only requires the server not to send invalid messages.  Such partial
correctness may be trivially satisfied by a silent implementation that never
responds.  The bug was revealed by manually observing the flow of network
packets, where the tester kept sending requests while the server never
responded.  My experiments have shown the complementarity between testing and
verification for improving software quality.

These results show that my tester is capable of finding different kinds of bugs
in server implementations, within and beyond functional correctness.  Next, I'll
evaluate how long the tester takes to reveal bugs.

\subsection{Quantitative Results}
\label{sec:http-quant}

\begin{figure*}
  \includegraphics[width=\textwidth]{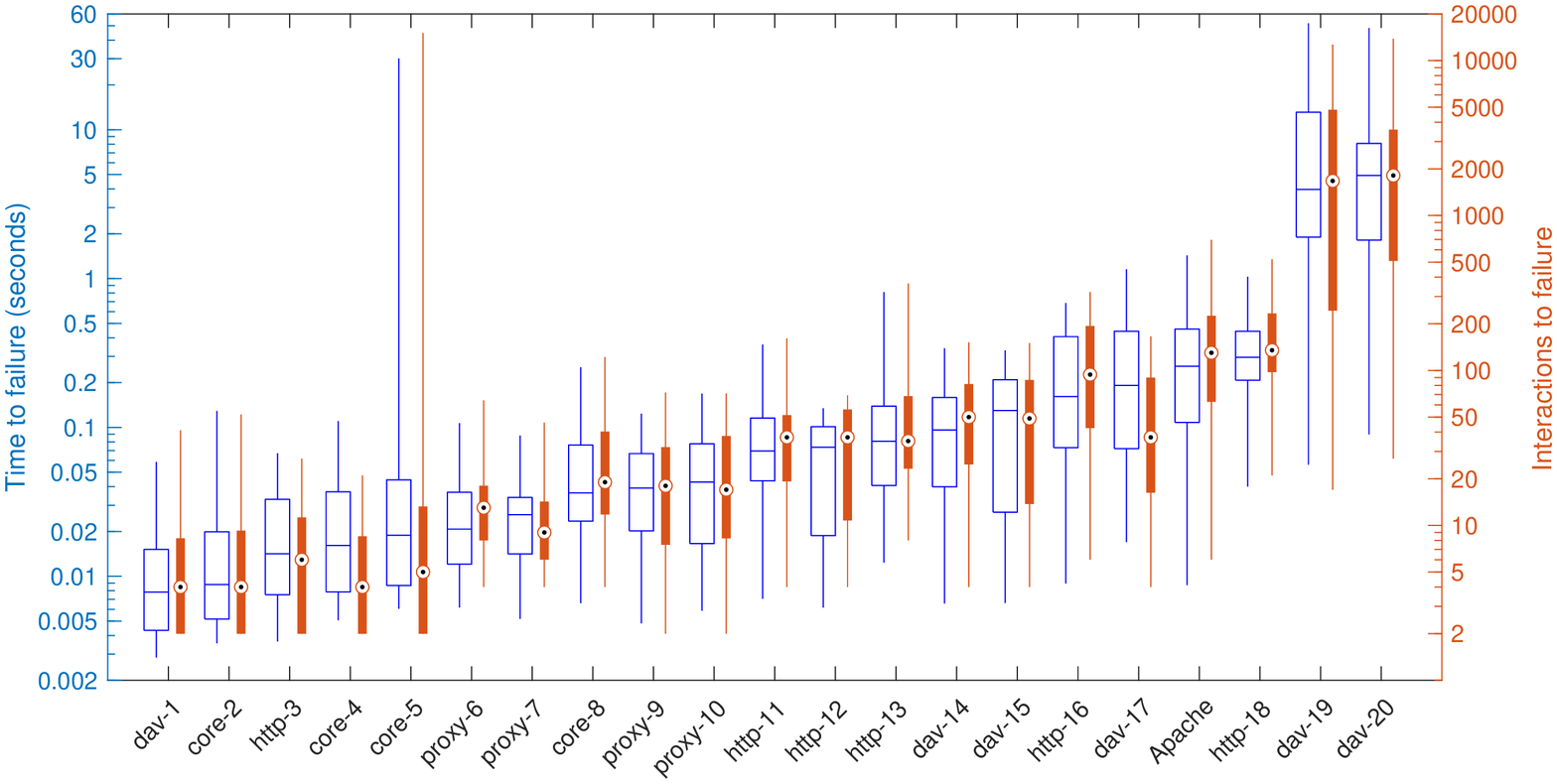}
  \caption[Cost of detecting bug in each server/mutant.]{Cost of detecting bug
    in each server implementation.  The left box with median line is the
    tester's execution time before rejecting the server, which includes
    interacting with the server and checking its responses.  The right bar with
    median circle is the number of \http messages sent and received by the
    tester before finding the bug.  Results beyond 25\%--75\% are covered by
    whiskers.
  }
  \label{fig:checker-performance}
\end{figure*}

To answer the research question at the beginning of this chapter quantitatively,
I measured the execution time and network interactions taken to reject vanilla
Apache and its mutants, as shown in \autoref{fig:checker-performance}.  The 20
mutants are named after the modules where I inserted the bugs.  The tester
rejected all the buggy implementations within 1 minute, and in most cases,
within 1 second.  This does not include the time for shrinking counterexamples.

Some bugs took longer time to find, and they usually required more interactions
to reveal.  This may be caused by (1) The counterexample has a certain pattern
that my generator didn't optimize for, or (2) The tester did produce a
counterexample, but failed to reject the wrong behavior.  I determine the real
cause by analyzing the bugs and their counterexamples:

\begin{figure}
  \includegraphics[width=.9\linewidth]{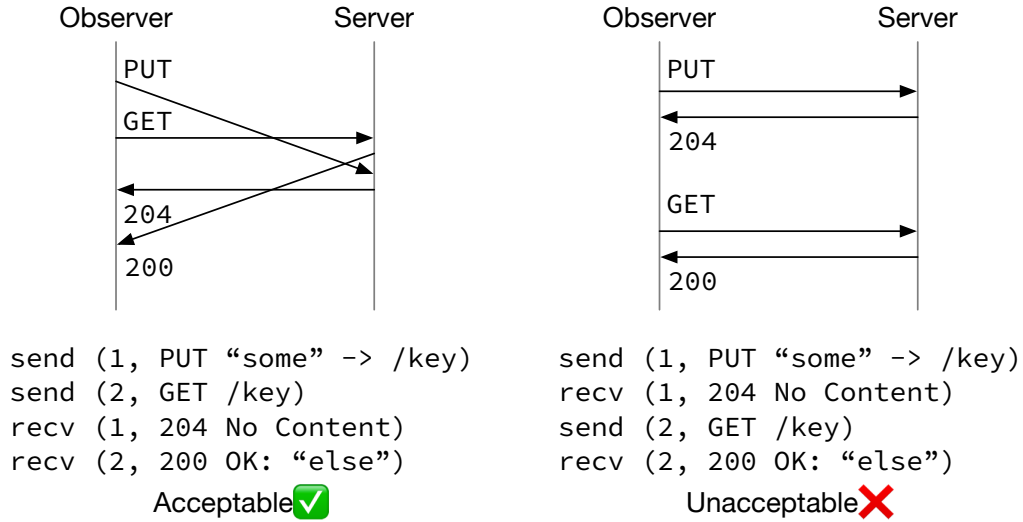}
  \caption[GET-after-PUT bug pattern in Apache mutants.]{GET-after-PUT bug
    pattern in Apache mutants.  The trace on the left does not convince the
    tester that the server is buggy, because there exists a certain network
    delay that explains why the PUT request was not reflected in the 200
    response.  When the trace is ordered as shown on the right, the tester
    cannot imagine any network reordering that causes such observation, thus
    must reject the server.}
  \label{fig:put-bug}
\end{figure}
  
\begin{itemize}
  \item Mutants 19 and 20 are related to the WebDAV module, which handles PUT
    requests that modify the target's contents.  The buggy servers wrote to a
    different target from that requested but responds a successful status to
    the client.

    The tester cannot tell that the server is faulty until it queries the
    target's latest contents and observes an unexpected value.  To reject the
    server with full confidence, these observations must be made in a certain
    order, as shown in \autoref{fig:put-bug}.

  \item Mutant 18 is similar to the bug in vanilla Apache: the server should
    have responded with 304 Not Modified but sent back 200 OK instead.  To
    reveal such a violation, a minimal counterexample consists of 4 messages:
    \begin{enumerate}
    \item GET request,
    \item 200 OK response with some ETag \inlinec{"x"},
    \item GET request conditioned over \inlinec{If-None-Match: "x"}, and
    \item 200 OK response, indicating that the ETag \inlinec{"x"} did not match
      itself.
    \end{enumerate}
    Note that (2) must be observed before (3), otherwise the tester will not
    reject the server, with a similar reason as in \autoref{fig:put-bug}.

  \item Mutant 5 is the one shown in \autoref{fig:mutant-example}.  It causes
    the server to skip the return statement when the response should be 404 Not
    Found.  The counterexample can be as small as one GET request on a
    non-existent target, followed by an unexpected response like 403
    Forbidden.  However, my tester generates request targets within a small
    range, so the requests' targets are likely to be created by the tester's
    previous PUT requests.

    Narrowing the range of test case generation might improve the performance in
    aforementioned Mutants 18--20, but Mutant 5 shows that it could also degrade
    the performance of finding some bugs.

  \item The mutants in the proxy module caused the server to forward wrong
    requests or responses.

    To test servers' forward proxying functionality, the tester consists of
    clients and origin servers, both derived by dualization.  When the origin
    server part of the tester accepts a connection from the proxy, it does not
    know for which client the proxy is forwarding requests.  Therefore, the
    tester needs to check the requests sent by all clients, and make sure none
    of them matches the forwarded proxy request.

    The more client connections the tester has created, the longer it takes the
    tester to check all connections before rejecting a buggy proxy.
\end{itemize}

These examples show that the time-consuming issue of some mutants are likely
caused by the generators' heuristics.  Cases like Mutant 5 can be optimized by
state-based heuristics in \autoref{sec:heuristic-state}; Proxy-related bugs can
be more easily found by trace-based heuristics in \autoref{sec:heuristic-trace};
For Mutants 18--20, the requests should be sent at specific time periods so that
the resulting trace is unacceptable per specification, possibly by integrating
packet dynamics~\cite{pkt-dyn} in the test executor.

\section{Testing a File Synchronizer}
\label{sec:sync}

To demonstrate the generality of my specification-based testing methodology, I
also applied it to file synchronizers.  \autoref{sec:file-sut} introduces my
specification of the file system and synchronization semantics.  From these
specifications, I derived a tester program for the Unison file
synchronizer~\cite{unison}, with results shown in \autoref{sec:file-result}.

\subsection{System Under Test}
\label{sec:file-sut}
A file synchronizer manipulates the file system to reconcile updates in
different replicas~\cite{what-sync}.  To check a synchronizer's correctness, the
tester needs to update replicas, launch the synchronization process, and observe
the propagated updates.

\begin{figure}
\begin{coq}
  Inductive node :=
    File      : content          -> node
  | Directory : list (name*node) -> node.

  Context read : path -> node    -> option content.
  Context write: path -> content -> node -> option node.
  Context mkdir: path -> node    -> option node.
  Context ls   : path -> node    -> list name.
  Context rm   : path -> node    -> option node.
\end{coq}
\vspace*{1em}
\caption{File system specification.}
\label{fig:file-spec}
\end{figure}

My specification consists of two parts:
\begin{enumerate}
  \item A file system model represented as a tree, where the leaves are files,
    and the branches are directories.  As shown in \autoref{fig:file-spec}, the
    file system model is a simplified abstraction from the POSIX file interface,
    ignoring metadata and file permissions.  Specifying more aspects of the file
    system is discussed in \autoref{chap:discussion}.

    Based on the file system model, I specified five basic file operations that
    the tester may perform: (i) reading contents from a file, (ii) writing
    contents to a file, (iii) creating a new directory, (iv) listing entries
    under a directory, and (v) removing a file or directory recursively.

    Some file operations may fail, \eg, when reading from a path that
    refers to a directory.  These failures are represented as return value
    \ilc{(None: option _)} in the node functions.
  \item A file synchronizer model that syncs updates between two replicas,
    implementing the reconcilation algorithm by \citet{what-sync}:
\begin{coq}
  Definition sigma := node * node * node.
      
  Context reconcile: sigma -> sigma.
\end{coq}
Note that the \ilc{reconcile} function manipulates three replicas.  This is
because the synchronizations might be partial: Upon write-write and write-delete
conflicts, the synchronizer does not propagate the conflicting updates, and
leaves the dirty files untouched in both replicas.

The third parameter of the reconcile function represents the subset of the two
replicas that were synchronized: \ilc{(reconcile (a,b,g))} syncs replicas \ilc a
and \ilc b based on their previous consensus \ilc g.  The consensus is initially
empty and updated when a change in one replica is propagated to another, or the
two replicas have made identical changes.
\end{enumerate}

\begin{figure}
\begin{coq}
    Variant F :=        (* file operations *)
    Fls    (p: path)
  | Fread  (p: path)
  | Fwrite (p: path) (c: content)
  | Fmkdir (p: path)
  | Frm    (p: path).

  Variant R := R1 | R2. (* replicas      *)

  Variant Q :=          (* query type    *)
    QFile (r: R) (f: F)
  | QSync.

  Variant A :=          (* response type *)
    Als   (l: list name)
  | Aread (c: content)
  | Aexit (z: Z).       (* exit code     *)
\end{coq}
\vspace*{1em}
\caption{Query and response types for testing file synchronizers.}
\label{fig:file-type}
\end{figure}

Having specified the file system interface and the reconciliation semantics, I
modelled the SUT as a deterministic QA server described in
\autoref{def:qaserver}.  As shown in \autoref{fig:file-type}, the request type
\ilc Q can be file operations or the synchronization command; the response type
\ilc A carries the return value of file system queries, or the transactions'
exit code.  For example, when synchronizing the two replicas, code 1 indicates
partial synchronization with conflicts unresolved, and code 2 means the
synchronizer has crashed due to uncaught exceptions or interruptions.

The QA model for the file system+synchronizer was dualized into a tester program
that makes system calls to manipulate files, launch the synchronizer, and
observe the updates.  The system calls are made one at a time, and the file
synchronizer is run as a foreground process that blocks other interactions.
Testing the synchronizer as a nonblocking background process is discussed in
\autoref{chap:discussion}.

\subsection{Experiment Results}
\label{sec:file-result}
The tester rejected the Unison file synchronizer in two ways, and I reported
the counterexample to the developers.  By analyzing the program's behavior, we
determined that one rejection was a valid but defective feature, and the other
rejection was a documented feature not included in my specification.  The
revealed features are as follows:

\paragraph{Synchronizing read-only directories}
When the tester creates a directory in one replica with read and execution
permission (mode \ilj{500}) and calls the synchronization command, Unison
crashed without creating the corresponding directory in the other replica.

The crashing behavior only occurs on macOS, and is caused by Unison's mechanism
for propagating changes: When copying directory \ilj{foo} from replica \ilj A to
replica \ilj B, the synchronizer first creates a temporary directory
``\ilj{B/.unison.foo.xxxx.unison.tmp}'', and then renames it to ``\ilj{B/foo}''.
The \inlinec{rename} implementation in macOS requires write permission to
proceed, so the change was not propagated.

This issue is not considered a violation in Unison or macOS, because: (1) Unison
is allowed to halt without propagating an expected change, as long as its exit
code has indicated an error, and no unexpected change was propagated.  (2) The
POSIX specification~\cite{posix} says the \inlinec{rename} function {\it may}
require write access to the directory.

Despite being compliant to the specification, this feature in Unison is
considered a defect, as it disables synchronization of read-only files and
directories.  A potential fix might be substituting \inlinec{rename} with other
system calls.

This defect was revealed by accident: My file system specification in
\autoref{fig:file-spec} does not mention file permissions, so I defaulted to
mode \ilj{755}.  However, when implementing the test executor, I made a typo
that wrote the permission as hexadecimal \ilc{0x755} while it should be octal
\ilc{0o755}.  This caused the created directory to have mode bits \ilj{525},
which triggered the aforementioned behavior.

This experiment shows that my current abstraction of the file system is worth
expanding to include more information like file permissions, which might reveal
other features of the file synchronizer.

The experiment also reveals a challenge in specifying programs, that the
underlying operating system might also pose uncertainty to the program's
behavior.  Such external nondeterminism may be handled by parameterizing the
program's specification over the OS's, in a similar way as composing the server
model with the network model in \autoref{sec:external-nondet}.

\paragraph{Detecting write-delete conflict}
Suppose replicas \ilj A and \ilj B have a synchronized file \ilj{foo.txt}.  If
the tester deletes \ilj{A/foo.txt} and writes to \ilj{B/foo.txt}, then there is
a conflict between the two replicas.  My specification required Unison to
indicate this conflict with exit code 1.  During the tests, Unison did not
notice that \ilj{B/foo.txt} was changed, decided to propagate the deletion from
replica \ilj A to replica \ilj B, and halted with exit code 2, representing
``non-fatal failures occured during file transfer''.

This behavior is caused by Unison's ``fastcheck'' feature that improves the
performance at the risk of ignoring conflicts.  With fastcheck enabled by
default on *nix systems, the synchronizer detects file modifications by checking
if their timestamps have changed.  When the tester writes to the file within a
short interval, the timestamps might remain unchanged, so the synchronizer
treats the written file as unmodified.  Such behavior can be avoided by updating
the file contents after a longer interval or turning of the preference.

Note that although Unison decided to propagate the file deletion, it does
check the file contents before deleting it.  As a result, the Unison process
crashed and complained that \ilj{B/foo.txt} was modified during synchronization.

If the fastcheck feature was disabled, then Unison can detect the conflict by
comparing the file contents bit-by-bit.  It will propagate no updates and leave
the two replicas as-is, which results the same as enabling fastcheck but at
lower performance.

To include the space of behaviors with fastcheck enabled or upon failing file
transfers, I modified the synchronization semantics in two ways: (1) Instead of
computing the set of all dirty paths and synchronizing all of them, allow
choosing any subset of the dirty paths to synchronize.  (2) If there is a
conflict in the chosen paths, allow halting the synchronizer with exit code 1
(conflict detected) or 2 (propagation failed).  The loosened specification
accepted all behavior it observed from Unison.

These experiments show that my testing methodology can effectively reject
implementations that do not comply with the specification.  The incompliance
indicates that either (i) the implementation does not conform with the protocol,
\eg, Apache and Nginx violating RFC 7232, or (2) the specification is
incorrect, \eg, my initial file synchronizer specification not modeling all
features of Unison.  Automatically detecting the incompliances helps developers
to correctly specify and implement systems in real-world practices.

\chpt{Related Work}
\label{chap:related-work}
This chapter explores methodologies in specifying and automatically testing
interactive systems.
\autoref{sec:related-validate} compares different specification
techniques for testing purposes.
\autoref{sec:related-harness} then discusses practices in developing test
harnesses.

\section{From Specifications to Validators}
\label{sec:related-validate}

Different testing scenarios exhibit various challenges that motivate the
specification design.  This section partitions the validation techniques by the
languages used for developing the specifications.

\subsection{State machine specification: Quviq QuickCheck}
Property-based testing with QuickCheck has been well adopted in industrial
contexts~\cite{Hughes2016}.  The specification is a boolean function over traces,
\ie, the validator.  My solutions to addressing internal and external
nondeterminism are inspired by practices in QuickCheck.

\paragraph{Internal nondeterminism}
My \http specification was initially written as a QuickCheck property.  Before
handling preconditions like \inlinec{If-Match} and \inlinec{If-None-Match}, the
validator implemented a deterministic server model and compares its behavior
with the observations, as shown in the \ilc{validate} function in
\autoref{sec:interactive-testing}.  When expanding the specification to cover
conditional requests, I implemented the ad hoc validator by manually translating
the trace into tester-side knowledge, as shown in \autoref{fig:etag-tester}.

The complexity in describing ``what behavior is valid'' motivated me to rephrase
the specification.  I applied the idea of model-based
testing~\cite{broy2005model}, and specified the protocol in terms of ``how to
produce valid behavior''.  My specification represented internal nondeterminism
with symbolic variables.  The validator then checks whether the trace is
producible by the symbolic specification, by reducing the producibility problem
to constraint solving.

\paragraph{External nondeterminism}
\citet{testing-dropbox} have used QuickCheck to test Dropbox.  The specification
does not involve internal nondeterminism but does handle external
nondeterminism that local nodes may communicate with the server at any time.
This is done by introducing ``conjectured events'' to represent the possible
communications.  The validator checks if the conjectured events can be inserted
to somewhere in the trace to make it producible by the model.

To specify servers' allowed observations delayed by the network, \citet{cpp19}
introduced the concept of ``network refinement''.  The network may scramble the
traces by delaying some messages after others, with one exception: If the client
has received response \ilc A before it sends request \ilc Q, then by causality,
the server-side trace must have sent response \ilc A before receiving request
\ilc Q.  Upon observing a client-side trace \ilc T, our QuickCheck validator
searches for a server-side trace that (i) can be reasonably scrambled into trace
\ilc T, and (ii) complies with the protocol specification.

My network model design in \autoref{sec:external-nondet} was inspired by these
ideas of conjecturing the environment's behavior.  Instead of inserting
conjectured communication events or reordering the trace, I specified the
network as an independent module and composed it with the server specification.
This provides more flexibility in specifying asynchronous systems: (i) To adjust
the network configurations, \eg, limiting the buffer size or the number of
concurrent connections, we only need to adjust the network model without
modifying the server's events.  (ii) Instead of carefully defining the space of
valid scrambled traces for each network configuration, we can derive it from the
network model by dualization.  (iii) The network model is reusable and allows
specifying various protocols on top of it.

\subsection{Process algebra: LOTOS and TorXakis}
Language of termporal ordering specification (LOTOS)~\cite{lotos} is the ISO
standard for specifying OSI protocols.  It defines distributed concurrent
systems as {\em processes} that interact via {\em channels}.  Using a formal
language strongly inspired by LOTOS, \citet{torxakis-dropbox} implemented a test
generation tool called TorXakis and used it for testing DropBox.

TorXakis supports internal nondeterminism by defining a process for each
possible value.  This requires the space of invisible values to be finite.  In
comparison, I represented invisible values as symbolic variables and employed
constraint solving that can handle inifitite space of data like strings and
integers.

As for external nondeterminism, TorXakis hardcodes all channels between each
pair of processes, assuming no new process joins the system.  Whereas in my
network model, ``channels'' are the ``source'' and ``destination'' fields of
network packets, which allows specifying a server that exposes its port to
infinitely many clients.

\subsection{Transition systems: NetSem and Modbat}
Using labelled transition systems (LTS), \citet{netsem} have developed rigorous
specification for TCP, UDP, and the Sockets API.  To handle internal
nondeterminism in real-world implementations, they used symbolic the model
states, which are then evaluated with a special-purpose symbolic model checker.
The tester helped reveal their post-hoc specifications' mismatch with mainstream
systems like FreeBSD, Linux, and WinXP.  I borrowed the idea of symbolic
evaluation in validating observations and used it for detecting mainstream
servers' violations against the formalized RFC specification.

\citet{modbat} have generated test cases for Java network API, which involves
blocking and non-blocking communications.  Their abstract model was based on
extended finite state machines (EFSM) and could capture bugs in the network
library \verb|java.nio|.  Their validator rejects the SUT upon unexpected
exceptions or observations that fail its {\em encoded} assertions.  In
comparison, assertions in my validator are {\em derived} from the abstract
model, which covers full functional correctness of the SUT modulo external
nondeterminism.

\section{Test Harnesses}
\label{sec:related-harness}
This section explores techniques of generating and shrinking test inputs.
\autoref{sec:related-gen} compares different heuristics used by test generators;
\autoref{sec:related-shrink} explains existing shrinking techniques for
interactive testing scenarios.

\subsection{Generator Heuristics}
\label{sec:related-gen}

In addition to state-based and trace-based heuristics discussed in
\autoref{sec:heuristics}, other kinds of heuristics can be applied to generating
inputs for various testing scenarios.

\paragraph{Constraint-based heuristics}
The reason for introducing heuristics is to increase the chance of triggering
invalid behavior.  I specified the heuristics as ``how to produce interesting
input'', while, in some cases, it's more convenient to specify ``what inputs are
interesting''.  For example, well-typed lambda expressions can be easily defined
in terms of typing rules but are less intuitive to enumerate by a generator.

Narrowing~\cite{narrowing} allows generating data that satisfy certain
constraints.  \citet{gengood} have applied the idea of narrowing to the
QuickChick testing framework in Coq~\cite{quickchick}, representing constraints
as inductive relations.  The relations are compiled into efficient generators
that produce satisfying data.

The narrowing-based generator in QuickChick was used during my preliminary
experiments with \http, where I defined a well-formed relation for HTTP requests
to guide the generator.  However, the QuickChick implementation was unstable and
failed to derive generators as my type definition becomes more and more complex.
This constraint-based heuristics may be integrated to my current testing
framework by interpreting QuickChick's \ilc{Gen} (generator) monad into
the \ilc{IO} monad used by the test harness, provided the generator derivation
failure gets resolved.

\paragraph{Coverage-based heuristics}
Another strategy to increase the chance of revealing invalid behavior is to
cover more execution paths of the SUT.  This idea is applied to fuzz
testing~\cite{fuzz}, with popular implementations AFL~\cite{afl} and
Honggfuzz~\cite{honggfuzz}, and combined with property-based testing by
\citet{fuzzchick}.

Coverage-based testers mutate the test inputs and observe the programs'
execution paths.  An input is considered interesting if it causes the program to
traverse a previously unvisited path.  The interesting inputs will be mutated
and rerun to cover potentially more paths.

To track the program's execution paths, coverage-based heuristics need to
instrument the SUT during compilation, making it inapplicable for black-box
testing.  When fuzzing interactive systems like web servers, the SUT is run in a
simulator where the requests are provided as files instead of network packets.
The responses are ignored by the comtemporary fuzzing tools, which only checks
whether the SUT crashes or hangs and does not care about functional correctness
like the responses' contents.  To produce a trace for validation, the SUT needs
to be carefully modified to record its send and receive events.  Addressing
these challenges would enable coverage-based heuristics to be integrated in
interactive testing of systems' functional correctness.

\subsection{Shrinking Interactive Tests}
\label{sec:related-shrink}

To address inter-execution nondeterminism in Quviq QuickCheck,
\citet{Hughes2016} introduced the idea of shrinking abstract representations of
test input.  The tester first generates a {\em script} using state-based
heuristics and instantiates the script with the tester's runtime state.  The
scripts can be shrunk and adapted to new runtimes when rerunning the test.

The language design of my test harness is inspired by the QuickCheck approach,
and extends it in two dimensions:
\begin{enumerate}
\item The QuickCheck framework assumes synchronous interactions, where the
  requests are function calls that immediately return.  When testing
  asynchronous systems, \eg, networked servers, the responses might be
  indefinitely delayed by the environment, which would block the QuickCheck
  tester's state transition.
  
  To generate test inputs asynchronously, I implemented a non-blocking algorithm
  for instantiating scripts into requests.  When a dependant observation is
  absent due to delays or loss by the environment, the test harness tries to
  instantiate the request using other observations, instead of waiting for the
  observation from the environment.

\item QuickCheck requires the test developer to specify a runtime state to guide
  the generator, which requires them to define the state transition rules for
  each interaction.  The heuristics are implemented based on the tester's point
  of view.

  In comparison, my framework replaced these requirements with state- and
  trace-based heuristics, where the model state was exposed by the underlying
  specification, and the trace was automatically recorded by the test harness.
  This allows test developers to design heuristics based on the implementation's
  point of view, instead of reasoning on the implementation's internal and
  external nondeterminism based on its observations.
\end{enumerate}

\chpt{Conclusion and Future Work}
\label{chap:discussion}
This dissertation presented a systematic technique for testing interactive systems
with uncertain behavior.  I proposed a theory of dualizing protocol
specification into validators, with formal guarantee of soundness and
completeness (\autoref{chap:theory}).  To test systems in real-world practice, I
applied the dualization theory to the interaction tree specification language,
and derived specifications into interactive testing programs
(\autoref{chap:practices}).  I then presented a test harness design to generate
and shrink interactive test inputs effectively (\autoref{chap:harness}).  The
entire methodology was evaluated by detecting programs' incompliance with the
specification using automatically derived testers (\autoref{chap:eval}).

To address challenges posed by internal, external, and inter-execution
nondeterminism, I introduced various flavors of symbolic abstract
interpretation.  Systems' internal choices are represented as symbolic variables
and unified against the tester's observations (\autoref{sec:dualize-prog},
\autoref{sec:internal-nondet}).  Possible impacts made by the environment are
represented as nondeterministic branches (\autoref{sec:external-nondet}) and
determined by backtrack searching (\autoref{sec:backtrack}).  The test inputs
are generated as symbolic intermediate representations that can adapt to
different traces during runtime (\autoref{sec:shrinking}).

The techniques in this thesis can be expanded in several dimensions:
\paragraph{Specifying and testing in various scenarios}
The Unison file synchronizer in \autoref{sec:sync} was specified and tested as a
command line tool that is manually triggered.  Whereas, Unison can run as a
background process that monitors the file system and automatically propagates
the updates, just like Dropbox.  Running in the background introduces more
nondeterminism, \eg, whether the tester's modification has been buffered or
propagated to the file system, whether the operating system has scheduled the
synchronizer process to check the file system's updates, whether the
synchronizer has finished propagating the changes to the other replica, etc.
Migrating testing scenarios might expose the limitations of my
methodology or introduce new challenges in testing.

\paragraph{Integrating other testing techniques to the framework}
The experiments with Apache and Unison have shown that some bugs are
time-related, \eg, only revealed by specific request and response orders, or
require conflicting modifications to occur within the same millisecond.  My
current test harness doesn't tune the timing of executions and might benefit
from integrating packet dynamics by \citet{pkt-dyn}.

The randomized generator in my framework produced a large number of inputs with
guidance of heuristics based on domain-specific knowledges.  To find bugs with
fewer tests, \citet{judge-cover} applied ideas from combinatorial testing to
tune the generators' distributions for better coverage.  Measuring and improving
coverage of asynchronous tests is yet to be studied.

\end{mainf}

\begin{bibliof}
\bibliography{bibliography}
\end{bibliof}
\end{document}